\documentclass[12pt,psfig,twoside]{report}
\usepackage{graphicx}
\usepackage{amssymb}
\usepackage{epsfig}
\usepackage{amsmath}
\usepackage{amsthm}
\usepackage{setspace}
\usepackage{dsfont}
\usepackage{units}
\newtheorem{theorem}{Theorem}[chapter]
\newtheorem{cor}{Corollary}[chapter]
\newtheorem{lem}{Lemma}[chapter]

\newtheorem{defn}{Definition}[chapter]

\DeclareMathOperator*{\argmax}{arg\,max}

\newcommand{\emptypage}{\newpage \thispagestyle{empty} $\;$ \newpage}
 \usepackage{lscape}

\begin{document} \pagenumbering{roman}
\topmargin 1.5cm \clearpage \thispagestyle{empty}

\newcommand{\jMAtitle}{Dispersion Analysis of Infinite Constellations in Ergodic Fading Channels }
\newcommand{\jAuthor}{Shlomi Vituri}
\newcommand{\jSubmitDate}{March, 2013}

\begin{center}

  \hbox{\vbox to 2.5 cm
    {
      \vfil
      \includegraphics{logo2.ps}
      }
    }
  \vspace{-1.2cm}
THE IBY AND ALADAR FLEISCHMAN
FACULTY OF ENGINEERING\\
THE ZANDMAN-SLANER SCHOOL OF GRADUATE STUDIES\\
THE DEPARTMENT OF ELECTRICAL ENGINEERING - SYSTEMS\\
  \vspace{3.5cm}
\Huge  {\bf   \jMAtitle}\\
  \vspace{4.5cm}
  \large{Thesis submitted toward the degree of \\Master of Science in Electrical and Electronic Engineering}\\
  \vspace{0.4cm}
  by \\
  \vspace{0.4cm}
  \Large {\bf \jAuthor}

  \vspace{0.4cm}

  \vspace{3 cm}

  \Large {\jSubmitDate}
\end{center}

\newpage
\emptypage

\thispagestyle{empty}

\begin{center}

  \hbox{\vbox to 2.5 cm
    {
      \vfil
      \includegraphics{logo2.ps}
      }
    }
   \vspace{-1.2cm}
THE IBY AND ALADAR FLEISCHMAN
FACULTY OF ENGINEERING\\
THE ZANDMAN-SLANER SCHOOL OF GRADUATE STUDIES\\
THE DEPARTMENT OF ELECTRICAL ENGINEERING - SYSTEMS\\
  \vspace{1.3cm}
  \Huge {\bf \jMAtitle}\\
  \vspace{1.3cm}
  \large{Thesis submitted toward the degree of \\Master of Science in Electrical and Electronic Engineering}\\
  \vspace{0.3cm}
  by \\
  \vspace{0.3cm}
  \Large {\bf \jAuthor}

  \vspace{3cm}

  \normalsize

  This research was carried out at the \\Department of Electrical
  Engineering - Systems,\\
  Tel-Aviv University\\

  \vspace{1cm}

  Advisor: \\
  {\bf Prof. Meir Feder}\\
  \vspace{3cm}
  \large {\jSubmitDate}
\end{center}

\newpage
\topmargin 0cm

\onehalfspace

\newpage
\emptypage

\thispagestyle{empty}

\

\vspace{1.5cm}

\begin{center}
\textbf{Acknowledgements}
\end{center}

\vspace{2cm}

I would like to thank my supervisor, Prof. Meir Feder, for his guidance,
trust and support throughout this work.
Working with Meir was both challenging and enjoyable experience.

I would like also to thank my colleagues to the Information Theory Laboratory, and especially to Yair Yona,
for our fruitful discussions.

Finally, I would like to thank my beloved wife, Aya, for her love, trust
and support during my studies.

\newpage
\emptypage

\begin{abstract}
This thesis considers infinite constellations
in fading channels, without power constraint and with perfect
channel state information available at the receiver.
Infinite constellations are the framework, proposed by Poltyrev, for
analyzing coded modulation codes.
The Poltyrev's capacity, is the highest achievable normalized
log density (NLD) of codewords per unit volume, at possibly large block length,
that guarantees a vanishing error probability.
For a given finite block length and a fixed error probability,
there is a gap between the highest achievable NLD and Poltyrev's capacity.
The dispersion analysis quantifies asymptotically this gap.

The thesis begins by the dispersion analysis of infinite constellations
in scalar fading channels.
Later on, we extend the analysis to the case of multiple input multiple output fading channels.
As in other channels, we show that the gap between the highest achievable NLD and the Poltyrev's
capacity, vanishes asymptotically as the square root of the channel dispersion over the block length,
multiplied by the inverse Q-function of the allowed error probability.

Moreover, exact terms for Poltyrev's capacity and channel dispersion,
are derived in the thesis. The relations to the amplitude and to the power constrained fading channels
are also discussed, especially in terms of capacity, channel dispersion and error exponents.
These relations hint that in typical cases the unconstrained model can be interpreted as the limit of the constrained model, when
the signal to noise ratio tends to infinity.
\end{abstract}
\doublespace

\newpage
\emptypage

\tableofcontents \singlespace \listoffigures

\chapter{Introduction}
\pagenumbering{arabic}
Wireless communication channels are traditionally modeled as fading channels, where the transmitted signal is multiplied by a fading process and observed with additive white Gaussian noise (AWGN).
Here we assume that a perfect knowledge of the channel state information (CSI) is available at the receiver.

Classical coding problems over the fading channels often include a peak or an average power restriction of the transmitted signal. Without power constraint the capacity of the channel is not limited, since we can choose an infinite number of codewords to be arbitrarily far apart from each other, and hence get an arbitrarily small error probability and infinite rate. Nevertheless, coded modulation methods ignore the power constraint by designing infinite constellations (IC), and then taking only a subset of codewords which are included in some ``shaping region'' to get a finite constellation (FC) that holds the power constraint. Hence, IC is a very convenient framework for designing such codes.

Poltyrev studied in \cite{Poltyrev} the IC performance over the AWGN channel without power constraint. He defined the density (the average number of codewords per unit volume) and the normalized log density (NLD) of the IC, in analogy to the number of codewords and the communication rate in the power constrained model, respectively. He showed that the highest achievable NLD over the unconstrained AWGN channel, with arbitrarily small error probability, is limited by a maximal NLD, sometimes termed the Poltyrev's capacity. He also derived an exact term for the maximal NLD and error exponent bounds using random coding and sphere packing techniques, for any NLD below the capacity.

In classical channel coding problems, the capacity gives the maximal achievable communication rate when arbitrarily small error probability is required (and arbitrary large codeword length $n$ is permitted). The error exponent provides the exponential rate of convergence (with $n$) in which the error probability goes to zero, for any fixed rate below the capacity. Another interesting question is: for a fixed error probability $\epsilon$ and a fixed codeword length $n$, what is the maximal achievable rate, denoted by $R^*(n,\epsilon)$. Although this question is still unsolved precisely for any finite $n$, the recently revisited dispersion analysis \cite{Polyanskiy} gives the rate of convergence of $R^*(n,\epsilon)$ to the capacity. According to the dispersion analysis, for any fixed $\epsilon$ and finite $n$ the following holds:
\begin{equation}
\label{eq_dispersion_analysis}
R^*(n,\epsilon) = C - \sqrt{\frac{V}{n}}Q^{-1}(\epsilon) + O\left(\frac{\ln(n)}{n}\right),
\end{equation}
where $Q$ is the standard complementary Gaussian CDF, $C$ is the channel capacity and $V$ is the channel dispersion. The channel dispersion is given by the variance of the information density $i(x;y) \triangleq \ln\left(\frac{P(x,y)}{P(x)P(y)}\right)$ for a capacity achieving input distribution. Polyanskiy et al. showed in \cite{Polyanskiy} that \eqref{eq_dispersion_analysis} holds for discrete memoryless channels (DMCs) and for AWGN channel. Note that for AWGN channel $V=\frac{P(P+2)}{2(P+1)^2}$, where $P$ denotes the channel signal to noise ratio (SNR). In \cite{PolyanskiyFading} the result was extended to stationary fading channels, and in \cite{PolyanskiyGE} the dispersion of the Gilbert-Elliot channel was analyzed.

In \cite{Ingber} Ingber et al. showed that in AWGN channel without power constraint and with noise variance $\sigma^2$, the analogy of \eqref{eq_dispersion_analysis} for IC is given by:
\begin{equation}
\label{eq_dispersion_analysis_IC}
\delta^*(n,\epsilon) = \delta^* - \sqrt{\frac{V}{n}}Q^{-1}(\epsilon) + O\left(\frac{\ln(n)}{n}\right),
\end{equation}
where $\delta^*(n,\epsilon)$ is the optimal NLD for fixed $\epsilon$ and finite $n$, and $\delta^* \triangleq \frac{1}{2}\ln\left(\frac{1}{2\pi e\sigma^2}\right)$ is Poltyrev's capacity. For AWGN, the channel dispersion is given by $V = \frac{1}{2}$ (in $\text{nats}^2$ per channel use), which is equal to the limit of the channel dispersion of the power constrained AWGN, when the
SNR tends to infinity.

In this thesis we extend Poltyrev's setting to the case of fading channels with AWGN and CSI at the receiver.
First, we analyze the case of scalar fast fading channels, where the fading process is a series of independent and identically distributed (i.i.d.) random variables (RV's).
This channel is a reasonable model for many practical wireless communication systems, such as systems that communicating over a flat fading channel, or systems that use a (pseudo) random interleaver between the transmitted digital symbols (e.g. BICM techniques) over a frequency selective wireless channel.
Using the dependence testing bound, the sphere packing bound and some normal approximation techniques, we show that an analogous expression to \eqref{eq_dispersion_analysis_IC} holds for fast fading channels.
Later on, using similar but more elaborate tools, we show that \eqref{eq_dispersion_analysis_IC} holds also, in the general case of stationary fading processes, where the channel dispersion is affected by the fading dynamics, but not the Poltyrev's capacity \cite{FadingChannels}\cite{PolyanskiyFading}.
Moreover, in typical fading processes, this dispersion is increased relative to the fast fading channel, with the same marginal fading distribution. This fact can motivate the usage of random interleaver in practical systems with finite block length, in order to get effectively a fast fading channel, with smaller channel dispersion.

In this thesis we also analyze the dispersion of multiple input multiple output (MIMO) fast fading channels without power constraint.
It is well known that the usage of multiple antennas in wireless communication is very beneficial.
This usage increases the number of the degrees of freedom available by the channel, which is expressed immediately by an increasing channel capacity.
This increase is also called the ``multiplexing gain'' of the channel.
In \cite{Telatar}\cite{Foschini} the capacity of the ergodic power constrained MIMO channel with $t$ transmit and $r$ receive antennas was obtained, where the gains between the transmitting-receiving antenna pairs are i.i.d. Rayleigh faded RV's. Moreover, in \cite{Foschini} it was shown that in the high SNR regime the multiplexing gain equals to the number of available degrees of freedom, i.e. the minimum between $t$ and $r$.
Note that there are also communication techniques that allow to increase the reliability of the transmitted signal at the cost of a reduced multiplexing gain. The increasing of the reliability by the usage of multiple antennas is also called diversity.
The fundamental tradeoff between diversity and multiplexing was derived in \cite{Tse}, in the case of non-ergodic Rayleigh fading MIMO channels with power constraint. This result was extended to the case of IC's over the same MIMO channel, but without power constraint in \cite{Yona}.

Here in the thesis, we focus on the case of the ergodic fast fading MIMO channels without power constraint. Moreover, we assume that the gains between the transmitting-receiving antenna pairs are i.i.d. Rayleigh faded RV's, which are available at the receiver. By similar techniques as in the scalar channel, we derive the dispersion and the Poltyrev's capacity of this MIMO channel under the constraint of \emph{Full Dimensional Transmission} (\emph{FDT}). This constraint means that all of the transmission dimensions are in use during the transmission.
Later on, we compare the $t \times t$ MIMO setting to the setting of $t$ parallel, identical and independent scalar fast fading channels with Rayleigh fading distribution. This comparison promise lower channel dispersion and greater Poltyrev's capacity in the MIMO setting relative to the parallel channels setting, due to the dependency between the received signals. Finally, we discuss the general case of MIMO dispersion analysis without any constraint. This discussion reveals a very surprising phenomena of Poltyrev's capacities in MIMO fading channels: In contrast to the capacity of FC's over MIMO fading channels, reducing the IC's transmission dimension can increase the Poltyrev's capacity of the channel.

The relations to the amplitude and to the power constrained fading channels are also discussed in the thesis, especially, in terms of capacity, channel dispersion and error exponents. These relations hint that in most cases, including single input single output (SISO) and \emph{FDT} MIMO the unconstrained model can be interpreted as the limit of the constrained model, when the SNR tends to infinity.

The thesis is arranged as follows: In Chapters 2 and 3 the basic definitions are formulated, and
previous results are surveyed.
In Chapter 4 the dispersion of infinite constellations in scalar fading channels is analyzed. This chapter starts with the analysis of IC's over fast fading channels, which is extended later on to the special cases of lattices and general fading channels with memory.
In Chapter 5 the dispersion analysis of infinite constellations in MIMO fading channels and its relation to the independent parallel channels is analyzed.
Conclusions, discussion and further research follow in Chapter 6.

\chapter{Basic Definitions}
\label{chapter:BasicDefinitions}
In this chapter we review the notations and the basic definitions of this thesis. Section \ref{sec_channel_model} presents and defines the scalar and the MIMO fading channels, and Section \ref{sec_infinite_constellations} extends the Poltyrev's setting of infinite constellations without power constraint to these channels. Finally, the most important quantity that is analyzed in this thesis, the channel dispersion, is defined and reviewed in Section \ref{sec_channel_dispersion}.
\section{Notation}
Vectors are denoted by bold-face lower case letters, e.g. $\mathbf{x}$ and $\mathbf{y}$. Matrices are denoted by bold-face capital letters, e.g. $\mathbf{H}$. Components of random vector $\mathbf{x}$ are denoted by capital letters, $X_1, X_2, \dots, X_n$. In the same manner, components of a random matrix $\mathbf{H}$ are denoted by $\left\{H_{ij}\right\}$.
Concatenation of $n$ consecutive vectors is denoted by $\mathbf{x}^n = (\mathbf{x}_1^{\dagger},\dots,\mathbf{x}_n^{\dagger})^{\dagger}$, and a concatenation of $n$ consecutive matrices to a block diagonal matrix is denoted by $\mathbf{H}^n = \text{diag}(\mathbf{H}_1,\dots,\mathbf{H}_n)$.
Instances of random variables are denoted by lower case letters, e.g. $x, ~y$ and $h$.
\section{Channel Model}
\label{sec_channel_model}
\subsection{Scalar Channel Model}
The scalar real fading channel model is given by
\begin{equation}
\label{eq_channel_model}
Y_i = H_i\cdot X_i + Z_i, ~i = 1,2,\dots
\end{equation}
where,
\begin{itemize}
  \item $\{X_i\}$ is a series of channel inputs,
  \item $\{H_i\}$ is a series of fading coefficients satisfying $E\{H_i^2\}=1$,
  \item $\{Z_i\}$ is a series of i.i.d. normal random variables, such that $Z_i\sim N(0,\sigma^2)$,
  \item $\{Y_i\}$ is a series of channel outputs.
\end{itemize}
The series $\{X_i\}, ~\{H_i\}$ and $\{Z_i\}$ are independent of each other. In vector notation (for finite $n$) the channel model is given by:
\begin{equation}
\label{eq_vectoric_channel_model}
\mathbf{y} = \mathbf{H} \cdot \mathbf{x} + \mathbf{z},
\end{equation}
where $\mathbf{H} \triangleq \text{diag}\left(H_1, H_2, \dots, H_n\right)$. We assume a perfect CSI available at the receiver, and hence the receiver's channel output is the couple $\left(\mathbf{y},\mathbf{H}\right)$.

The first fading process that we will analyze in the thesis, is the fast fading process. In fast fading, we mean that all the fading coefficients are i.i.d. RVs. Later on, we will extend the analysis to the more general case of stationary fading processes.

Without loss of generality, since we have a perfect CSI at the receiver, we can assume that the fading coefficients are nonnegative. Moreover, we restrict the marginal fading distribution to probability density functions (PDF) with zero probability to equal zero. We will denote such a fading distribution by \emph{regular fading distribution}, which is defined formally below.
\begin{defn}
(Regular fading distribution): A fading PDF $f\left(h\right)$ is called regular fading distribution if there exists some positive constant $\alpha$, s.t. $f\left(h\right) \propto \frac{1}{h^{1-\alpha}}$ for small enough $h > 0$.
\end{defn}
A popular statistical model for the fading channel is the Nakagami-$m$ distribution.
This popular family of fading distributions are given by:
\begin{equation}
f_m(h) = \frac{2m^m}{\Gamma(m)}h^{2m-1}e^{-mh^2},~h\geq0,~m\geq\frac{1}{2}.
\end{equation}
It is easy to verify that this distribution is a \emph{regular fading distribution} for all $m \geq \frac{1}{2}$.
\subsection{MIMO Channel Model}
The basic MIMO channel model is given by the following equation:
\begin{equation}
\label{eq_MIMO_model}
\mathbf{y}_i = \mathbf{H}_i\cdot\mathbf{x}_i + \mathbf{z}_i,~i=1,2,\dots
\end{equation}
where $\mathbf{x} \in \mathbb{C}^{t}, \mathbf{y}, \mathbf{z} \in \mathbb{C}^{r}$, $\mathbf{H} = \{H_{ij}\} \in \mathbb{C}^{r \times t}$, $H_{ij}$ are circular symmetric i.i.d. $CN(0,1)$ and $\mathbf{z} \sim CN(0, \sigma^2 I_r)$, where the subscripts are removed for simplicity of presentation.

The following extended channel model:
\begin{equation}
\mathbf{y}^n = \mathbf{H}^n\cdot\mathbf{x}^n + \mathbf{z}^n
\end{equation}
is getting by the concatenation of $n$ consecutive channel uses. We assume fast fading model, namely, $\{\mathbf{H}_i\}_{i=1}^{n}$ is a set of i.i.d. matrices.

Note that by the singular value decomposition (SVD) theorem (e.g. \cite{Telatar}), any matrix $\mathbf{H} \in \mathbb{C}^{r \times t}$ can be written as
\begin{equation}
\mathbf{H} = \mathbf{U}\mathbf{D}\mathbf{V}^{\dagger}
\end{equation}
where $\mathbf{U} \in \mathbb{C}^{r \times r}$ and $\mathbf{V} \in \mathbb{C}^{t \times t}$ are unitary, and $\mathbf{D} \in \mathbb{R}^{r \times t}$ is non-negative and diagonal. Moreover, the diagonal entries of $\mathbf{D}$ are equal to the square root of the eigenvalues of $\mathbf{H}\mathbf{H}^{\dagger}$.
Using it, an equivalent model to \eqref{eq_MIMO_model} can be written as
\begin{equation}
\tilde{\mathbf{y}} = \mathbf{D}\mathbf{V}^{\dagger}\mathbf{x} + \tilde{\mathbf{z}}
\end{equation}
where $\tilde{\mathbf{y}} \triangleq \mathbf{U}^{\dagger}\mathbf{y}$ and $\tilde{\mathbf{z}} \triangleq \mathbf{U}^{\dagger}\mathbf{z}$. Note that given the CSI the distributions of $\tilde{\mathbf{z}}$ and $\mathbf{z}$ are the same.

In the thesis we show results for the case where $t \leq r$. The case where $t>r$ is still under consideration.
Since $\mathbf{D}$ is of rank at most $t$ for $t \leq r$, we can define the following equivalent model
\begin{equation}
\label{eq_equiv_MIMO_model}
\mathbf{y}' = \mathbf{D}'\mathbf{V}^{\dagger}\mathbf{x} + \mathbf{z}'
\end{equation}
where $\mathbf{y}'$ and $\mathbf{z}'$ are equal the first $t$ entries of $\tilde{\mathbf{y}}$ and $\tilde{\mathbf{z}}$, respectively. The matrix $\mathbf{D}'$ is a $t \times t$ diagonal matrix, whose $t$ diagonal entries are equal to the first $t$ diagonal entries of $\mathbf{D}$. In the analysis of the MIMO channel, we will use the simplified equivalent MIMO model \eqref{eq_equiv_MIMO_model}.
\section{Infinite Constellations}
\label{sec_infinite_constellations}
\subsection{Infinite Constellations in scalar real fading channels}
An infinite constellation of dimension $n$ is any countable set of points $S = \left\{s_1,s_2,\dots\right\}$ in $\mathbb{R}^n$.
Let $\text{Cb}(a)$ denote an $n$ dimensional hypercube in $\mathbb{R}^n$:
\begin{equation}
\text{Cb}(a) \triangleq \left\{\mathbf{x}\in \mathbb{R}^n ~ s.t. ~ \forall_i \left|x_i\right| < \frac{a}{2} \right\}.
\end{equation}
We denote by $M\left(S,a\right) = \left|S\bigcap\text{Cb}(a)\right|$ the number of points in the intersection of $\text{Cb}(a)$ and $S$.
The density of points per unit volume of $S$ is denoted by $\gamma$ and defined by
\begin{equation}
\gamma \triangleq \limsup_{a \rightarrow \infty} \frac{M\left(S,a\right)}{a^n}.
\end{equation}
The normalized log density of $S$ is denoted by $\delta$ and defined by
\begin{equation}
\delta \triangleq \frac{1}{n}\ln\left(\gamma\right).
\end{equation}
In the receiver, given the channel state information (i.e. given $\mathbf{H}$), the receiver's IC, denoted by $S_\mathbf{H}$, is defined by
\begin{equation}
S_\mathbf{H} \triangleq \left\{s_{\text{rc}}: s_{\text{rc}} = \mathbf{H}\cdot s, s \in S\right\}.
\end{equation}
We also define the set ${\mathbf{H}}\cdot\text{Cb}(a)$ as the multiplication of each point in $\text{Cb}(a)$ with the matrix $\mathbf{H}$.
The density of $S_\mathbf{H}$ is defined by
\begin{align}
\gamma_{\text{rc}}\left(\mathbf{H}\right) &\triangleq \limsup_{a \rightarrow \infty} \frac{M\left(S_{\mathbf{H}},a\right)}{\text{Vol}\left(\mathbf{H}\cdot\text{Cb}(a)\right)}
\\ &= \limsup_{a \rightarrow \infty} \frac{M\left(S,a\right)}{\det{\left(\mathbf{H}\right)} \cdot a^n}
\\ &= \frac{\gamma}{\det{\left(\mathbf{H}\right)}}
\end{align}
where $M\left(S_{\mathbf{H}},a\right) \triangleq \left|S_{\mathbf{H}}\bigcap{\mathbf{H}}\cdot\text{Cb}(a)\right|$.
For $s_{\text{rc}} \in S_{\mathbf{H}}$, let $P_e\left(s_{\text{rc}}|\mathbf{H}\right)$ denote the error probability when $s$, such that $s_{\text{rc}} = \mathbf{H} \cdot s$, was transmitted and the CSI at the receiver is $\mathbf{H}$. Then, using maximum likelihood (ML) decoding the error probability is given by
\begin{equation}
P_e\left(s_{\text{rc}}|\mathbf{H}\right) = Pr\left\{s_{\text{rc}} + \mathbf{z} \notin W\left(s_{\text{rc}}\right) | \mathbf{H} \right\}
\end{equation}
where $W\left(s_{\text{rc}}\right)$ is the Voronoi cell of $s_{\text{rc}}$, i.e. the convex polytope of the points that are closer to $s_{\text{rc}}$ than to any other point $s'_{\text{rc}} \in S_{\mathbf{H}}$.
\begin{defn}
(Conditional expectation over a faded hypercube): For any function $f:S_{\mathbf{H}}\rightarrow\mathbb{R}$, the conditional expectation of $f(s_{\emph{\text{rc}}})$ given $\mathbf{H}$, where $s_{\emph{\text{rc}}}$ is drawn uniformly from the code points that reside in the faded hypercube ${\mathbf{H}}\cdot\emph{\text{Cb}}(a)$, will be denoted and defined by
\begin{equation}
E_{S,a|\mathbf{H}}\left\{ f(s_{\emph{\text{rc}}}) \right\} \triangleq \frac{1}{M\left(S_{\mathbf{H}},a\right)} \sum_{s_{\emph{\text{rc}}} \in S_{\mathbf{H}}\bigcap{\mathbf{H}}\cdot\emph{\text{Cb}}(a)}f(s_{\emph{\text{rc}}}).
\end{equation}
\end{defn}
The average error probability using ML decoding and equiprobable messages transmission is given by
\begin{align}
P_e\left(S\right) = E\left\{ P_e\left(S_{\mathbf{H}}\right) \right\} &\triangleq E\left\{ \limsup_{a \rightarrow \infty} \frac{1}{M\left(S_{\mathbf{H}},a\right)} \sum_{s_{\text{rc}} \in S_{\mathbf{H}}\bigcap{\mathbf{H}}\cdot\text{Cb}(a)} P_e\left(s_{\text{rc}}|\mathbf{H}\right) \right\}
\\                &\triangleq E\left\{ \limsup_{a \rightarrow \infty} E_{S,a|\mathbf{H}}\left\{ P_e\left(s_{\text{rc}}|\mathbf{H}\right) \right\} \right\}.
\end{align}
\subsection{Infinite Constellations in complex MIMO fading channels}
Here we extend, briefly, the setting of infinite constellations in real scalar fading channels, to the general case of complex MIMO fading channels.

An infinite constellation of complex dimension $l$ is any countable set of points $S = \left\{s_1,s_2,\dots\right\}$ in $\mathbb{C}^l$.
Let $\text{Cb}(a,l)$ denote an $l$ complex dimensional hypercube in $\mathbb{C}^l$:
\begin{equation}
\text{Cb}(a,l) \triangleq \left\{\mathbf{x}\in \mathbb{C}^l ~ s.t. ~ \forall_i \left|\mathfrak{Re}(x_i)\right|,\left|\mathfrak{Im}(x_i)\right| < \frac{a}{2} \right\}. \nonumber
\end{equation}
The density of points per unit volume of $S$ is defined by
\begin{equation}
\gamma \triangleq \limsup_{a \rightarrow \infty} \frac{M\left(S,a\right)}{a^{2l}}. \nonumber
\end{equation}
The normalized log density of $S$, using $n$ channel uses, where $l=nt$, is defined by
\begin{equation}
\delta \triangleq \frac{1}{n}\ln\left(\gamma\right). \nonumber
\end{equation}
In the receiver, given the CSI, the receiver's IC, denoted by $S_{\mathbf{H}^n}$, is defined by
\begin{equation}
S_{\mathbf{H}^n} \triangleq \left\{s_{\text{rc}}: s_{\text{rc}} = \mathbf{H}^n\cdot s, s \in S\right\}. \nonumber
\end{equation}
The density of $S_{\mathbf{H}^n}$ is defined by
\begin{align}
\begin{aligned}
\gamma_{\text{rc}} &\triangleq \limsup_{a \rightarrow \infty} \frac{M\left(S_{\mathbf{H}^n},a\right)}{\text{Vol}\left(\mathbf{H}^n\cdot\text{Cb}(a,l)\right)}
\\ &= \limsup_{a \rightarrow \infty} \frac{M\left(S,a\right)}{\det{\left(\mathbf{H}^{n\dagger}\mathbf{H}^n\right)} \cdot a^{2l}}
\\ &= \frac{\gamma}{\prod_{i=1}^{n}\det{(\mathbf{H}_i^{\dagger}\mathbf{H}_i)}}.
\end{aligned} \nonumber
\end{align}
For $s_{\text{rc}} \in S_{\mathbf{H}^n}$, let $P_e\left(s_{\text{rc}}|\mathbf{H}^n\right)$ denote the error probability when $s$, such that $s_{\text{rc}} = \mathbf{H}^n \cdot s$, was transmitted and the CSI at the receiver is $\mathbf{H}^n$. Then, using maximum likelihood (ML) decoding the error probability is given by
\begin{equation}
P_e\left(s_{\text{rc}}|\mathbf{H}^n\right) = Pr\left\{s_{\text{rc}} + \mathbf{z}^n \notin W\left(s_{\text{rc}}\right) |\mathbf{H}^n \right\}, \nonumber
\end{equation}
where $W\left(s_{\text{rc}}\right)$ is the Voronoi cell of $s_{\text{rc}}$.
The average error probability using ML decoding and equiprobable messages transmission is given by
\begin{align}
P_e\left(S\right) \triangleq E\left\{ \limsup_{a \rightarrow \infty}
\frac{1}{M\left(S_{\mathbf{H}^n},a\right)} \hspace{-0.2cm} \sum_{s_{{\text{rc}}} \in S_{\mathbf{H}^n}\bigcap{\mathbf{H}^n}\cdot{\text{Cb}}(a,l)} \hspace{-0.5cm} P_e\left(s_{\text{rc}}|\mathbf{H}^n\right) \right\}. \nonumber
\end{align}
\section{Channel Dispersion}
\label{sec_channel_dispersion}
The channel capacity is the highest achievable rate, at possibly large block length, that guarantees a vanishing error probability, when communicating over a channel. In the setting of a given fixed error probability $\epsilon$, and a finite block length $n$, there is a gap between the highest achievable rate, denoted by $R^*(n,\epsilon)$, and the capacity. The asymptotically convergence rate of this gap, when the block length tends to infinity, is given by the channel dispersion.

Formally, the operational channel dispersion was defined in \cite{Polyanskiy} as follows:
\begin{equation}
V = \lim_{\epsilon \to 0}\limsup_{n \to \infty}n\cdot\left(\frac{C-R^*(n,\epsilon)}{Q^{-1}(\epsilon)}\right)^2,
\end{equation}
where $C$ is the channel capacity.

The dispersion of DMCs, Gilbert-Elliot channel and the power constrained AWGN and fading channels, were analyzed in \cite{Strassen}\cite{Polyanskiy}\cite{PolyanskiyGE}\cite{PolyanskiyFading}. Moreover, it was shown that the operational channel dispersion equals to the information theoretic channel dispersion, which is given by:
\begin{equation}
V = Var(i(X;Y)),
\end{equation}
where $i(x;y)$ is the information density, which is given by
\begin{equation}
i(x;y) = \ln\left(\frac{P(x,y)}{P(x)P(y)}\right),
\end{equation}
for a capacity achieving input distribution that also minimizes $V$.

Inspired by the dispersion analysis of infinite constellations over the unconstrained AWGN channel in \cite{Ingber}, let define the operational channel dispersion of infinite constellations over the unconstrained fading channel, as follows:
\begin{equation}
V = \lim_{\epsilon \to 0}\limsup_{n \to \infty}n\cdot\left(\frac{\delta^*-\delta^*(n,\epsilon)}{Q^{-1}(\epsilon)}\right)^2,
\end{equation}
where $\delta^*$ is the Poltyrev's capacity and $\delta^*(n,\epsilon)$ is the highest achievable NLD in the setting of fixed error probability $\epsilon$, and finite block length $n$.
In this thesis, we will analyze the dispersion of infinite constellations over the unconstrained scalar and MIMO fading channels.

\chapter{Previous Results}
\label{chapter:PreviousResults}
This chapter reviews existing results in fields relevant to the research of this thesis. In Sections \ref{sec_dispersion_of_power_constrained_fading_channels} and \ref{sec_dispersion_of_infinite_constellations_in_AWGN_channel} the dispersion analysis of channels with and without power constraint are presented, respectively. Finally, related results in MIMO fading channels are presented in Section \ref{sec_related_results_in_MIMO_fading_channels}.
\section{Dispersion of power constrained fading channels}
\label{sec_dispersion_of_power_constrained_fading_channels}
In \cite{PolyanskiyFading} Polyanskiy et al. analyzed the channel dispersion of power constrained stationary fading processes, with perfect channel knowledge at the receiver. The main result of this paper is given by the following theorem.
\begin{theorem}[Polyanskiy et al. \cite{PolyanskiyFading}]
Assume that the stationary process $H_1,H_2,\dots$ satisfies the following assumptions:
\begin{enumerate}
  \item $E\{H_i^2\}=1$
  \item $H_1,H_2,\dots$ is a strong mixing\footnote{The \emph{strong mixing} stationary process will be defined rigorously later on, in section \ref{sec_fading_channels_with_memory}.} process such that for some $r<1$:
  \begin{equation}
  \sum_{k=1}^{\infty}k(\alpha_H(k))^r < \infty.
  \end{equation}
  \item For all $j>1$ we have
  \begin{equation}
  Pr\{H_{j+1}H_1\neq0\}>0.
  \end{equation}
\end{enumerate}
Then, as $n$ grows, for any $0<\epsilon<\frac{1}{2}$, the highest achievable rate $R^*(n,\epsilon)$, is given by:
\begin{equation}
R^*(n,\epsilon) = C - \sqrt{\frac{V}{n}}Q^{-1}(\epsilon) + o\left(\frac{1}{\sqrt{n}}\right),
\end{equation}
where,
\begin{enumerate}
  \item $C(H) \triangleq \frac{1}{2}\ln(1+H^2 \cdot SNR)$,
  \item $C = E\{C(H)\}$,
  \item $V = Var\left(C(H)\right) + 2\sum_{k=1}^{\infty}R_{C(H)}(k) + \frac{1}{2}\left(1 - E^2\left\{\frac{1}{1+H^2 \cdot SNR}\right\} \right)$,
  \item $R_{C(H)}(k)$ is the auto-correlation function of the process $C(H_1),C(H_2),\dots$
\end{enumerate}
regardless whether $\epsilon$ is maximal or average error probability.
\end{theorem}
In addition, in \cite{Polyanskiy}, a similar dispersion analysis was derived to DMCs and power constrained AWGN channels. In \cite{PolyanskiyGE} the dispersion of the Gilbert-Elliot channel was also derived.
\section{Dispersion of IC's in the AWGN channel}
\label{sec_dispersion_of_infinite_constellations_in_AWGN_channel}
In \cite{Poltyrev} Poltyrev studied the performance of IC's over the unconstrained AWGN channel. He showed that the highest achievable NLD, at possibly large block length, that guarantees a vanishing error probability, namely the Poltyrev's capacity, is given by:
\begin{equation}
\delta^* = \frac{1}{2}\ln\left(\frac{1}{2 \pi e \sigma^2}\right),
\end{equation}
where, $\sigma^2$ is the noise variance of the AWGN.
He also derived the asymptotic optimal error probability (with $n$), for any fixed $\delta \leq \delta^*$, in the manner of the error exponent.
The lower and upper bounds of the error exponent, were given by the random coding error exponent, and by the spherical bound exponent, respectively.

In \cite{Ingber} Ingber et al. derived a more tighter asymptotic analysis for the optimal error probability, for any fixed $\delta \leq \delta^*$. In addition, the asymptotic analysis for a fixed error probability, was also derived. This analysis, which is actually the dispersion analysis, is given by the following theorem.
\begin{theorem}[Ingber et al. \cite{Ingber}]
Let $\epsilon>0$ be a given, fixed, error probability. Denote by $\delta^*(n,\epsilon)$ the highest NLD for which there exists an $n$-dimensional infinite constellation with error probability at most $\epsilon$. Then, as $n$ grows,
\begin{equation}
\delta^*(n,\epsilon) = \delta^* - \sqrt{\frac{1}{2n}}Q^{-1}(\epsilon) + \frac{1}{2n}\ln(n) + O\left(\frac{1}{n}\right).
\end{equation}
\end{theorem}
\section{Related results in MIMO fading channels}
\label{sec_related_results_in_MIMO_fading_channels}
\subsection{Capacity and Error Exponent}
\label{sec_capacity_and_error_exponent}
In \cite{Telatar} the capacity and the random coding error exponent, of the ergodic power constrained MIMO Rayleigh fading channel, with $t$ transmit and $r$ receive antennas, and perfect channel knowledge at the receiver, were analyzed.

It was shown that the capacity is given by the following expression:
\begin{equation}
C = E\left\{ \ln\left(\det\left(I_t + SNR\cdot\mathbf{H}^{\dagger}\mathbf{H}\right)\right) \right\},
\end{equation}
which a simple numerical calculation of it, can be done, by using the following theorem.
\begin{theorem}[Telatar \cite{Telatar}]
\label{thm_Telatar_capacity}
The capacity of the power constrained channel with $t$ transmitters and $r$ receivers equals
\begin{equation}
C = \int_0^{\infty}\ln(1 + SNR\cdot\lambda)\sum_{k=0}^{m-1}\frac{k!}{(k+l-m)!}\left[L_k^{l-m}(\lambda)\right]^2\lambda^{l-m}e^{-\lambda}d\lambda,
\end{equation}
where $m = \min(r,t)$, $l=\max(r,t)$ and $L_j^i(x) = \frac{1}{j!}e^xx^{-i}\frac{d^j}{dx^j}(e^{-x}x^{i+j})$ are the Laguerre polynomials.
\end{theorem}

Moreover, it was shown that the random coding error exponent of the channel, is given by:
\begin{equation}
E_r(R) = \max_{0\leq\rho\leq1}-\ln E\left\{ \det\left( I_t + \frac{SNR}{1+\rho}\cdot\mathbf{H}^{\dagger}\mathbf{H} \right)^{-\rho} \right\}-\rho R.
\end{equation}
Note, that this is not the optimal random coding error exponent, since it was derived by using the suboptimal Gaussian input distribution. Although, the choice of uniform input distribution on a ``thin spherical shell'' will give better results as in \cite{Gallager}, the Gaussian input distribution leads to simpler expressions, and also gives an upper bound on the error probability.

Finally, in \cite{MIMOErExp} the Gallager's error exponent for MIMO block fading channels with spatial correlation, can also be found.
\subsection{Moments of the Mutual Information}
\label{sec_moments_of_the_information_density}
In \cite{Oyman}, Oyman et al. analyzed the ergodic power constrained MIMO Rayleigh fading channel, with $t$ transmit and $r$ receive antennas, and perfect channel knowledge at the receiver. For Gaussian input distribution, the mutual information given the CSI is given by $I(\mathbf{H})=\ln\left(\det\left(I_t + SNR\cdot\mathbf{H}^{\dagger}\mathbf{H}\right)\right)$. Using it, they derived analytical closed-form approximations for the capacity (the expectation of the mutual information with Gaussian input distribution), and for the variance of the mutual information, at the high SNR regime. These approximations are given by the following:
\begin{align}
C &= E\{I(\mathbf{H})\} \approx m\ln(SNR) - \gamma m + \sum_{j=1}^{m}\sum_{p=1}^{l-j}\frac{1}{p},
\\
V_I &= Var(I(\mathbf{H})) \approx \sum_{j=1}^{m}\sum_{p=1}^{\infty}\frac{1}{(p+l-j)^2},
\end{align}
where $m = \min(r,t)$, $l=\max(r,t)$ and $\gamma = 0.577\dots$ is the Euler's constant.
\subsection{The Non-Ergodic Model}
In the setting of infinite constellations over the unconstrained MIMO Rayleigh channel, only the case of non-ergodic channel was analyzed.
In the non-ergodic channel it is assumed that the block length is much smaller than the channel coherence time. In other words, the channel fading matrix remains constant throughout all the codeword transmission.
It is a well known fact that the usage of multiple antennas in wireless communication is very beneficial.
On one hand, this usage increases the number of the degrees of freedom available by the channel, which allows to increase the transmission rate, i.e. increasing the multiplexing gain.
On the other hand, other communication techniques allow to increase the reliability of the transmitted signal, i.e. increasing the diversity order.
A trivial example for such a technique is the transmission of the same information on different paths of transmitting-receiving antenna pairs in the price of the multiplexing gain.
In \cite{Yona} Yona et al. derived the DMT (\emph{Diversity and Multiplexing Tradeoff}) for IC's, as Zheng et al. derived in \cite{Tse}, for the power constrained setting. Namely, for each multiplexing gain they found the maximal diversity order that can be achieved.

\chapter{Dispersion of Infinite Constellations in Fast Fading Channels}
\label{chapter:DispersionOfICInFastFadingChannels}
\renewcommand{\thefootnote}{} 
\footnotetext[1]{The material in this chapter was partially presented in \cite{Vituri} and \cite{Allerton}.}
\renewcommand{\thefootnote}{\arabic{footnote}} 
In this chapter we analyze the dispersion of infinite constellation in scalar real fast fading channels without power constraint. In Section \ref{sec_main_result} we present our main result, whose converse and direct parts are proven in Sections \ref{sec_converse_part} and \ref{sec_direct_part}, respectively. Later on, in Section \ref{sec_extension_to_the_complex_model} we extend our main result to the case of scalar complex fading channels, and in Section \ref{sec_VNR} we present our main result in terms of the VNR (Volume to Noise Ratio). Relation to the power constrained fading channel and comparison to the unconstrained AWGN channel are discussed in Sections \ref{sec_relation_to_the_power_constrained_model} and \ref{sec_comparison_to_the_awgn_channel}, respectively. Finally, in Section \ref{sec_fading_channels_with_memory} we extend the dispersion analysis to the general case of stationary fading channels with memory.
\section{Main Result}
\label{sec_main_result}
\begin{theorem}
\label{thm_main_result}
Let $\epsilon>0$ be a given, fixed, error probability. Denote by $\delta^*(n,\epsilon)$ the optimal NLD for which there exists an $n$-dimensional infinite constellation with average error probability at most $\epsilon$. Then, for any regular fading distribution of $H$, as $n$ grows,
\begin{equation}
\label{eq_main_result}
\delta^*(n,\epsilon) = \delta^* - \sqrt{\frac{V}{n}}Q^{-1}(\epsilon) + O\left(\frac{\ln(n)}{n}\right),
\end{equation}
where,
\begin{align}
\delta^* &\triangleq E\left\{ \delta(H) \right\} = E\left\{ \frac{1}{2}\ln\left(\frac{H^2}{2\pi e\sigma^2}\right) \right\}
\\ V     &\triangleq \frac{1}{2} + Var(\delta(H)) = \frac{1}{2} + Var\left(\frac{1}{2}\ln(H^2)\right)
\end{align}
noting that
\begin{equation}
\delta(H) \triangleq \frac{1}{2}\ln\left(\frac{H^2}{2\pi e\sigma^2}\right).
\end{equation}
\end{theorem}
The converse and the direct parts of the proof of this theorem are given in Sections \ref{sec_converse_part} and \ref{sec_direct_part}, respectively.
\begin{cor}
The highest achievable NLD with arbitrary small error probability, namely the Poltyrev's capacity, over the unconstrained fast fading channel with available CSI at the receiver, is given by
\begin{equation}
\delta^* \triangleq E\left\{ \frac{1}{2}\ln\left(\frac{H^2}{2\pi e\sigma^2}\right) \right\}.
\end{equation}
\end{cor}
\begin{proof}
By taking the limit $n\rightarrow\infty$ in \eqref{eq_main_result} we get the desired result (for any $0 < \epsilon < 1$).
\end{proof}
\section{Converse Part}
\label{sec_converse_part}
In this section we prove the converse part of Theorem \ref{thm_main_result}. The converse part is based on normal approximation of the \emph{sphere packing lower bound} on the average error probability.
The sphere packing lower bound of IC's over fading channels is presented in Section \ref{sec_sphere_packing_bound}, and in Section \ref{sec_proof_of_converse_part} we complete the proof by a derivation of an appropriate normal approximation technique.
\subsection{The Sphere Packing Bound}
\label{sec_sphere_packing_bound}
In this section we prove the following \emph{sphere packing bound} for any IC $S$ with NLD $\delta$.
\begin{theorem}
\label{thm_sphere_packing_lower_bound}
For any IC $S$ with NLD $\delta$, the average error probability is lower bounded by the following sphere packing bound:
\begin{equation}
\label{eq_general_SPB}
P_e\left(S\right) \geq P_e^{\text{SB}}\left(\delta\right) \triangleq Pr\left\{\left\|\mathbf{z}\right\|^2 \geq e^{-2\delta}\left(\frac{\det(\mathbf{H})}{V_n}\right)^{\frac{2}{n}} \right\}.
\end{equation}
\end{theorem}
The proof will be done in stages, first for the case of IC's where all the Voronoi cells have equal volume (e.g. lattices), then for the case of IC's with bounded  Voronoi cells' volume and finally for the general case of any IC.

In the case where all the Voronoi cells have equal volume $V_{\text{tr}}$,
in the receiver given the CSI $\mathbf{H}$, we get an IC with Voronoi cell volume that equals $V_{\text{rc}} = V_{\text{tr}}\cdot\det(\mathbf{H}) = V_{\text{tr}}\Pi_{i=1}^{n}{H_i}$.
By the \emph{equivalent sphere} argument \cite{Poltyrev}\cite{Tarokh}, the probability that the noise leaves the Voronoi cell in the receiver is lower bounded by the probability to leave a sphere of the same volume:
\begin{equation}
\label{eq_SPB}
P_e\left(S\right) \geq Pr\left\{\left\|\mathbf{z}\right\|^2 \geq r_{\text{eff}}^{2}(\mathbf{H}) \right\},
\end{equation}
where
\begin{equation}
\label{eq_equivalent_sphere}
V_nr_{\text{eff}}^{n}(\mathbf{H}) \triangleq V_{\text{rc}}
\end{equation}
and
\begin{equation}
\label{eq_sphere_vol_coef}
V_n = \frac{\pi^{n/2}}{\frac{n}{2}\Gamma\left(\frac{n}{2}\right)}.
\end{equation}
Combining \eqref{eq_SPB}, \eqref{eq_equivalent_sphere}, \eqref{eq_sphere_vol_coef} with the definition of $\delta = -\frac{1}{n}\ln(V_{\text{tr}})$ we get:
\begin{equation}
P_e\left(S\right) \geq Pr\left\{\left\|\mathbf{z}\right\|^2 \geq e^{-2\delta}\left(\frac{\det(\mathbf{H})}{V_n}\right)^{\frac{2}{n}} \right\} \triangleq P_e^{\text{SB}}\left(\delta\right).
\end{equation}
Now let us extend the correctness of the bound to any IC with bounded Voronoi cells' volume (\emph{regular IC's}).
\begin{defn}
(Regular IC's): An IC S is called regular if there exists a radius $r_0 > 0$, s.t. for all $s \in S$, the Voronoi cell $W(s)$ is contained in $\emph{\text{Ball}}(s,r_0) \triangleq \left\{ \mathbf{x} \in \mathbb{R}^n ~ s.t. ~\|\mathbf{x} - s\| < r_0 \right\}$.
\end{defn}
For $s \in S$, denote by $v(s)$ the volume of the Voronoi cell of $s$, and denote by $V\left(S\right)$ the average Voronoi cell volume of $S$. Then, by definition
\begin{equation}
V\left(S\right) \triangleq \liminf_{a \rightarrow \infty} E_{S,a}\left\{v(s)\right\} = \liminf_{a \rightarrow \infty} \frac{1}{M(S,a)} \sum_{s \in S\bigcap\text{Cb}(a)}{v(s)}.
\end{equation}
It is easy to verify that for any regular IC, the density is given by $\gamma = \frac{1}{V\left(S\right)}$.

Clearly, for any given $\mathbf{H}$, the receiver IC is also regular. Hence, in the same manner, we can define the receiver's average Voronoi cell volume of $S_{\mathbf{H}}$ by
\begin{equation}
V\left(S_{\mathbf{H}}\right) \triangleq \liminf_{a \rightarrow \infty} E_{S,a|\mathbf{H}}\left\{v(s_{\text{rc}})\right\}.
\end{equation}
The density at the receiver is given by $\gamma_{\text{rc}} = \frac{1}{V\left(S_{\mathbf{H}}\right)} = \frac{\gamma}{\det\left(\mathbf{H}\right)}$.

To prove the sphere bound for regular IC's it is desirable for the clarity of the proof to denote by $\text{SPB}\left(v|\mathbf{H}\right)$, the probability that the noise vector $\mathbf{z}$ leaves a sphere of volume $v$ given the CSI $\mathbf{H}$. With this notation,
\begin{equation}
P_e\left(s_{\text{rc}}|\mathbf{H}\right) \geq \text{SPB}\left(v\left(s_{\text{rc}}\right) | \mathbf{H} \right)
= Pr\left\{\left\|\mathbf{z}\right\|^2 \geq \left(\frac{v\left(s_{\text{rc}}\right)}{V_n}\right)^{\frac{2}{n}} \Big| \mathbf{H} \right\}
\end{equation}
for any $s_{\text{rc}} \in S_{\mathbf{H}}$.
\begin{lem}
\label{lem_sphere_packing_lower_bound_for_regular_ICs}
For any regular IC $S$ with NLD $\delta$, the average error probability is lower bounded by the following sphere packing bound
\begin{equation}
P_e\left(S\right) \geq P_e^{\text{SB}}\left(\delta\right) \triangleq Pr\left\{\left\|\mathbf{z}\right\|^2 \geq e^{-2\delta}\left(\frac{\det(\mathbf{H})}{V_n}\right)^{\frac{2}{n}} \right\}.
\end{equation}
\end{lem}
\begin{proof}
By definition the average error probability is given by
\begin{align}
P_e\left(S\right) &\triangleq E\left\{ \limsup_{a \rightarrow \infty} E_{S,a|\mathbf{H}}\left\{ P_e\left(s_{\text{rc}}|\mathbf{H}\right) \right\} \right\}
\\                \label{align_SPB_regular_IC_step_1}
                  &\geq       E\left\{ \limsup_{a \rightarrow \infty} E_{S,a|\mathbf{H}}\left\{ \text{SPB}\left(v\left(s_{\text{rc}}\right) | \mathbf{H} \right) \right\} \right\}
\\                \label{align_SPB_regular_IC_step_2}
                  &\geq       E\left\{ \limsup_{a \rightarrow \infty}  \text{SPB}\left(E_{S,a|\mathbf{H}}\left\{v\left(s_{\text{rc}}\right)\right\} | \mathbf{H} \right)  \right\}
\\                \label{align_SPB_regular_IC_step_3}
                  &=          E\left\{ \text{SPB}\left(\limsup_{a \rightarrow \infty} E_{S,a|\mathbf{H}}\left\{v\left(s_{\text{rc}}\right)\right\} | \mathbf{H} \right)  \right\}
\\                &=          E\left\{ \text{SPB}\left(V\left(S_{\mathbf{H}}\right) | \mathbf{H} \right)  \right\}
\\                &=          Pr\left\{\left\|\mathbf{z}\right\|^2 \geq \left(\frac{V\left(S_{\mathbf{H}}\right)}{V_n}\right)^{\frac{2}{n}} \right\}
\\                &=          Pr\left\{\left\|\mathbf{z}\right\|^2 \geq e^{-2\delta}\left(\frac{\det(\mathbf{H})}{V_n}\right)^{\frac{2}{n}} \right\}
                  \triangleq P_e^{\text{SB}}\left(\delta\right)
\end{align}
where \eqref{align_SPB_regular_IC_step_1} follows from the sphere packing bound for each $s_{\text{rc}} \in S_{\mathbf{H}}$, \eqref{align_SPB_regular_IC_step_2} follows from Jensen's inequality and the convexity of the function $\text{SPB}\left(v|\mathbf{H}\right)$ in $v$ and \eqref{align_SPB_regular_IC_step_3} follows from the fact that $\text{SPB}\left(v|\mathbf{H}\right)$ is monotone decreasing and a continuous function of $v$. All the next steps are trivial.
\end{proof}
Now we are ready to proof the validity of the sphere packing bound to any IC. This includes IC's with unbounded Voronoi's cells and IC's with density which oscillates with the cube size $a$ (i.e. only the limsup exists in the definition of $\gamma$). The proof is based on a very similar regularization process as done in \cite[Lemma 1]{Ingber} for AWGN channels. Here, in the fading channel case, we will need to separate from the analysis all the ``strong'' fading channel realizations, which are formally defined in the following, and use the regularization process only for the rest of the ``weak'' fading realizations. By showing that the ``strong'' fading realizations in regular fading distributions are an arbitrarily small fraction of the whole realizations space, we will complete the proof of the bound.
\begin{defn}
($\xi$ - strong fading realization): Let us denote by $\mathrm{H} = \emph{\text{diag}}(h_1, \dots, h_n)$ a fading channel realization drawn from a regular fading distribution of the random fading matrix $\mathbf{H}$. For a given $\xi > 0$, let us define a fading threshold $h^*_{\min}(\xi)$ as the solution of $Pr\{H_{\min} \leq h^*_{\min} \} = \xi$, where $H_{\min} \triangleq \min(H_1, \dots, H_n)$. If $h_{\min} \triangleq \min(h_1, \dots, h_n) \leq h^*_{\min}(\xi)$ then $\mathrm{H}$ is called a $\xi$ - strong fading channel realization.
\end{defn}
\begin{lem}
\label{lem_regularixation}
(Regularization): Given the fading channel realization $\mathbf{H}$, let $S_{\mathbf{H}}$ be an IC with density $\gamma_{\emph{\text{rc}}}\left({\mathbf{H}}\right)$ and average error probability $P_e\left(S_{\mathbf{H}}\right) = \epsilon(\mathbf{H})$. For any $\xi > 0$, if $\mathbf{H}$ is not a $\xi$ - strong fading realization then there exists a regular IC, denoted by $S^{'}_{\mathbf{H}}$, with density $\gamma^{'}_{\emph{\text{rc}}}\left({\mathbf{H}}\right) \geq \gamma_{\emph{\text{rc}}}\left({\mathbf{H}}\right)/(1 + \xi)$ and average error probability $P_e\left(S^{'}_{\mathbf{H}}\right) \leq \epsilon(\mathbf{H})(1 + \xi)$.
\end{lem}
\begin{proof}
See Appendix \ref{app_proof_regularization_lemma}.
\end{proof}
\begin{proof}[Proof of Theorem \ref{thm_sphere_packing_lower_bound}]
For a given $\mathbf{H}$, denote the receiver IC by $S_{\mathbf{H}}$. For any $\xi > 0$, by the \emph{regularization lemma}, if $\mathbf{H}$ is not a \emph{$\xi$ - strong fading realization}, then there exists a regular IC, denoted by $S^{'}_{\mathbf{H}}$, with density
\begin{equation}
\gamma^{'}_{\text{rc}}\left(\mathbf{H}\right) \geq \gamma_{\text{rc}}\left(\mathbf{H}\right)/\left(1+\xi\right)= \frac{\gamma}{\left(1+\xi\right)}\cdot\frac{1}{\det\left({\mathbf{H}}\right)}
\end{equation}
and average error probability
\begin{equation}
P_e\left(S^{'}_{\mathbf{H}}\right) \leq P_e\left(S_{\mathbf{H}}\right)\left(1+\xi\right),
\end{equation}
where $\gamma = e^{n\delta}$. Moreover, by the $\xi$ - strong fading definition $Pr\left\{ H_{\text{min}} \leq h^*_{\text{min}} \right\} = \xi$. Following this, we can derive the inequalities below:
\begin{align}
\left(1+\xi\right)P_e\left(S\right) &= E\left\{ \left(1+\xi\right)P_e\left(S_{\mathbf{H}}\right) \right\}
\\ \label{align_spb_step_1}
   &\geq E\left\{ \left(1+\xi\right)P_e\left(S_{\mathbf{H}}\right) \cdot 1_{\left\{ H_{\text{min}} > h^*_{\text{min}} \right\}} \right\}
\\ \label{align_spb_step_2}
   &\geq E\left\{ P_e\left(S^{'}_{\mathbf{H}}\right) \cdot 1_{\left\{ H_{\text{min}} > h^*_{\text{min}} \right\}} \right\}
\\ \label{align_spb_step_3}
   &\geq E\left\{ \text{SPB}\left( {\gamma^{'-1}_{\text{rc}}} \Big| \mathbf{H} \right) \cdot 1_{\left\{ H_{\text{min}} > h^*_{\text{min}} \right\}} \right\}
\\ \label{align_spb_step_4}
   &\geq E\left\{ \text{SPB}\left( {\gamma}^{-1}\det\left(\mathbf{H}\right)\left(1+\xi\right) \Big| \mathbf{H} \right) \cdot 1_{\left\{ H_{\text{min}} > h^*_{\text{min}} \right\}} \right\}
\\ \label{align_spb_step_5}
   &=    E\left\{ \text{SPB}\left( {\gamma}^{-1}\det\left(\mathbf{H}\right)\left(1+\xi\right) \Big| \mathbf{H} \right) \cdot \left(1 - 1_{\left\{ H_{\text{min}} \leq h^*_{\text{min}} \right\}} \right) \right\}
\\ \label{align_spb_step_6}
   &\geq E\left\{ \text{SPB}\left( {\gamma}^{-1}\det\left(\mathbf{H}\right)\left(1+\xi\right) \Big| \mathbf{H} \right) \right\} - Pr\left\{ H_{\text{min}} \leq h^*_{\text{min}} \right\}
\\ \label{align_spb_step_7}
   &=    E\left\{ \text{SPB}\left( {\gamma}^{-1}\det\left(\mathbf{H}\right)\left(1+\xi\right) \Big| \mathbf{H} \right) \right\} - \xi,
\end{align}
where \eqref{align_spb_step_3} follows from the regularity of $S^{'}_{\mathbf{H}}$, \eqref{align_spb_step_4} is due to the fact that $\text{SPB}\left(\cdot|\mathbf{H}\right)$ is a monotone decreasing function and \eqref{align_spb_step_6} is due to $\text{SPB}\left(\cdot|\mathbf{H}\right) \leq 1$.

Equivalently, we get the following:
\begin{equation}
P_e\left(S\right) \geq E\left\{ \frac{ \text{SPB}\left( {\gamma}^{-1}\det\left(\mathbf{H}\right)\left(1+\xi\right) \Big| \mathbf{H} \right)}{ 1 + \xi } - \frac{\xi}{1 + \xi} \right\}
\end{equation}
for all $\xi > 0$. Since $\text{SPB}(\cdot | \mathbf{H})$ is a continuous function we can take the limit $\xi \rightarrow 0$ (meaning implicitly that the ``strong'' fading realizations are an arbitrarily small fraction of the whole realizations space in regular fading distribution) and get the sphere packing lower bound:
\begin{align}
P_e\left(S\right) &\geq E\left\{ \text{SPB}\left( e^{-n\delta}\det\left(\mathbf{H}\right)\Big| \mathbf{H} \right) \right\}
\\                &=          Pr\left\{\left\|\mathbf{z}\right\|^2 \geq e^{-2\delta}\left(\frac{\det(\mathbf{H})}{V_n}\right)^{\frac{2}{n}} \right\} \triangleq P_e^{\text{SB}}\left(\delta\right).
\end{align}
\end{proof}

By taking the fading matrix $\mathbf{H}$ to be equal constantly to the identity matrix $I_n$, the bound \eqref{eq_general_SPB} coincides with the sphere packing bound of the unconstrained AWGN channel, which is given by $Pr\left\{\left\|\mathbf{z}\right\| \geq V_n^{-\frac{1}{n}}e^{-\delta} \right\}$. Although this one dimensional integral is hard to evaluate analytically for general $n$, Ingber et al. derived in \cite{Ingber} an easy to evaluate and very tight analytical bounds for it. These bounds coincide with the sphere packing bound's error exponent, derived by Poltyrev in \cite{Poltyrev}, for asymptotic $n$. Moreover, Tarokh et al. represented this integral in \cite{Tarokh} as a sum of $n/2$ elements, which helps in numerical evaluation of the bound. In contrast, in the case of fading channel the sphere packing bound \eqref{eq_general_SPB} is an $n+1$ dimensional integral, which is extremely hard to evaluate both numerically and analytically. Nevertheless, in the asymptotic case, this bound can be approximated by normal distribution according to the central limit theorem. In the next section, this fact will help us to prove the converse part of our main result.
\subsection{Proof of Converse Part}
\label{sec_proof_of_converse_part}
Assume a transmission of IC $S$ with NLD $\delta$ over the fading channel. By the \emph{sphere packing lower bound} of Theorem \ref{thm_sphere_packing_lower_bound},
\begin{equation}
\label{eq_SPB_step1}
P_e \geq P_e^{\text{SB}}\left(\delta\right) = Pr\left\{\left\|\mathbf{z}\right\|^2 \geq e^{-2\delta}\left(\frac{\det(\mathbf{H})}{V_n}\right)^{\frac{2}{n}} \right\}.
\end{equation}

In \cite{Ingber} Ingber et al. proved the converse part of the dispersion analysis, in the unconstrained AWGN channel, by approximating the distribution of $\left\|\mathbf{z}\right\|^2 = \sum_{i=1}^{n}{Z_i^2}$ by a normal distribution using the Berry-Esseen lemma (see Lemma \ref{lem_Berry_Esseen}) for sum of i.i.d RVs. Here, we cannot use the same analysis due to the fact that $\mathbf{H}$ is also random. By taking the logarithm and rearranging of the inequality in the argument of \eqref{eq_SPB_step1} we get:
\begin{align}
\begin{aligned}
P_e \geq Pr &\Bigg\{ \frac{ \ln\left(\left\|\mathbf{z}\right\|^2\right) - \ln(n\sigma^2) }{\sqrt{\frac{2}{n}}} - \sqrt{\frac{2}{n}}\sum_{i=1}^{n}{\left(\ln(H_i) - E\{\ln(H)\}\right)}
\\ &\geq \sqrt{2n}\left( E\left\{\frac{1}{2}\ln\left(\frac{H^2}{n\sigma^2}\right)\right\} - \delta - \frac{\ln(V_n)}{n} \right) \Bigg\}.
\end{aligned}
\end{align}

For simplicity, let us define $Y_n \triangleq  \frac{ \ln\left(\left\|\mathbf{z}\right\|^2\right) - \ln\left(n\sigma^2\right) }{\sqrt{\frac{2}{n}}}$,
$S_n \triangleq \frac{\sum_{i=1}^{n}X_i}{\sqrt{n}}$ where
$X_i \triangleq \frac{\ln(H_i)-E\{\ln(H)\}}{\sqrt{Var(\delta(H))}}$ (for $i=1,..,n$) and $\zeta_n \triangleq \frac{1}{\sqrt{2}}Y_n - \sqrt{Var(\delta(H))}S_n$ to get:
\begin{equation}
\label{eq_SPB_step3}
P_e \geq
Pr\left\{ \zeta_n \geq \zeta \right\},
\end{equation}
where $\zeta \triangleq \sqrt{n}\left( E\left\{\frac{1}{2}\ln\left(\frac{H^2}{n\sigma^2}\right)\right\} - \delta - \frac{\ln(V_n)}{n} \right)$.

Although $\zeta_n$ is a sum of $n+1$ independent RVs, and despite of the existence of expansions for the Berry-Essen Lemma for a sum of independent RVs with varying distributions, in the standard derivation of these expansions it is assumed that all the RVs' variances are of the same order (see \cite[pp. 542-548]{Feller} for details). Here, $Var(Y_n) = O(1)$ (see Lemma \ref{lem_ln_chi2_pdf}) and $Var\left(\frac{X_i}{\sqrt{n}}\right) = O\left(\frac{1}{n}\right)$. Hence, a more careful analysis should be done for proving that the distribution of $\zeta_n$ is approximately normal. The following three lemmas allow it. By Lemma \ref{lem_ln_chi2_pdf} and by Lemma \ref{lem_Berry_Esseen} we prove that the PDF of $Y_n$ and the CDF of $S_n$ are approximately normal for large enough $n$, respectively. Finally by Lemma \ref{lem_sum_of_almost_normal_RVs} we prove that the distribution of a sum of two independent RVs, each of which has an approximately normal distribution, is also approximately normal. Therefore, the distribution of $\zeta_n$ is also approximately normal for large enough $n$.
\begin{lem}
\label{lem_ln_chi2_pdf}
(Log of chi square distribution) Let $Y_n \triangleq \frac{\ln\left(X\right)-\ln(n)}{\sqrt{\frac{2}{n}}}$, where $X \sim \chi^2_n$. Then
\begin{equation}
\label{eq_Yn_PDF}
f_{Y_n}(y) = \frac{ (\frac{n}{2})^{\frac{n-1}{2}}}{\Gamma(\frac{n}{2})}e^{\sqrt{\frac{n}{2}}y - \frac{n}{2}e^{\sqrt{\frac{2}{n}}y}},
\end{equation}
and for large enough $n$:
\begin{equation}
\label{eq_Yn_asymptotic_PDF}
f_{Y_n}(y) = N(0,1) + e_n(y)
\\~s.t. \int_{-\infty}^{\infty}|e_n(y)|dy = O\left(\frac{1}{\sqrt{n}}\right),
\end{equation}
where $N(0,1)$ is the standard normal distribution's PDF.
\end{lem}
\begin{proof}
See Appendix \ref{app_proof_lemma_ln_chi2_density_moments}. Illustratively, the convergence of $f_{Y_n}(y)$ to the standard normal distribution's PDF $N(0,1)$, can be seen in figure \ref{fig_log_of_chi2_distribution}.
\begin{figure}[htp]
\center{\includegraphics[width=0.7\columnwidth]{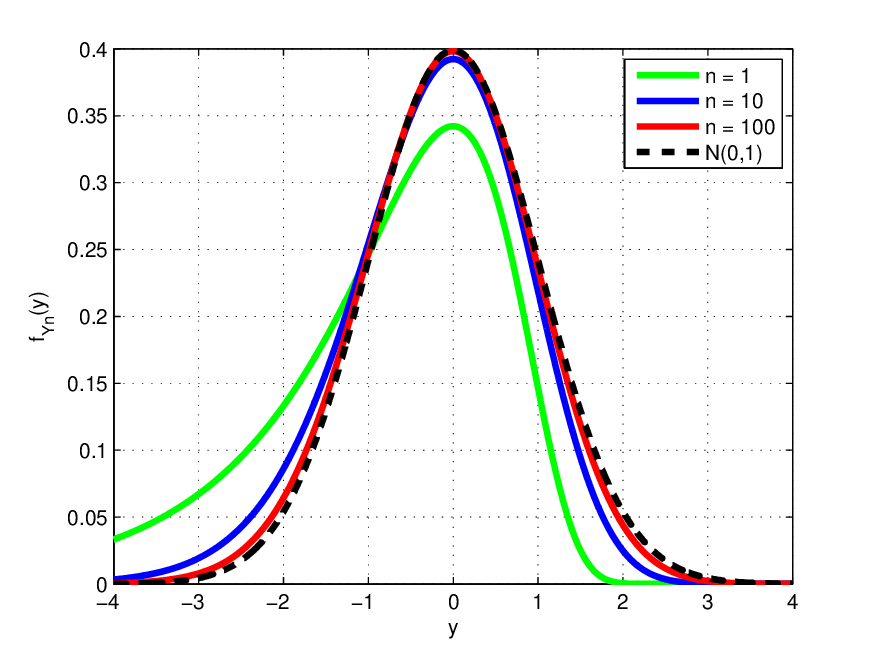}}
\caption{\label{fig_log_of_chi2_distribution} The PDF $f_{Y_n}(y)$ for different values of $n$. The convergence to $N(0,1)$ can be observed.}
\end{figure}
\end{proof}
\begin{lem}
\label{lem_Berry_Esseen}
(Berry-Esseen) Let $X_1,X_2,\dots,X_n$ be $n$ i.i.d. random variables with mean, variance and third absolute moments that equal $\mu = E\{X_i\}, \sigma^2 = Var(X_i)$ and $\rho_3=E\{|X_i-\mu|^3\}$, respectively, for $i=1,\dots,n$. If the third absolute moment exists, then for all $-\infty < s < \infty$ and $n$,
\begin{equation}
\Big| F_{S_n}(s) - F_{N(0,1)}(s) \Big| \leq \frac{6\rho_3}{\sqrt{n}\sigma^3},
\end{equation}
where $S_n\triangleq\frac{\sum_{i=1}^{n}(X_i - \mu)}{\sqrt{n}\sigma}$ and $F_{N(0,1)}(\cdot)$ is the standard normal distribution's CDF.
\end{lem}
\begin{proof}
See Berry-Esseen theorem for sum of i.i.d. RVs in \cite[pp. 542, Theorem 1]{Feller}.
\end{proof}
\begin{lem}
\label{lem_sum_of_almost_normal_RVs}
(Sum of two almost normal RVs) Suppose that $X_1$ and $X_2$ are two independent random variables s.t. the PDF of $X_1$ is given by
\begin{equation}
f_{X_1}(x_1) = N(0, \sigma_1^2) + e_n(y)
\\~s.t. \int_{-\infty}^{\infty}|e_n(y)|dy = O\left(\frac{1}{\sqrt{n}}\right), \nonumber
\end{equation}
and the CDF of $X_2$ is given by
\begin{equation}
F_{X_2}(x_2) = F_{N(0,\sigma_2^2)}(x_2) + O\left(\frac{1}{\sqrt{n}}\right). \nonumber
\end{equation}

Let $Y \triangleq X_1 + X_2$, then the following holds:
\begin{equation}
\label{eq_sum_of_almost_normal_CDF}
F_Y(y) = F_{N(0,\sigma_y^2)}(y) + O\left(\frac{1}{\sqrt{n}}\right),
\end{equation}
where $\sigma_y^2 \triangleq \sigma_1^2 + \sigma_2^2$.
\end{lem}
\begin{proof}
See Appendix \ref{app_proof_lemma_sum_of_almost_normal_RVs}.
\end{proof}

Combining Lemmas \ref{lem_ln_chi2_pdf}, \ref{lem_Berry_Esseen} and \ref{lem_sum_of_almost_normal_RVs} we get:
\begin{equation}
\label{eq_SPB_step4}
P_e \geq Q\left(\frac{\zeta}{\sqrt{V}}\right) - O\left(\frac{1}{\sqrt{n}}\right).
\end{equation}
By Stirling approximation for the Gamma function, $V_n$ can be approximated as
\begin{equation}
\frac{\ln(V_n)}{n} = \frac{1}{2}\ln\left(\frac{2\pi e}{n}\right) - \frac{1}{2n}\ln(n) + O\left(\frac{1}{n}\right)
\end{equation}
and hence we get:
\begin{equation}
\label{eq_zeta}
\zeta = \sqrt{n}\left( \delta^* - \delta + \frac{1}{2n}\ln(n) + O\left(\frac{1}{n}\right) \right).
\end{equation}
The assignment of \eqref{eq_zeta} in \eqref{eq_SPB_step4} gives us:
\begin{equation}
\label{eq_SPB_step5}
\epsilon \geq P_e \geq
Q\left(\frac{ \delta^* - \delta + \frac{1}{2n}\ln(n) + O\left(\frac{1}{n}\right) }{ \sqrt{ \frac{V}{n} } }\right) - O\left(\frac{1}{\sqrt{n}}\right).
\end{equation}
Taking $Q^{-1}(\cdot)$ from both sides of \eqref{eq_SPB_step5} gives us:
\begin{equation}
\label{eq_SPB_step6}
\delta \leq \delta^* - \sqrt{ \frac{V}{n} }Q^{-1}\left(\epsilon + O\left(\frac{1}{\sqrt{n}}\right)\right) + \frac{1}{2n}\ln(n) + O\left(\frac{1}{n}\right).
\end{equation}

By Taylor approximation (around $\epsilon$) $  Q^{-1}\left(\epsilon + O\left(\frac{1}{\sqrt{n}}\right)\right) = Q^{-1}(\epsilon) + O\left(\frac{1}{\sqrt{n}}\right)$, which gives us the desired result:
\begin{equation}
\label{eq_SPB_step7}
\delta \leq \delta^* - \sqrt{ \frac{V}{n} }Q^{-1}(\epsilon) + \frac{1}{2n}\ln(n) + O\left(\frac{1}{n}\right).
\end{equation}
\section{Direct Part}
\label{sec_direct_part}
In this section we prove the direct part of Theorem \ref{thm_main_result}. The direct part is based on normal approximation of the \emph{Dependence Testing upper bound} on the average error probability.
The dependence testing upper bound over fading channels is presented in Section \ref{sec_dependence_testing_bound}, and in Section \ref{sec_proof_of_direct_part} we complete the proof by a derivation of an appropriate normal approximation technique.
\subsection{Dependence Testing Bound}
\label{sec_dependence_testing_bound}
In this section we extend Polyanskiy's \emph{Dependence Testing Bound} \cite[Theorems 17,18]{Polyanskiy}, to the case of fading channels with available CSI at the receiver. In \cite{Polyanskiy} the DT bound was used to prove the dispersion analysis for DMCs, or more precisely, for memoryless channels without a power constraint (or any other constraint on the channel input). Here, the channel input does not have any restriction, and hence we can use the DT bound to prove the direct part of our main result.
\begin{theorem}
\label{thm_DT_bound}
(DT bound) For any input distribution $f_X(\cdot)$ on $\mathbb{R}$, there exists a code with $M$ codewords and an average error probability over the fading channel, with available CSI at the receiver, not exceeding
\begin{equation}
\label{eq_DT_bound}
P_e \leq Pr\left\{ i(\mathbf{x};\mathbf{y},\mathbf{H}) \leq \ln\left(\frac{M-1}{2}\right) \right\}
    + \frac{M-1}{2} Pr\left\{i(\mathbf{x};\bar{\mathbf{y}},\mathbf{H}) > \ln\left(\frac{M-1}{2}\right)\right\},
\end{equation}
or equivalently,
\begin{align}
\begin{aligned}
\label{align_DT_bound_equivalence}
P_e &\leq E\left\{e^{-\left[i(\mathbf{x};\mathbf{y},\mathbf{H}) - \ln\left(\frac{M-1}{2}\right)\right]^+}\right\}
\\ &= Pr\left\{ i\left(\mathbf{x};\mathbf{y},\mathbf{H}\right) \leq \ln\left(\frac{M-1}{2}\right) \right\}
+ \frac{M-1}{2} E\left\{ e^{-i\left(\mathbf{x};\mathbf{y},\mathbf{H}\right)}1_{\left\{i\left(\mathbf{x};\mathbf{y},\mathbf{H}\right)>\ln\left(\frac{M-1}{2}\right)\right\}} \right\},
\end{aligned}
\end{align}
where $f_{\mathbf{x}\mathbf{y}\bar{\mathbf{y}}\mathbf{H}}({x},{y},\bar{{y}},{h}) = f_{\mathbf{x}}({x})f_{\mathbf{y}|\mathbf{x},\mathbf{H}}({y}|{x},{h})f_{\mathbf{y}|\mathbf{H}}(\bar{{y}}|{h})f_\mathbf{H}({h})$ is the joint PDF of all the random vectors and matrices arising above, $f_{\mathbf{x}}({x})=\Pi_{i=1}^nf_X(x_i)$ and $i({x};{y},{h}) \triangleq \ln\left( \frac{f_{\mathbf{x}\mathbf{y}\mathbf{H}}({x}, {y}, {h})}{f_{\mathbf{x}}({x})f_{\mathbf{y}\mathbf{H}}({y}, {h})} \right)$.
\end{theorem}
\begin{proof}
The proof is based on Shannon's random coding technique and on a suboptimal decoder. For a given input distribution $f_X(x)$ , let us define the following deterministic function:
\begin{equation}
g_{\mathbf{x}}\left(\mathbf{y},\mathbf{H}\right) = 1_{\left\{i(\mathbf{x};\mathbf{y},\mathbf{H}) > \ln\left(\frac{M-1}{2}\right)\right\}}.
\end{equation}

For a given codebook $C = \left\{c_1, \dots, c_M\right\}$, the decoder computes the $M$ values of $g_{c_j}\left(\mathbf{y},\mathbf{H}\right)$ for the given channel output $\left(\mathbf{y},\mathbf{H}\right)$ and returns the lowest index $j$ for which $g_{c_j}\left(\mathbf{y},\mathbf{H}\right) = 1$, or declares an error if there is no such index. Hence, the error probability, given that $\mathbf{x} = c_j$ was transmitted, is given by:
\begin{align}
\label{align_DT_bound_step_1}
\begin{aligned}
&Pr\left\{ \{g_{c_j}\left(\mathbf{y},\mathbf{H}\right) = 0\} \bigcup_{i<j} \{g_{c_i}\left(\mathbf{y},\mathbf{H}\right) = 1\} | ~\mathbf{x} = c_j \right\} \leq
\\  &Pr\left\{ i(c_j;\mathbf{y},\mathbf{H}) \leq \ln\left(\frac{M-1}{2}\right) | ~\mathbf{x} = c_j \right\}
 + \sum_{i < j}Pr\left\{ i(c_i;\bar{\mathbf{y}},\mathbf{H}) > \ln\left(\frac{M-1}{2}\right) | ~\mathbf{x} = c_j \right\},
\end{aligned}
\end{align}
where the right hand side (RHS) of \eqref{align_DT_bound_step_1} is obtained by using the union bound and the definition of $\bar{\mathbf{y}}$ as a random vector which is independent of $\mathbf{x}$ and given $\mathbf{H}$ has the same conditional distribution as $\mathbf{y}$ given $\mathbf{H}$.

Let us define the ensemble of the codebooks of size M, that every codeword's component in it is drawn independently of each other by $f_X(x)$. Averaging \eqref{align_DT_bound_step_1} over this ensemble and over the $M$ equiprobable codewords we obtain
\begin{align}
\begin{aligned}
P_e &\leq Pr\left\{ i(\mathbf{x};\mathbf{y},\mathbf{H}) \leq \ln\left(\frac{M-1}{2}\right) \right\}
    + \sum_{j=1}^{M}{ \frac{j-1}{M}Pr\left\{i(\mathbf{x};\bar{\mathbf{y}},\mathbf{H}) > \ln\left(\frac{M-1}{2}\right)\right\} },
\end{aligned}
\end{align}
which completes the proof of the existence of a code with $M$ codewords whose average error probability is upper bounded by \eqref{eq_DT_bound}.

Now we turn to prove the equivalent bound \eqref{align_DT_bound_equivalence} of the theorem. For any positive $\gamma$ the following identities hold:
\begin{align}
E\left\{e^{-\left[i(\mathbf{x};\mathbf{y},\mathbf{H}) - \ln(\gamma)\right]^+}\right\}
   &= E\left\{ 1_{\left\{i\left(\mathbf{x};\mathbf{y},\mathbf{H}\right) \leq \ln(\gamma) \right\}}
   + \gamma e^{-i\left(\mathbf{x};\mathbf{y},\mathbf{H}\right)}1_{\left\{i\left(\mathbf{x};\mathbf{y},\mathbf{H}\right)>\ln(\gamma) \right\}} \right\}
\\ &= Pr\left\{ i\left(\mathbf{x};\mathbf{y},\mathbf{H}\right) \leq \ln(\gamma) \right\}
+ \gamma E\left\{ e^{-i\left(\mathbf{x};\mathbf{y},\mathbf{H}\right)}1_{\left\{i\left(\mathbf{x};\mathbf{y},\mathbf{H}\right)>\ln(\gamma) \right\}} \right\}
\\ &= Pr\left\{ i\left(\mathbf{x};\mathbf{y},\mathbf{H}\right) \leq \ln(\gamma) \right\}
+ \gamma E\left\{ \frac{f(\mathbf{x})f(\mathbf{y},\mathbf{H})}{f(\mathbf{x},\mathbf{y},\mathbf{H})} 1_{\left\{i\left(\mathbf{x};\mathbf{y},\mathbf{H}\right)>\ln(\gamma) \right\}} \right\}
\\ &= Pr\left\{ i\left(\mathbf{x};\mathbf{y},\mathbf{H}\right) \leq \ln(\gamma) \right\}
+ \gamma Pr\left\{ i(\mathbf{x};\bar{\mathbf{y}},\mathbf{H}) > \ln(\gamma) \right\}.
\end{align}
By taking $\gamma = \frac{M-1}{2}$ we complete the proof.

It is important to notice that the dependence testing bound is based on a suboptimal decoder which is actually a threshold crossing decoder. The decoder computes $M$ binary hypothesis tests in parallel and declares as the decoded codeword the first one that crosses the threshold $\ln\left(\frac{M-1}{2}\right)$.
\end{proof}
\subsection{Proof of Direct Part}
\label{sec_proof_of_direct_part}
For the proof of the direct part, we will first construct an ensemble of finite constellations with $M$ codewords, which are uniformly distributed in an $n$ dimensional cube $\text{Cb}(a)$, for some fixed $a$ and $n$. Then, using the \emph{Dependence Testing bound} of Theorem \ref{thm_DT_bound} with $f_X(x) = U(-\frac{a}{2},\frac{a}{2})$, we will find a lower bound on the optimal achievable number of codewords, for a FC in such an ensemble, whose error probability is upper bounded by some fixed $\epsilon > 0$. We will denote this lower bound by $M(n,\epsilon,a/\sigma)$. Theorem \ref{thm_DT_bound} also ensures the existence of such a FC that achieves this lower bound. Finally, we will construct an IC by tiling this FC to the whole space $\mathbb{R}^n$, in a way that will preserve the density of codewords and the error probability, asymptotically in the dimension $n$, as in this FC.

To use the DT bound of Theorem \ref{thm_DT_bound}, we need to prove that for some $\gamma$ the following inequality holds:
\begin{align}
\label{align_DT_bound}
P_e \leq Pr\left\{ i\left(\mathbf{x};\mathbf{y},\mathbf{H}\right) \leq \ln(\gamma) \right\} + \gamma E\left\{ e^{-i\left(\mathbf{x};\mathbf{y},\mathbf{H}\right)}1_{\left\{i\left(\mathbf{x};\mathbf{y},\mathbf{H}\right)>\ln(\gamma)\right\}} \right\} \leq \epsilon.
\end{align}
Denote for arbitrary $\tau$
\begin{equation}
\ln(\gamma) = nI(X;Y,H) - \tau\sqrt{nVar(i(X;Y,H))}.
\end{equation}
The information density is a sum of $n$ i.i.d. RVs:
\begin{equation}
i\left(\mathbf{x};\mathbf{y},\mathbf{H}\right) = \sum_{j=1}^{n}{ i(X_j;Y_j,H_j) },
\end{equation}
where $i(X;Y,H) \triangleq \ln\left( \frac{f(Y|H,X)}{f(Y|H)} \right)$ and its moments, for large enough $a/\sigma$, are given by the following lemma.
\begin{lem}
\label{lem_information_density_moments}
(Information density's moments) If $X\sim U\left(-\frac{a}{2},\frac{a}{2}\right)$ and if the PDF of $H$ is a regular fading distribution, then for large enough $a/\sigma$ and for some positive constant $0<\alpha\leq1$, the moments of the information density $i(X;Y,H)$ are given by:
\begin{enumerate}
  \item $I(X;Y,H) \triangleq E\{i(X;Y,H)\} = E\left\{ \frac{1}{2}\ln\left(\frac{a^2H^2}{2\pi e\sigma^2}\right) \right\} + O\left((\frac{\sigma}{a})^{\alpha}\right)$
  \item $Var(i(X;Y,H)) = \frac{1}{2} + Var(\delta(H)) + O\left((\frac{\sigma}{a})^{\frac{\alpha}{2}}\right)$
  \item $\rho_3 \triangleq E\left\{|i(X;Y,H) - I(X;Y,H)|^3\right\} < \infty$.
\end{enumerate}
\end{lem}
\begin{proof}
See Appendix \ref{app_proof_lemma_information_density_moments}.
\end{proof}

According to the Berry-Essen lemma (see Lemma \ref{lem_Berry_Esseen}) for i.i.d. RVs,
\begin{equation}
\label{eq_using_Berry_Esseen_UB}
|Pr\{ i\left(\mathbf{x};\mathbf{y},\mathbf{H}\right) \leq \ln\gamma \} - Q(\tau)| \leq \frac{B(a/\sigma)}{\sqrt{n}}
\end{equation}
where
\begin{equation}
\label{eq_B}
B(a/\sigma) = \frac{6\rho_3}{Var^{\frac{3}{2}}(i(X;Y,H))}.
\end{equation}
For sufficiently large $n$, let
\begin{equation}
\tau = Q^{-1}\left(\epsilon - \left(\frac{2\ln(2)}{\sqrt{2\pi Var(i(X;Y,H))}} + 5B(a/\sigma)\right)\frac{1}{\sqrt{n}} \right).
\end{equation}
Then, from \eqref{eq_using_Berry_Esseen_UB} we obtain
\begin{equation}
\label{eq_part1_UB_for_DT_bound}
Pr\left\{ i\left(\mathbf{x};\mathbf{y},\mathbf{H}\right) \leq \ln(\gamma) \right\} \leq \epsilon - 2\left(\frac{\ln(2)}{\sqrt{2\pi Var(i(X;Y,H))}} + 2B(a/\sigma)\right)\frac{1}{\sqrt{n}}.
\end{equation}
Using Lemma \ref{lem_47_in_Polyanskiy} (see in Appendix \ref{app_lemma_47_in_Polyanskiy}), we get
\begin{equation}
\label{eq_part2_UB_for_DT_bound}
\gamma E\left\{ e^{-i\left(\mathbf{x};\mathbf{y},\mathbf{H}\right)}1_{\left\{i\left(\mathbf{x};\mathbf{y},\mathbf{H}\right)>\ln(\gamma)\right\}} \right\} \leq 2\left(\frac{\ln(2)}{\sqrt{2\pi Var(i(X;Y,H))}} + 2B(a/\sigma)\right)\frac{1}{\sqrt{n}}.
\end{equation}
Summing \eqref{eq_part1_UB_for_DT_bound} and \eqref{eq_part2_UB_for_DT_bound} we prove the inequality \eqref{align_DT_bound}. Hence, by Theorem \ref{thm_DT_bound}, there exists a FC, denoted by $S(n, \epsilon, a/\sigma)$, with $M(n,\epsilon,a/\sigma)$ codewords and average error probability upper bounded by $\epsilon$, such that
\begin{align}
\label{align_Direct_step_1}
\begin{aligned}
\ln \left(M(n,\epsilon,a/\sigma)\right) &= \ln(\gamma) + O(1)
\\ &= nI(X;Y,H) - \tau\sqrt{nVar(i(X;Y,H))}  + O(1)
\\ &= nI(X;Y,H) - \sqrt{nVar(i(X;Y,H))}Q^{-1}(\epsilon) + O(1),
\end{aligned}
\end{align}
where the last equality is derived by Taylor approximation for $Q^{-1}\left(\epsilon + O\left(\frac{1}{\sqrt{n}}\right)\right)$ around $\epsilon$.
Let us define the NLD of the FC in $\text{Cb}(a)$ by
\begin{equation}
\label{eq_NLD_in_FC}
\delta(n, \epsilon, a/\sigma) \triangleq \frac{1}{n}\ln\left(\frac{M(n,\epsilon,a/\sigma)}{a^n}\right).
\end{equation}
From \eqref{align_Direct_step_1} we obtain
\begin{align}
\label{align_Direct_step_2}
\delta(n, \epsilon, a/\sigma) = I(X;Y,H) - \ln(a) - \sqrt{\frac{Var(i(X;Y,H))}{n}}Q^{-1}(\epsilon) + O\left(\frac{1}{n}\right).
\end{align}
Note that the results of Lemma \ref{lem_information_density_moments} hold in general for large enough $a$. Specifically, we can choose $a$ to be a monotonic increasing function of $n$ s.t. $\lim_{n \to \infty}a=\infty$, and then the results of Lemma \ref{lem_information_density_moments} will hold for any large enough $n$. Assigning the results of Lemma \ref{lem_information_density_moments} with appropriate choice of $a = a(n)$, we get
\begin{align}
\label{align_Direct_step_3}
\delta(n, \epsilon, a/\sigma) &= \delta^* -
\sqrt{ \frac{V + O\left(\left(\frac{\sigma}{a}\right)^{\frac{\alpha}{2}}\right)}{n} }Q^{-1}(\epsilon)
+ O\left(\frac{1}{n} + \left(\frac{\sigma}{a}\right)^{\alpha}\right).
\end{align}
Using Taylor approximation for large enough $n$,
\begin{equation}
\sqrt{V + O\left(\left(\frac{\sigma}{a}\right)^{\frac{\alpha}{2}}\right)} =
\sqrt{V} + O\left(\left(\frac{\sigma}{a}\right)^{\frac{\alpha}{2}}\right).
\end{equation}
Hence, we get
\begin{align}
\label{align_Direct_step_4}
\delta(n, \epsilon, a/\sigma) = \delta^* - \sqrt{ \frac{V}{n} }Q^{-1}(\epsilon) +
O\left(\frac{1}{n} + \frac{1}{\sqrt{n}}\left(\frac{\sigma}{a}\right)^{\frac{\alpha}{2}} + \left(\frac{\sigma}{a}\right)^{\alpha}\right).
\end{align}

By tiling the FC, denoted by $S(n, \epsilon, a/\sigma)$, to the whole space $\mathbb{R}^n$ and by choosing for example $a(n) = \sigma\cdot n^{2+\frac{2}{\alpha}}$, we can construct an IC (See Appendix \ref{app_tiling} for details) with average error probability which is upper bounded by $\epsilon$ and NLD $\delta(n,\epsilon)$ that satisfies
\begin{equation}
\label{eq_Direct_step_5}
\delta(n, \epsilon) = \delta^* - \sqrt{ \frac{V}{n} }Q^{-1}(\epsilon) + O\left(\frac{1}{n}\right).
\end{equation}
Hence, the optimal NLD $\delta^*(n,\epsilon)$ necessarily satisfies
\begin{align}
\label{align_Direct_step_6}
\delta^*(n, \epsilon) &\geq \delta(n, \epsilon) = \delta^* - \sqrt{ \frac{V}{n} }Q^{-1}(\epsilon) + O\left(\frac{1}{n}\right),
\end{align}
which completes the proof of the direct part.

We can observe that in the case of AWGN, namely $H = 1$ deterministically, our result coincides with the weaker achievability
bound of the dispersion analysis of Ingber et al. in \cite{Ingber}. This weaker bound is based on the suboptimal typicality decoder. The stronger bound in \cite{Ingber}, which is based on the optimal ML decoder, is greater than the typicality bound in $\frac{1}{2n}\ln(n)$. Hence, we conjecture that by using a ML decoder, instead of the suboptimal dependence testing decoder, the achievability bound is, actually, given by:
\begin{equation}
\delta^*(n, \epsilon) \geq \delta^* - \sqrt{ \frac{V}{n} }Q^{-1}(\epsilon) + \frac{1}{2n}\ln(n) + O\left(\frac{1}{n}\right).
\end{equation}
\section{Extension to the Complex Channel Model}
\label{sec_extension_to_the_complex_model}
In this section we extend our main result to the case of scalar complex channel model. First, we will define the complex fading channel model and then we will explain its similarity to the scalar real model. Finally, we will give the outline of the proof of the theorem in this setting.

In the complex model, $Y = H\cdot X + Z$ where $X$, $H$ and $Z$ are independent complex RVs. Moreover, $E\left\{|H|^2\right\} = 1$ and $Z \sim CN(0, \sigma^2)$ with i.i.d. real and imaginary components.

Generally, $H$ is a complex RV, but since in our model the CSI is known at the receiver, we can assume
that $H$ is a real and nonnegative RV, without loss of generality. Hence, the complex model is equivalent to the following two scalar real models:
\begin{align}
Y_r = |H|\cdot X_r + Z_r
\\
Y_i = |H|\cdot X_i + Z_i
\end{align}
where, $X = X_r + jX_i$, $Y = Y_r + jY_i$ and $Z = Z_r + jZ_i$.
\begin{theorem}
\label{thm_main_result_complex_extension}
Let $\epsilon>0$ be a given, fixed, error probability. Denote by $\delta_c^*(n,\epsilon)$ the optimal NLD for which there exists an $n$ complex-dimensional infinite constellation with average error probability at most $\epsilon$. Then, for any regular fading distribution of $|H|$, as $n$ grows,
\begin{equation}
\label{eq_main_result_complex}
\delta_c^*(n,\epsilon) = \delta_c^* - \sqrt{\frac{V_c}{n}}Q^{-1}(\epsilon) + O\left(\frac{\ln(n)}{n}\right),
\end{equation}
where,
\begin{align}
\delta_c^* &\triangleq E\left\{\delta_c(H)\right\} = E\left\{ \ln\left(\frac{|H|^2}{\pi e\sigma^2}\right) \right\}
\\ V_c     &\triangleq 1 + Var(\delta_c(H)) = 1 + Var\left(\ln\left(|H|^2\right)\right)
\end{align}
noting that
\begin{equation}
\delta_c(H) \triangleq \ln\left(\frac{|H|^2}{\pi e\sigma^2}\right).
\end{equation}
\end{theorem}
\subsection{Proof outline of the direct part}
In a similar way to the proof of the direct part of scalar real models, we will construct an ensemble of finite constellations with $M$ codewords, which are uniformly distributed in an $n$ complex-dimensional cube $\text{Cb}(a)$. To be more precise, each codeword's component (its real and imaginary parts) in this ensemble is drawn uniformly according to the distribution $U(-\frac{a}{2},\frac{a}{2})$, independently of each other. Then, using the \emph{Dependence Testing bound} of Theorem \ref{thm_DT_bound} over this ensemble and the Berry-Essen lemma (see Lemma \ref{lem_Berry_Esseen}), we can prove the existence of a FC with $M(n,\epsilon,a/\sigma)$ codewords and with an average error probability upper bounded by $\epsilon$, which satisfies the following:
\begin{align}
\label{align_Direct_complex_case_step_1}
\begin{aligned}
\delta_c(n, \epsilon, a/\sigma) &\triangleq \ln\left(\frac{M(n,\epsilon,a/\sigma)}{a^{2n}}\right)
\\ &= I(X;Y,H) - \ln(a^2) - \sqrt{\frac{Var(i(X;Y,H))}{n}}Q^{-1}(\epsilon) + O\left(\frac{1}{n}\right).
\end{aligned}
\end{align}
In this case the information density is given by
\begin{align}
\begin{aligned}
i(X;Y,H) &= \ln\left( \frac{f(Y|X,H)}{f(Y|H)} \right)
\\       &= \ln\left( \frac{f(Y_r|X_r,|H|)f(Y_i|X_i,|H|)}{f(Y_r||H|)f(Y_i||H|)} \right)
\\       &= i(X_r;Y_r,|H|) + i(X_i;Y_i,|H|).
\end{aligned}
\end{align}
Hence, by equivalent calculations as in Lemma \ref{lem_information_density_moments}, we can obtain
\begin{align}
\begin{aligned}
\label{align_moments_in_complex_case}
I(X;Y,H) &= E\left\{ \ln\left(\frac{a^2|H|^2}{\pi e\sigma^2}\right) \right\} + o(1)
\\
Var(i(X;Y,H)) &= 1 + Var\left(\ln\left(\frac{a^2|H|^2}{\pi e\sigma^2}\right)\right) + o(1),
\end{aligned}
\end{align}
where $o(1)$ converges to zero as $\sigma/a$ tends to zero. Combining \eqref{align_Direct_complex_case_step_1} and \eqref{align_moments_in_complex_case} gives us the following:
\begin{align}
\label{align_Direct_complex_case_step_2}
\delta_c(n, \epsilon, a/\sigma) = \delta_c^* + o(1) - \sqrt{\frac{V_c + o(1)}{n}}Q^{-1}(\epsilon) + O\left(\frac{1}{n}\right).
\end{align}

By tiling this FC to the whole space $\mathbb{C}^n$ we can prove the existence of IC with an average error probability upper bounded by $\epsilon$ and NLD that equals the RHS of \eqref{eq_main_result_complex}. This completes the proof of the direct part.
\subsection{Proof outline of the converse part}
Using the same arguments as in the scalar real fading channel model, we can prove that the sphere packing lower bound of complex fading channels, is given by
\begin{equation}
\label{eq_SPB_complex_step1}
P_e \geq P_e^{\text{SB}}\left(\delta_c\right) = Pr\left\{\left\|\mathbf{z}\right\|^2 \geq e^{-\delta_c}\left(\frac{\det\left(\mathbf{H}^{\dagger}\mathbf{H}\right)}{V_{2n}}\right)^{\frac{1}{n}} \right\}
\end{equation}
for any IC $S$ with NLD $\delta_c$, where $2\cdot\left\|\mathbf{z}\right\|^2 / \sigma^2 \sim \chi^2_{2n}$ and $\mathbf{H} = \text{diag}(H_1,\dots,H_n)$.

Using the same normal approximation techniques as in the case of the scalar real fading model, we can prove that for any $n$ complex-dimensional IC, with NLD $\delta_c$ and average error probability upper bounded by $\epsilon$, over the complex fading channel, the following holds:
\begin{equation}
\label{eq_converse_complex_case}
\delta_c \leq \delta^*_c - \sqrt{ \frac{V_c}{n} }Q^{-1}(\epsilon) + \frac{1}{2n}\ln(n) + O\left(\frac{1}{n}\right),
\end{equation}
which completes the proof of the converse part.
\section{Extension for Lattices}
\label{sec_extension_for_lattices}
In this section we extend the validity of our main result in Theorem \ref{thm_main_result} to the special case of \emph{Lattices}.
Lattices are the most practical infinite constellations due to theirs structure, and they are essentially the Euclidean space analog of linear codes.
These properties may allow efficient encoding and decoding algorithms \cite{Sommer}.
The proof here is based on an extension of the suboptimal \emph{Typicality Decoder} proposed by Ingber et al. in \cite{Ingber} for communicating over the unconstrained AWGN channel.
\begin{theorem}
\label{thm_typicality_decoder_based_bound}
(Typicality decoder based bound): Denote $r=r(\mathbf{H})=r_0\sqrt[n]{\det(\mathbf{H})}$. Then for any $n$, $r_0>0$ and $\delta = \frac{1}{n}\ln(\gamma)$, there exists an $n$-dimensional lattice $\Lambda$ with NLD $\delta$ and average (maximal) error probability over the unconstrained fading channel with CSI at the receiver, which satisfy:
\begin{align}
\label{align_typicality_decoder_based_bound}
\begin{aligned}
P_e(\Lambda) &\leq Pr\left\{\|\mathbf{z}\|>r\right\} + \gamma V_n r_0^n
               + Pr\left\{H_{\max} > g_{\max}(n) \cup H_{\min} < g_{\min}(n) \right\}
\end{aligned}
\end{align}
where, $H_{\min/\max}\triangleq\min/\max(H_1,\dots,H_n)$ and $g_{\min}(n)\leq g_{\max}(n)$ are arbitrary thresholds.
\end{theorem}
\begin{proof}
Let $\Lambda$ be a lattice that is used as IC for communicating over the unconstrained fading channel. Suppose that $\lambda\in\Lambda$ was sent. Then, $\mathbf{y}=\mathbf{H}\cdot\lambda + \mathbf{z}$. Denote by $\Lambda_{\mathbf{H}}\triangleq\mathbf{H}\cdot\Lambda$ the receiver's lattice. In addition, let $r$ be a parameter that plays the role of a threshold for decoding using the suboptimal \emph{typicality decoder}, which operates as follows. If the ball $\text{Ball}(\mathbf{y},r)$ contains only a single point $\lambda_{\text{rc}}=\mathbf{H}\cdot\lambda_0$ in the receiver's lattice, then the point $\lambda_0$ will be the decoded codeword. Otherwise, an error will be declared.
We note that the decoding operation is only restricted to the case where the minimal fading coefficient $H_{\min}=\min(H_1,\dots,H_n)$ and the maximal fading coefficient $H_{\max}=\max(H_1,\dots,H_n)$ are not crossing a predefined thresholds $g_{\min}(n)$ and $g_{\max}(n)$, respectively. This is in order to guarantee a finite support of the fading channel, i.e., any fading coefficient satisfies $H \in [g_{\min}(n),g_{\max}(n)]$. Otherwise, an error will also be declared.
Hence, the error probability given $\mathbf{H}$ satisfies:
\begin{align}
\begin{aligned}
P_e(\Lambda|\mathbf{H})
&\leq \mathds{1}\left\{H_{\max} > g_{\max}(n) \cup H_{\min} < g_{\min}(n) |\mathbf{H} \right\} \\
&+ \mathds{1}\left\{H_{\max} \leq g_{\max}(n) \cap H_{\min} \geq g_{\min}(n) |\mathbf{H} \right\} \hspace{-0.2em} \cdot \hspace{-0.2em} Pr\left\{\mathbf{z}\notin\text{Ball}(r) |\mathbf{H} \right\} \\
&+ \hspace{-0.6em} \sum_{\lambda\in\Lambda\setminus\{0\}} \hspace{-0.7em} \mathds{1} \hspace{-0.1em} \left\{H_{\max} \leq g_{\max}(n) \cap H_{\min} \geq g_{\min}(n) |\mathbf{H} \right\} \hspace{-0.2em} \cdot \hspace{-0.2em} Pr\left\{\mathbf{z}\in\text{Ball}(\mathbf{H}\cdot\lambda,r)\cap\text{Ball}(r)|\mathbf{H}\right\} \hspace{-0.2em},
\end{aligned}
\end{align}
where the first term is due to the cases where the fading channel exceeds the predefined finite support, the second term is due to the cases where the decoding ball is empty and the third term is due to the cases where it includes more than one receiver's codeword.
We can simplify the above conditional error probability upper bound by the following:
\begin{align}
\begin{aligned}
P_e(\Lambda|\mathbf{H})
&\leq \mathds{1}\left\{H_{\max} > g_{\max}(n) \cup H_{\min} < g_{\min}(n) |\mathbf{H} \right\}
 + Pr\left\{\mathbf{z}\notin\text{Ball}(r) |\mathbf{H} \right\} \\
&+ \hspace{-0.6em} \sum_{\lambda\in\Lambda\setminus\{0\}} \hspace{-0.7em} \mathds{1} \hspace{-0.1em} \left\{H_{\max} \leq g_{\max}(n) \cap H_{\min} \geq g_{\min}(n) |\mathbf{H} \right\} \hspace{-0.2em} \cdot \hspace{-0.2em} Pr\left\{\mathbf{z}\in\text{Ball}(\mathbf{H}\cdot\lambda,r)\cap\text{Ball}(r)|\mathbf{H}\right\} \hspace{-0.2em}.
\end{aligned}
\end{align}
By averaging over the fading distribution we can upper bound the error probability by the following:
\begin{align}
\label{align_lattice_upper_bound}
\begin{aligned}
P_e(\Lambda) &= E\{P_e(\Lambda|\mathbf{H})\} \\
   &\leq Pr\left\{H_{\max} > g_{\max}(n) \cup H_{\min} < g_{\min}(n) \right\} + Pr\left\{\mathbf{z}\notin\text{Ball}(r)\right\} \\
   &+ \hspace{-0.8em} \sum_{\lambda\in\Lambda\setminus\{0\}}{ \hspace{-0.8em} E\Big\{ \hspace{-0.1em} \mathds{1} \hspace{-0.15em} \left\{H_{\max} \leq g_{\max}(n) \cap H_{\min} \geq g_{\min}(n) |\mathbf{H} \right\} \hspace{-0.2em} \cdot \hspace{-0.2em} Pr \hspace{-0.2em} \left\{\mathbf{z} \hspace{-0.1em} \in \hspace{-0.1em} \text{Ball}(\mathbf{H}\cdot\lambda,r)\cap\text{Ball}(r)|\mathbf{H}\right\} \hspace{-0.3em} \Big\} } \hspace{-0.1em} .
\end{aligned}
\end{align}
\indent
Recall the Minkowski-Hlawka theorem \cite{Hlawka}\cite{Gruber}: Let $f:\mathbb{R}^n\rightarrow\mathbb{R}^+$ be a nonnegative integrable function with bounded support. Then for every $\gamma>0$, there exists a lattice $\Lambda$ with density $\gamma={\det(G_{\Lambda})}^{-1}$ (where $G_{\Lambda}$ is its generator matrix) that satisfies
\begin{equation}
\sum_{\lambda\in\Lambda\setminus\{0\}}f(\lambda) \leq
\gamma\int_{\mathbb{R}^n}f(\lambda)d\lambda.
\end{equation}
We now apply the Minkowski-Hlawka theorem to evaluate \eqref{align_lattice_upper_bound}.
Let us denote,
\begin{equation}
f(\lambda|\mathbf{H}) =  \mathds{1}\left\{H_{\max} \leq g_{\max}(n) \cap H_{\min} \geq g_{\min}(n) |\mathbf{H} \right\} \cdot Pr\left\{\mathbf{z}\in\text{Ball}(\mathbf{H}\cdot\lambda,r)\cap\text{Ball}(r)|\mathbf{H}\right\}\nonumber
\end{equation}
and choose $ f(\lambda) = E\left\{ f(\lambda|\mathbf{H}) \right\}$.
Note that $f(\lambda|\mathbf{H})=0$ for any $\lambda$ such that $\|\mathbf{H}\cdot\lambda\| > 2r$. The following proves that a sufficient condition for this is $\|\lambda\| > 2r_0\cdot\frac{g_{\max}(n)}{g_{\min}(n)}$ (which is not a function of $\mathbf{H}$):
\begin{align}
\begin{aligned}
\|\mathbf{H}\cdot\lambda\|
&\geq H_{\min} \cdot \|\lambda\| \\
&\geq g_{\min}(n) \cdot \|\lambda\| \\
&>    2r_0 \cdot g_{\max}(n) \\
&\geq 2r_0 \cdot H_{\max} \\
&\geq 2r_0 \cdot \sqrt[n]{\det(\mathbf{H})} = 2r,
\end{aligned}
\end{align}
where the second and the fourth inequalities are due to the fact that if there is a fading coefficient which is not in the range $[g_{\min}(n),g_{\max}(n)]$, then $f(\lambda|\mathbf{H}) = 0$ anyway. The third inequality is from the assumption. As an immediate consequence we get that $f(\lambda)$ also has a bounded support.
\newline
\indent
Combining all the above, there exists a lattice $\Lambda$ with average error probability (using the typicality decoder) that satisfies the following:
\begin{align}
\begin{aligned}
&P_e(\Lambda)
   \leq Pr\left\{H_{\max} > g_{\max}(n) \cup H_{\min} < g_{\min}(n) \right\} + Pr\left\{\|\mathbf{z}\|>r\right\} \\
   &+ \gamma \hspace{-0.05em} \int_{\mathbb{R}^n}{ \hspace{-0.45em} E\Big\{ \mathds{1}\left\{H_{\max} \leq g_{\max}(n) \cap H_{\min} \geq g_{\min}(n) |\mathbf{H} \right\} Pr\left\{\mathbf{z}\in\text{Ball}(\mathbf{H}\cdot\lambda,r)\cap\text{Ball}(r)|\mathbf{H}\right\}\Big\} }d\lambda \hspace{-0.05em}.
\end{aligned}
\end{align}
Trivially, we can simplify the above upper bound by replacing the indicator function $\mathds{1}\{\cdot\}$ by the value of one, which leads to the following after simple mathematical manipulations:
\begin{align}
\begin{aligned}
P_e(\Lambda)
   &\leq Pr\left\{H_{\max} > g_{\max}(n) \cup H_{\min} < g_{\min}(n) \right\} + Pr\left\{\|\mathbf{z}\|>r\right\}
   \\
   &+ E\left\{\frac{\gamma}{\det(\mathbf{H})}\int_{\mathbb{R}^n}Pr\left\{\mathbf{z}\in\text{Ball}(\lambda,r)\cap\text{Ball}(r)|\mathbf{H}\right\}d\lambda\right\}
\\&= Pr\left\{H_{\max} > g_{\max}(n) \cup H_{\min} < g_{\min}(n) \right\} + Pr\left\{\|\mathbf{z}\|>r\right\}
\\
      &+E\left\{\frac{\gamma}{\det(\mathbf{H})}\int_{\mathbb{R}^n}\int_{\text{Ball}(\lambda,r)\cap\text{Ball}(r)}{f_{\mathbf{Z}}(\mathbf{z})}d\mathbf{z}d\lambda\right\},
\end{aligned}
\end{align}
where $f_{\mathbf{Z}}(\mathbf{z})$ stands for the multivariate-normal distribution of the noise vector.
Finally, since $\text{Ball}(\lambda,r)\cap\text{Ball}(r) \subseteq \text{Ball}(\lambda,r)$ we obtain,
\begin{align}
\begin{aligned}
P_e(\Lambda) &\leq Pr\left\{H_{\max} > g_{\max}(n) \cup H_{\min} < g_{\min}(n) \right\} + Pr\left\{\|\mathbf{z}\|>r\right\}
\\    &+E\left\{\frac{\gamma}{\det(\mathbf{H})}\int_{\mathbb{R}^n}\int_{\text{Ball}(\lambda,r)}{f_{\mathbf{Z}}(\mathbf{z})}d\mathbf{z}d\lambda\right\}
\\ &= Pr\left\{H_{\max} > g_{\max}(n) \cup H_{\min} < g_{\min}(n) \right\} + Pr\left\{\|\mathbf{z}\|>r\right\}
\\    &+ E\left\{\frac{\gamma}{\det(\mathbf{H})}\int_{\text{Ball}(r)}\int_{\mathbb{R}^n}{f_{\mathbf{Z}}(\mathbf{z-\lambda})}d\lambda d\mathbf{z}\right\}
\\ &= Pr\left\{H_{\max} > g_{\max}(n) \cup H_{\min} < g_{\min}(n) \right\} + Pr\left\{\|\mathbf{z}\|>r\right\}
       + E\left\{\frac{\gamma V_n r^n}{\det(\mathbf{H})}\right\}
\\ &= Pr\left\{H_{\max} > g_{\max}(n) \cup H_{\min} < g_{\min}(n) \right\} + Pr\left\{\|\mathbf{z}\|>r\right\} + \gamma V_n r_0^n.
\end{aligned}
\end{align}
\end{proof}
It is interesting to observe the similarity between the \emph{Typicality decoder based bound} in \eqref{align_typicality_decoder_based_bound} and the \emph{Dependence testing bound} in \eqref{eq_DT_bound}. In both, the bound includes a sum of two probabilities, where the first is the probability that the correct codeword does not cross the decoding threshold, and the second is the probability that other codewords cross the threshold.

The following Lemma simplifies the typicality decoder bound of Theorem \ref{thm_typicality_decoder_based_bound} in a way that is sufficient for the extension of our main result to the special case of Lattices.
\begin{lem}
\label{lem_sufficient_typicality_decoder_based_bound}
(Sufficient typicality decoder based bound): For any $r=r_0\sqrt[n]{\det(\mathbf{H})}$, $r_0>0$, $\delta = \frac{1}{n}\ln(\gamma)$ and large enough $n$, there exist a positive constant $C>0$ and an $n$-dimensional lattice $\Lambda$ with NLD $\delta$ and average (maximal) error probability over the unconstrained fading channel with CSI at the receiver, which satisfy:
\begin{align}
\label{align_sufficient_typicality_decoder_based_bound}
\begin{aligned}
P_e(\Lambda) &\leq Pr\left\{\|\mathbf{z}\|>r\right\} + \gamma V_n r_0^n
               + \frac{C}{n^2}.
\end{aligned}
\end{align}
\end{lem}
\begin{proof}
See Appendix \ref{app_proof_of_the_sufficient_typicality_decoder_based_bound}.
\end{proof}
\begin{theorem}
\label{thm_lattices}
Let $\epsilon>0$ be a given, fixed, error probability. Denote by $\delta^*(n,\epsilon)$ the optimal NLD for which there exists an $n$-dimensional lattice with average (maximal) error probability at most $\epsilon$. Then, for any regular fading distribution of $H$, as $n$ grows,
\begin{equation}
\label{eq_dispersion_lattices}
\delta^*(n,\epsilon) = \delta^* - \sqrt{\frac{V}{n}}Q^{-1}(\epsilon) + O\left(\frac{\ln(n)}{n}\right),
\end{equation}
where,
\begin{align}
\delta^* &\triangleq E\left\{ \delta(H) \right\} = E\left\{ \frac{1}{2}\ln\left(\frac{H^2}{2\pi e\sigma^2}\right) \right\}
\\ V     &\triangleq \frac{1}{2} + Var(\delta(H)) = \frac{1}{2} + Var\left(\frac{1}{2}\ln(H^2)\right)
\end{align}
noting that
\begin{equation}
\delta(H) \triangleq \frac{1}{2}\ln\left(\frac{H^2}{2\pi e\sigma^2}\right).
\end{equation}
\end{theorem}
\begin{proof}
First, in Section \ref{sec_proof_of_converse_part} we have already proved the converse part for any IC, which includes the special case of lattices. Hence, we only need to prove the existence of a lattice with error probability $\epsilon$ and with NLD that satisfies the RHS of \eqref{eq_dispersion_lattices}.
For doing so, let us use Lemma \ref{lem_sufficient_typicality_decoder_based_bound} with $r^n=(\gamma V_n)^{-1}\frac{\epsilon}{\sqrt{n}}\det(\mathbf{H})$ and $\delta$ s.t. $Pr\left\{\|\mathbf{z}\|>r\right\}=\epsilon\left(1-\frac{1}{\sqrt{n}}\right)-C/n^2$ (for some positive constant $C$ and large enough $n$ s.t. the RHS is positive and the conditions of Lemma \ref{lem_sufficient_typicality_decoder_based_bound} hold). Hence, for large enough $n$ there exists a lattice with NLD $\delta$ and error probability not greater than $\epsilon$ such that:
\begin{align}
\epsilon\left(1-\frac{1}{\sqrt{n}}\right) - \frac{C}{n^2} &= Pr\left\{\|\mathbf{z}\|>r\right\}
\\  &=    Pr\left\{\ln\left(\|\mathbf{z}\|^2\right)>\ln(r^2)\right\}
\\  \label{align_expand_rhs}
    &=    Pr\left\{\frac{1}{\sqrt{2}}Y_n>\sqrt{n}\cdot\frac{1}{2}\ln\left(\frac{r^2}{n\sigma^2}\right)\right\},
\end{align}
where $Y_n\triangleq\sqrt{\frac{n}{2}}\ln\left(\frac{\|\mathbf{z}\|^2}{n\sigma^2}\right)$. Expanding the RHS in the argument of \eqref{align_expand_rhs}
\begin{align}
\frac{1}{2}\ln\left(\frac{r^2}{n\sigma^2}\right) &= -\frac{\ln(V_n)}{n} - \delta + \frac{1}{n}\sum_{i=1}^{n}\ln(H_i) + \frac{1}{2n}\ln\left(\frac{\epsilon^2}{n}\right) - \frac{1}{2}\ln(n\sigma^2)
\\ &= \delta^* - \delta + \sqrt{\frac{Var(\ln(H))}{n}}S_n + O\left(\frac{1}{n}\right),
\end{align}
where the last equality is due to Stirling's approximation for $V_n$, and the definition of $S_n \triangleq \sum_{i=1}^{n}\frac{\ln(H_i) - E\{\ln(H)\}}{\sqrt{nVar(\ln(H))}}$.
Combining all the above we obtain the following:
\begin{align}
Pr\left\{\zeta_n>\sqrt{n}\left(\delta^*-\delta+O\left(\frac{1}{n}\right)\right)\right\} = \epsilon\left(1-\frac{1}{\sqrt{n}}\right) - \frac{C}{n^2},
\end{align}
where, $\zeta_n \triangleq \frac{1}{\sqrt{2}}Y_n - \sqrt{Var(\ln(H))}S_n$. According to Lemmas \ref{lem_ln_chi2_pdf}, \ref{lem_Berry_Esseen} and \ref{lem_sum_of_almost_normal_RVs} we get the following:
\begin{equation}
Q\left(\sqrt{\frac{n}{V}}\left(\delta^*-\delta\right) + O\left(\frac{1}{\sqrt{n}}\right)\right) = \epsilon + O\left(\frac{1}{\sqrt{n}}\right),
\end{equation}
or equivalently by algebraic manipulations and first order Taylor's approximation,
\begin{equation}
\delta = \delta^* - \sqrt{\frac{V}{n}}Q^{-1}(\epsilon) + O\left(\frac{1}{n}\right).
\end{equation}
Because of the symmetric structure of lattices, our achievability result holds also in the stronger sense of maximal error probability.
\end{proof}
\section{Volume to Noise Ratio Analysis}
\label{sec_VNR}
The analogous term for the SNR for lattices is the VNR (\emph{Volume to Noise Ratio}). Ingber et al. extended the definition of the VNR in \cite{Ingber}, to any IC $S$ over the unconstrained AWGN channel.
In a similar way, let define the VNR of IC $S$, over the unconstrained fading channel, as the ratio between the highest noise variance that is tolerable for the given NLD $\delta$ of $S$, and the actual noise variance $\sigma^2$. Therefore, the VNR $\mu$, is given by:
\begin{equation}
\mu = \frac{e^{-2\delta+E\{\ln(H^2)\}}}{2\pi e\sigma^2} = e^{2(\delta^*-\delta)}.
\end{equation}
Clearly, $\mu=1$ for a capacity achieving IC, and otherwise $\mu>1$.
Inspired by \cite{Ingber}, let define also the VNR as function of the IC $S$ and the error probability $\epsilon$, over the unconstrained fading channel, by the following:
\begin{equation}
\mu(S,\epsilon) = \frac{e^{-2\delta(S)+E\{\ln(H^2)\}}}{2\pi e\sigma^2(\epsilon)},
\end{equation}
where $\sigma^2(\epsilon)$ is the noise variance such that the error probability of $S$ is exactly $\epsilon$. In the same manner, let denote by $\mu^*(n,\epsilon)$, the lowest $\mu(S,\epsilon)$ for a given error probability $\epsilon$, over all the $n$-dimensional IC's. The rate of convergence of $\mu^*(n,\epsilon) \to 1$, when $n$ tends to infinity, is given by the following theorem.
\begin{theorem}
Let $\epsilon>0$ be a given, fixed, error probability. Denote by $\mu^*(n,\epsilon)$ the optimal (minimal) VNR for which there exists an $n$-dimensional infinite constellation with average error probability at most $\epsilon$. Then, for any regular fading distribution of $H$, as $n$ grows,
\begin{equation}
\mu^*(n,\epsilon) = 1 + \sqrt{\frac{2 + Var(\ln(H^2))}{n}}Q^{-1}(\epsilon) + O\left(\frac{\ln(n)}{n}\right).
\end{equation}
\end{theorem}
\begin{proof}
From the definitions of $\mu^*(n,\epsilon)$ and $\delta^*(n,\epsilon)$, then:
\begin{align}
\mu^*(n,\epsilon) &= e^{2(\delta^*-\delta^*(n,\epsilon))}
\\                &= e^{\sqrt{\frac{4V}{n}}Q^{-1}(\epsilon) + O\left(\frac{\ln(n)}{n}\right)},
\end{align}
where the last equality is due to theorem \ref{thm_main_result}, and $V = \frac{1}{2}+Var\left(\frac{1}{2}\ln(H^2)\right)$.
Finally, for large enough $n$, we can use the first order Taylor's approximation of $e^x$ around zero, to get the desired result:
\begin{equation}
\mu^*(n,\epsilon) = 1 + \sqrt{\frac{2 + Var(\ln(H^2))}{n}}Q^{-1}(\epsilon) + O\left(\frac{\ln(n)}{n}\right).
\end{equation}
\end{proof}
\section{Relation to the Power Constrained Model}
\label{sec_relation_to_the_power_constrained_model}
The error exponent at rates near the capacity can be approximated by a parabola of the form
\begin{equation}
\label{eq_error_exponent_near_capacity}
E\left(R\right) \approx \frac{\left(C-R\right)^2}{2V},
\end{equation}
where $V$ is the channel dispersion. This fact was already known to Shannon (see \cite[Figure 18]{Polyanskiy}). By taking uniform input distribution in Gallager's random coding error exponent, precisely $X \sim U\left(-\frac{a}{2},\frac{a}{2}\right)$, over the power constrained fast fading channel with available CSI at the receiver, it can be shown (see Appendix \ref{app_scalar_error_exponent_uniform_prior_sec}) that \eqref{eq_error_exponent_near_capacity} holds with $C = E\left\{\frac{1}{2}\ln\left( \frac{a^2H^2}{2\pi e\sigma^2} \right)\right\}$ and $V = \frac{1}{2} + Var\left(\frac{1}{2}\ln\left(H^2\right)\right)$, when $a/\sigma$ tends to infinity (the high SNR regime). Since the unconstrained setting can be thought of as the limit of the power constrained setting, when the SNR tends to infinity, this result hints that $\delta^* = E\left\{ \frac{1}{2}\ln\left(\frac{H^2}{2\pi e\sigma^2}\right) \right\}$ and $V = \frac{1}{2} + Var\left(\frac{1}{2}\ln\left(H^2\right)\right)$, in that setting.

In \cite{PolyanskiyFading} Polyanskiy et al. studied the dispersion of the general case of power constrained stationary fading channels.
In case of fast fading channels with power constraint $P$, and AWGN variance $\sigma^2$, this dispersion (in $\text{nats}^2$ per channel use) is given by
\begin{equation}
\label{eq_dispesion_fading_power_constrained}
V = Var\left(\frac{1}{2}\ln\left(1+SNR\cdot H^2\right)\right) + \frac{1}{2}\left(1 - E^2\left\{\frac{1}{1+SNR\cdot H^2}\right\}\right),
\end{equation}
where $SNR \triangleq P/\sigma^2$. Another indication to the channel dispersion value in the unconstrained case, is given by taking the limit of \eqref{eq_dispesion_fading_power_constrained}, when the SNR tends to infinity. In the high SNR regime \eqref{eq_dispesion_fading_power_constrained} can be approximated by
\begin{align}
V &\approx \frac{1}{2} + Var\left(\frac{1}{2}\ln\left(SNR\cdot H^2\right)\right)
\\&= \frac{1}{2} + Var\left(\frac{1}{2}\ln\left(H^2\right)\right),
\end{align}
which coincides with the previous hint to the channel dispersion value in the unconstrained setting. The case of unconstrained stationary fading channels with memory, will be discussed later on in Section \ref{sec_fading_channels_with_memory}. We will see there a similar relations to the power constrained fading channels, as in the case of fast fading channels.

It should be noted that while the dispersion analysis accuracy of power constrained fading channels in \cite{PolyanskiyFading} is $o\left(\frac{1}{\sqrt{n}}\right)$, in our analysis the accuracy is slightly better, $O\left(\frac{\ln(n)}{n}\right)$. This faster convergence might be due to the fact that in \cite{PolyanskiyFading} a more general fading model was analyzed.

In Figure \ref{fig_dispersion_Vs_snr} we can see the power constrained channel dispersion rate of convergence to the unconstrained channel dispersion limit, with growing SNRs, at the popular Rayleigh fading channel.
\begin{figure}[htp]
\center{\includegraphics[width=0.7\columnwidth]{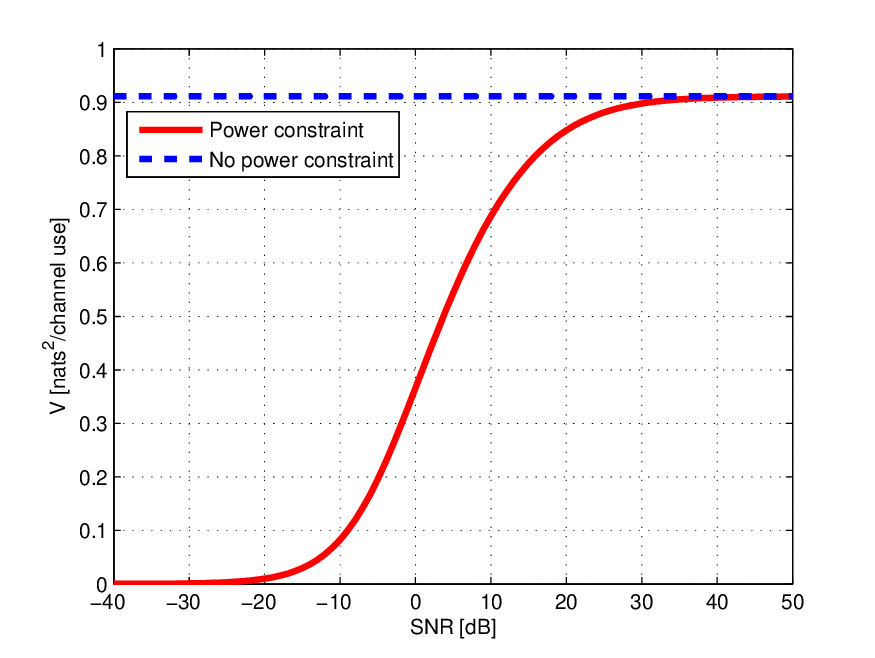}}
\caption{\label{fig_dispersion_Vs_snr} The power-constrained Rayleigh fast fading channel dispersion vs. the unconstrained channel dispersion.}
\end{figure}
\section{Comparison to the AWGN Channel}
\label{sec_comparison_to_the_awgn_channel}
Let's start with the comparison of the unconstrained fast fading channel to the AWGN channel, in terms of Poltyrev's capacity. By Jensen's inequality and the concavity of the logarithm function, we can derive the following result:
\begin{equation}
\delta^* \triangleq E\left\{ \frac{1}{2}\ln\left(\frac{H^2}{2\pi e\sigma^2}\right) \right\}
\leq \frac{1}{2}\ln\left( E\left\{\frac{H^2}{2\pi e\sigma^2}\right\} \right)
= \frac{1}{2}\ln\left(\frac{1}{2\pi e\sigma^2}\right) = \delta^*_{\text{AWGN}}.
\end{equation}
This proves that in the AWGN channel the Poltyrev's capacity is greater than its equivalent in the fast fading channel (with the same noise variance $\sigma^2$).
In Section \ref{sec_fading_channels_with_memory}, we will see that the Poltyrev's capacity, in stationary fading processes, is not affected by the dynamics of the channel. Hence, this result also holds for stationary fading processes.

This loss, relative to the AWGN channel, is given exactly by $-E\left\{\ln(H)\right\}$ in nats per channel use. Alternatively, this loss can be measured as the ratio between the highest noise variance that is tolerable in each channel model. It is easy to show that this ratio is given by $e^{-2E\left\{\ln(H)\right\}}$ in linear scale, or by $-8.6859E\left\{\ln(H)\right\}$ in dB. For example, this loss equals approximately 0.288 nats per channel use, or 2.5 dB, in the Rayleigh fading channel.

For the comparison in terms of channel dispersion, notice that according to \cite{Ingber}, the unconstrained AWGN channel dispersion is given by $V_{\text{AWGN}}=\frac{1}{2}$.
Hence, we can get the following inequality for fast fading channel dispersion:
\begin{equation}
V = \frac{1}{2} + Var\left(\frac{1}{2}\ln(H^2)\right) \geq V_{\text{AWGN}}. \nonumber
\end{equation}
In Section \ref{sec_fading_channels_with_memory}, we will prove that the inequality, $V\geq V_{\text{AWGN}}$, also holds for stationary fading processes.
This fact shows that there is another loss relative to the AWGN channel in the setting of fixed error probability and finite block length.
For example, in Rayleigh fast fading channel with $\epsilon=10^{-5}$ and $n=100$, there is another loss of approximately 0.92 dB.

\begin{figure}[htp]
\center{\includegraphics[width=0.7\columnwidth]{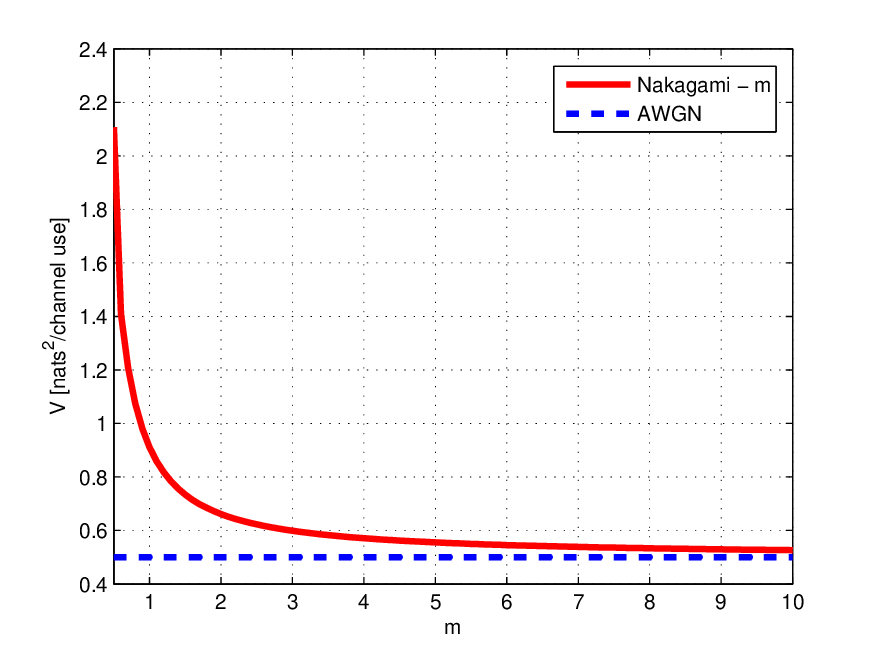}}
\caption{\label{fig_Nakagami_m_dispersion} The IC's channel dispersion of the Nakagami-$m$ fading channel converges to the channel dispersion of the AWGN channel.}
\end{figure}

In Figure \ref{fig_Nakagami_m_dispersion} we can see the unconstrained channel dispersion of the Nakagami-$m$ fading, for various values of $m$.
As we have already seen in Chapter \ref{chapter:BasicDefinitions}, this popular family of fading distributions are given by:
\begin{equation}
f_m(h) = \frac{2m^m}{\Gamma(m)}h^{2m-1}e^{-mh^2},~h\geq0,~m\geq\frac{1}{2}.\nonumber
\end{equation}
It can be seen that when $m\rightarrow\infty$ the dispersion converges from above to the unconstrained AWGN channel dispersion $\frac{1}{2}$, as expected (since in that case, the Nakagami-$m$ distribution converges to the $H=1$ with probability one).
\section{Fading Channels with Memory}
\label{sec_fading_channels_with_memory}
In this section we extend our main result, of dispersion analysis for IC's over fast fading channels, to the general case of stationary fading processes.
Loosely speaking, we will show that if the memory of the fading process decays fast enough, then the dispersion analysis holds, but with a channel dispersion $V$ that depends on the dynamics of the fading process. We will call such a process a \emph{weakly dependent process}.

Let $H_1, H_2, \dots$ be a narrow-sense stationary sequence of RV's. In the following we define rigorously three types of such a \emph{weakly dependent processes}.
\begin{defn}
(Strong mixing): If the sequence $\{H_i\}_{i=1}^{\infty}$ satisfies as $n \to \infty$,
\begin{equation}
\alpha_H(n)=\sup_{A,B}|P(A,B)-P(A)P(B)| \to 0,
\end{equation}
where the supremum is over all RV's $A \in \mathfrak{M}_{-\infty}^{k}$ and $B \in \mathfrak{M}_{k+n}^{\infty}$ ($\mathfrak{M}_{a}^{b}$ denotes the $\sigma$-algebra generated by the RV's $H_i$ when $i\in[a,b]$).
\end{defn}
\begin{defn}
(Complete regular): If the sequence $\{H_i\}_{i=1}^{\infty}$ satisfies as $n \to \infty$,
\begin{equation}
\rho_H(n)=\sup_{f,g}\frac{\left|Corr\left(f(\dots,H_{k-1},H_k),g(H_{k+n},H_{k+n+1},\dots)\right)\right|}{\sqrt{Var\left(f(\dots,H_{k-1},H_k)\right) \cdot Var\left(g(H_{k+n},H_{k+n+1},\dots)\right)}} \to 0,
\end{equation}
where the supremum is over all the functions $f$ and $g$ which are measurable w.r.t. the $\sigma$-algebras $\mathfrak{M}_{-\infty}^{k}$ and $\mathfrak{M}_{k+n}^{\infty}$.
\end{defn}
\begin{defn}
(m-dependent): If the sequence $\{H_i\}_{i=1}^{\infty}$ satisfies for any two vectors of the form $(H_{a-p},H_{a-p+1},\dots,H_{a-1},H_{a})$ and $(H_{b},H_{b+1},\dots,H_{b+q})$ are independent for $b-a>m$.
\end{defn}
\noindent
Roughly speaking, under some other restrictions, the distribution of the sum
\begin{equation}
S_n \triangleq \frac{\sum_{i=1}^{n}(H_i-E\{H\})}{\sqrt{Var(\sum_{i=1}^{n}H_i)}},\nonumber
\end{equation}
for weakly dependent processes, converges uniformly to the normal distribution. Hence, we can apply a similar analysis as we done for IC's over fast fading channels, to get the dispersion analysis of stationary weakly dependent fading processes.

\begin{lem}
\label{lem_Tikhomirov}
(Tikhomirov) If the narrow-sense stationary process $H_1,H_2,\dots$ is a strong mixing (complete regular) such that for some positive constants K and $\beta$
\begin{equation}
\alpha_H(n) \leq Ke^{-\beta n}~ \left(\rho_{H}(n) \leq Ke^{-\beta n}\right)
\end{equation}
and
\begin{equation}
E\left\{\left|X_1 - E\left\{X_1\right\}\right|^3\right\} < \infty.
\end{equation}
Then,
\begin{equation}
\sigma^2 \triangleq \lim_{n\to\infty}\frac{1}{n}Var\left(\sum_{i=1}^{n}X_i\right) = S_X\left(e^{j\omega}\right)\big|_{\omega=0} = R_X(0) + 2\sum_{k=1}^{\infty}R_X(k)
\end{equation}
and if $\sigma^2>0$ for any $-\infty<s<\infty$
\begin{equation}
\left|Pr\left\{ \frac{\sum_{i=1}^{n}\left( X_i - E\{X_i\} \right)}{\sqrt{n}\sigma} \leq s \right\} - F_{N(0,1)}(s) \right| \leq \frac{A\ln^2(n)}{\sqrt{n}},
\end{equation}
where the process $X_1,X_2,\dots$ is given by $X_i = \frac{1}{2}\ln(H_i^2)$, and its auto-correlation and PSD (power spectral density) are given, respectively, by
\begin{equation}
R_X(k) \triangleq E\{(X_{k+1}-E\{X_{k+1}\})(X_{1}-E\{X_{1}\})\}
\end{equation}
\begin{equation}
S_X\left(e^{j\omega}\right)\triangleq\sum_{k=-\infty}^{\infty}R_X(k)e^{-j \omega k}
\end{equation}
and $A$ is some positive constant.
\end{lem}
\begin{proof}
Clearly, $X_1,X_2,\dots$ is a narrow-sense stationary process. Moreover, since $X_i$ is a function of $H_i$ we obtain
\begin{equation}
\alpha_X(n) \leq \alpha_H(n) \leq Ke^{-\beta n} ~ \left(\rho_X(n) \leq \rho_H(n) \leq Ke^{-\beta n}\right).
\end{equation}
Hence, by using \cite[Theorems 1,2,3]{Tikhomirov} with $\delta=1$ we complete the proof of the lemma.
\end{proof}
The previous lemma will serve the same purpose as the Berry-Esseen lemma does in the proof of our main result in case of fast fading channels.

\begin{theorem}
\label{thm_dispersion_of_wealy_dependent_processes}
Let $\epsilon>0$ be a given, fixed, error probability. Denote by $\delta^*(n,\epsilon)$ the optimal NLD for which there exists an $n$-dimensional infinite constellation with average error probability at most $\epsilon$. Then, for any strong mixing (complete regular) narrow-sense stationary fading process $H_1,H_2,\dots$, such that
\begin{enumerate}
  \item $E\{H_i^2\} = 1$,
  \item The marginal distribution of $H_i$ is a regular fading distribution,
  \item $\alpha_H(n) \leq Ke^{-\beta n}~ \left(\rho_{H}(n) \leq Ke^{-\beta n}\right)$ for some positive constants K and $\beta$,
  \item $E\left\{\left|\frac{1}{2}\ln(H_1^2) - E\left\{\frac{1}{2}\ln(H_1^2)\right\}\right|^3\right\} < \infty$,
\end{enumerate}
as $n$ grows,
\begin{equation}
\delta^*(n,\epsilon) = \delta^* - \sqrt{\frac{V}{n}}Q^{-1}(\epsilon) + O\left(\frac{\ln^2(n)}{n}\right),
\end{equation}
where,
\begin{equation}
\delta^* = E\left\{ \frac{1}{2}\ln\left(\frac{H^2}{2\pi e\sigma^2}\right) \right\}
\end{equation}
and
\begin{align}
\label{align_channel_dispersion_with_memory}
V  &= \frac{1}{2} + \lim_{n\to\infty}\frac{1}{n}Var\left( \sum_{i=1}^{n}\frac{1}{2}\ln\left(H_i^2\right) \right)
\\ &= \frac{1}{2} + S_{{\frac{1}{2}\ln(H^2)}}\left(e^{j\omega}\right)\big|_{\omega=0}
\\ &= \frac{1}{2} + Var\left(\frac{1}{2}\ln(H_1^2)\right) + 2\sum_{k=1}^{\infty}R_{\frac{1}{2}\ln(H^2)}(k).
\end{align}
\end{theorem}
\begin{proof}
The proof is very similar to the proof of our main result in case of fast fading channels, except that, instead of the Berry-Esseen lemma for i.i.d. RV's, we use Tikhomirov lemma for weakly dependent processes.

In the direct part, we choose a uniform input distribution within a cube $\text{Cb}(a)$, and then by the dependence testing bound of Theorem \ref{thm_DT_bound}, for large enough $a(n)$ and Lemma \ref{lem_Tikhomirov}, we prove that there exists a finite cube constellation, with average error probability upper bounded by $\epsilon$, that holds the following:
\begin{equation}
\delta(n,\epsilon,a/\sigma) = \delta^* - \sqrt{\frac{V}{n}}Q^{-1}(\epsilon) + O\left(\frac{\ln^2(n)}{n}\right),
\end{equation}
where $\delta(n,\epsilon,a/\sigma)$ is the NLD of the finite cube constellation within $\text{Cb}(a)$. Finally, by the tiling operation of this finite constellation to the whole space $\mathbb{R}^n$, we complete the proof of the direct part.

In the converse part, using the sphere packing bound of Theorem \ref{thm_sphere_packing_lower_bound}, Lemma \ref{lem_ln_chi2_pdf}, Lemma \ref{lem_Tikhomirov} and a similar arguments as in Lemma \ref{lem_sum_of_almost_normal_RVs}, we prove that
\begin{equation}
\delta^*(n,\epsilon) \leq \delta^* - \sqrt{\frac{V}{n}}Q^{-1}(\epsilon) + \frac{1}{2n}\ln(n) + O\left(\frac{\ln^2(n)}{n}\right),
\end{equation}
which completes the proof of the converse part.
\end{proof}

Note that according to Theorem \ref{thm_dispersion_of_wealy_dependent_processes}, the channel dispersion is affected by the fading dynamics, this is in contrary to the Poltyrev's capacity, which is independent of this dynamics \cite{FadingChannels}\cite{PolyanskiyFading}. In \cite{PolyanskiyFading}, it was shown, that the channel dispersion of power constrained stationary fading processes, is given by:
\begin{equation}
V = \lim_{n\to\infty}\frac{1}{n}Var\left( \sum_{i=1}^{n}\frac{1}{2}\ln\left(1 + SNR \cdot H_i^2\right) \right) + \frac{1}{2}\left(1 - E^2\left\{\frac{1}{1+SNR \cdot H^2}\right\}\right).
\end{equation}
Hence, the limit of the power constrained channel dispersion, when $SNR\to\infty$, equals to the unconstrained channel dispersion \eqref{align_channel_dispersion_with_memory}, in the general case of stationary fading processes, as we have already seen in the special case of fast fading channels. Moreover, since $V_{\text{AWGN}}=\frac{1}{2}$ \cite{Ingber}, then it is obvious according to \eqref{align_channel_dispersion_with_memory} that $V \geq V_{\text{AWGN}}$, for stationary fading processes (as we also have already seen in fast fading channels).

\begin{cor}
\label{cor_Gausssian_ARMA_dispersion}
Let $\epsilon>0$ be a given, fixed, average error probability.
If the process $H_1,H_2,\dots$ is a finite-order auto-regressive moving average (ARMA) Gaussian process, then as $n$ grows,
\begin{equation}
\delta^*(n,\epsilon) = \delta^* - \sqrt{\frac{V}{n}}Q^{-1}(\epsilon) + O\left(\frac{\ln^2(n)}{n}\right),
\end{equation}
where,
\begin{equation}
\delta^* = E\left\{ \frac{1}{2}\ln\left(\frac{H^2}{2\pi e\sigma^2}\right) \right\}
\end{equation}
and
\begin{equation}
V = \frac{1}{2} + S_{{\frac{1}{2}\ln(H^2)}}\left(e^{j\omega}\right)\big|_{\omega=0} = \frac{1}{2} + Var\left(\frac{1}{2}\ln(H_1^2)\right) + 2\sum_{k=1}^{\infty}R_{\frac{1}{2}\ln(H^2)}(k).
\end{equation}
\end{cor}
\begin{proof}
According to \cite{Kolmogorov} if $H_1,H_2,\dots$ is a stationary Gaussian process with a PSD $S_H(e^{j\omega})$, which is rational w.r.t. $e^{i\omega}$, then $\alpha_H(n)$ decreases exponentially. Hence, if the process is an ARMA Gaussian process, then $\alpha_H(n)$ decreases exponentially. Moreover, when e.g. $H_1 \sim N(0,1)$, then the marginal distribution of $|H_1|$ is a regular fading distribution, and in addition
\begin{equation}
E\left\{\left|\frac{1}{2}\ln(H_1^2) - E\left\{\frac{1}{2}\ln(H_1^2)\right\}\right|^3\right\} \approx 2.9486 < \infty.\nonumber
\end{equation}
Therefore, all the conditions of Theorem \ref{thm_dispersion_of_wealy_dependent_processes} are satisfied.
\end{proof}
For an illustrative example of Corollary \ref{cor_Gausssian_ARMA_dispersion}, let define the Gaussian $AR(1)$ fading process with the parameter-$a$, by the following:
\begin{equation}
H_i = aH_{i-1} + W_i,~ W_i \sim N(0, 1-a^2),\nonumber
\end{equation}
where $|a| < 1$ and $W_i$ is a white process. In Figure \ref{fig_Gaussian_AR_dispersion} we can see its channel dispersion as function of the parameter $a$. Since the coherence time of the process increases with $a$, we can observe that the channel dispersion also increases, as $a$ grows.

A fading process, such that $\sum_{k=1}^{\infty}R_{\frac{1}{2}\ln(H^2)}(k)>0$, will be called fading process with ``positive correlation''. Clearly, the Gaussian $AR(1)$, is an example of such a fading process. The usage of random interleaver in practical systems with finite block-length, over such fading processes, seems very beneficial, in order to get effectively a fast fading channel, with smaller channel dispersion.

\begin{figure}[htp]
\center{\includegraphics[width=0.7\columnwidth]{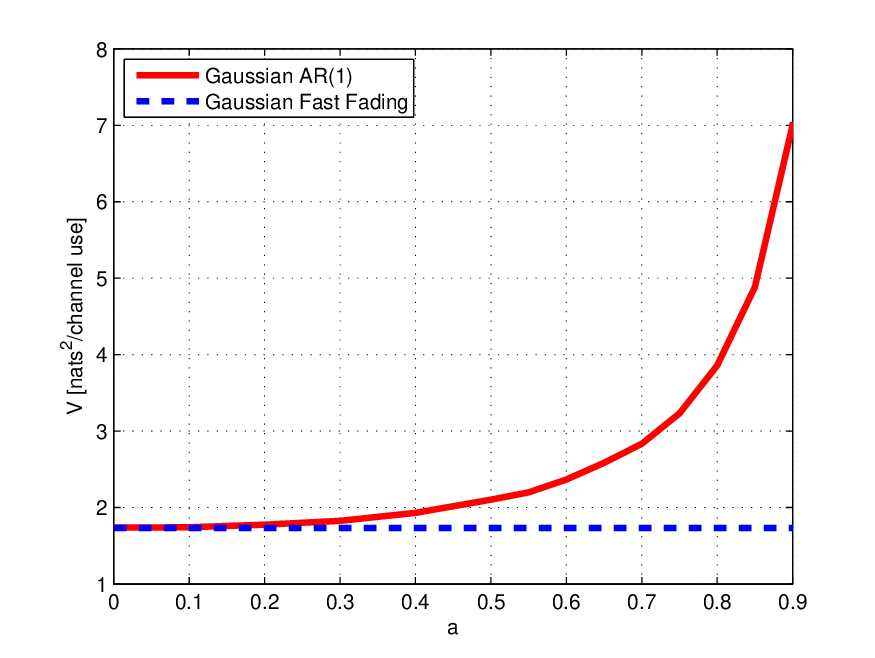}}
\caption{\label{fig_Gaussian_AR_dispersion} The dispersion of Gaussian AR(1) process fading as function of the parameter-$a$.}
\end{figure}

Finally, note that in \cite{Tikhomirov}, we can find theorems with more relaxed conditions on the dependency of the process, which also guarantee a uniform convergence to the normal distribution, but with a greater error than the guaranteed error of Lemma \ref{lem_Tikhomirov}. On the other hand, \cite[Theorem 5]{Tikhomirov} discuss the stronger dependency condition of \emph{m-dependent} processes, and shows that the convergence rate is $O\left(\frac{1}{\sqrt{n}}\right)$ in that case. Hence, for moving average (MA) processes that generated by i.i.d. white noise, for example, we can get the dispersion result of Theorem \ref{thm_dispersion_of_wealy_dependent_processes} with the accuracy of $O\left(\frac{\ln(n)}{n}\right)$. Moreover, in \cite{PolyanskiyFading} Polyanskiy et al. derived the dispersion analysis of the power constrained weakly dependent processes, with much relaxed conditions than those of Theorem \ref{thm_dispersion_of_wealy_dependent_processes}, but at the cost of accuracy which is only $o\left(\frac{1}{\sqrt{n}}\right)$.

\chapter{Dispersion of Infinite Constellations in MIMO Fading Channels}
\label{chapter:DispersionOfICInMIMOFadingChannels}
\renewcommand{\thefootnote}{} 
\footnotetext[1]{The material in this chapter was partially presented in \cite{Eilat}.}
\renewcommand{\thefootnote}{\arabic{footnote}} 
In this chapter we analyze the dispersion of infinite constellation in MIMO fast fading channels without power constraint and under the constraint of \emph{Full Dimensional Transmission} (\emph{FDT}). In Section \ref{sec_mimo_main_result} we present our main result, whose converse and direct parts are proven in Sections \ref{sec_mimo_converse_part} and \ref{sec_mimo_direct_part}, respectively.
Later on, in Section \ref{sec_derivation_of_simple_expressions_for_V_and_delta} we derive an extremely simple expressions for Poltyrev's capacity and channel dispersion in MIMO fading channels under the \emph{FDT} constraint. Relation to the power constrained MIMO fading channel and comparison to the unconstrained independent parallel channels model are also discussed in Sections \ref{sec_relation_to_the_MIMO_power_constrained_model} and \ref{sec_comparison_to_the_parallel_channels_model}, respectively. Finally, in Section \ref{sec_generalization} we discuss the general case of IC's in MIMO fast fading channels without any constraint.
\section{Main Result - FDT's Dispersion}
\label{sec_mimo_main_result}
Assume a transmission of an $l=nt$ complex dimensional IC over the unconstrained $t \times r$ MIMO model ($t \leq r$) using $n$ channel uses. Let us call such a transmission \emph{Full Dimensional Transmission} (\emph{FDT}).
In this section we present the dispersion analysis of MIMO channels under the constraint of \emph{FDT}.
\begin{theorem}[MIMO dispersion under the FDT constraint]
\label{thm_mimo_main_result}
Let $\epsilon>0$ be a given, fixed, error probability. Denote by $\delta^*(n,\epsilon)$ the optimal NLD for which there exists an $l = n \cdot t$ complex dimensional infinite constellation with average error probability at most $\epsilon$, over the $t \times r$ MIMO model ($t \leq r$). Then, as $n$ grows,
\begin{equation}
\label{eq_mimo_main_result}
\delta^*(n,\epsilon) = \delta^* - \sqrt{\frac{V}{n}}Q^{-1}(\epsilon) + O\left(\frac{\ln(n)}{n}\right),
\end{equation}
\noindent
where,
\begin{equation}
\delta^* \triangleq E\left\{\ln\left({\det\left(\frac{\mathbf{H}^{\dagger}\mathbf{H}}{\pi e\sigma^2}\right)}\right) \right\}
\end{equation}
\noindent
and
\begin{equation}
V \triangleq t + Var\left( \ln\left(\det(\mathbf{H}^{\dagger}\mathbf{H})\right) \right).
\end{equation}
\end{theorem}
The converse and the direct parts of the proof of this theorem are given in Sections \ref{sec_mimo_converse_part} and \ref{sec_mimo_direct_part}, respectively.
\begin{cor}
The highest achievable NLD with arbitrary small error probability, namely the Poltyrev's capacity, over the $t \times r$ MIMO model ($t \leq r$) without power constraint and under the FDT constraint, with available CSI at the receiver, is given by
\begin{equation}
\delta^* \triangleq E\left\{\ln\left({\det\left(\frac{\mathbf{H}^{\dagger}\mathbf{H}}{\pi e\sigma^2}\right)}\right) \right\}.
\end{equation}
\end{cor}
\begin{proof}
By taking the limit $n\rightarrow\infty$ in \eqref{eq_mimo_main_result} we get the desired result (for any $0 < \epsilon < 1$).
\end{proof}
\section{Converse Part}
\label{sec_mimo_converse_part}
In this section we prove the converse part of Theorem \ref{thm_mimo_main_result}. The converse part is based on normal approximation of the \emph{sphere packing lower bound} on the average error probability under the \emph{FDT} constraint.
The sphere packing lower bound of IC's over MIMO fading channels under the \emph{FDT} constraint is presented in Section \ref{sec_mimo_sphere_packing_bound}, and in Section \ref{sec_mimo_proof_of_converse_part} we complete the proof by a derivation of an appropriate normal approximation technique.
\subsection{The Sphere Packing Bound}
\label{sec_mimo_sphere_packing_bound}
In this section we give a sketch of proof for the following \emph{sphere packing bound}, for IC's over the MIMO fading channel under the \emph{FDT} constraint.
Note that the following theorem, is an extension of Theorem \ref{thm_sphere_packing_lower_bound}.
\begin{theorem}
\label{thm_sphere_packing_lower_bound_mimo}
For any IC $S$ with NLD $\delta$, over the $t \times r$ MIMO channel ($t \leq r$) under the \emph{FDT} constraint, the average error probability is lower bounded by the following sphere packing bound:
\begin{equation}
\label{eq_general_SPB_mimo}
P_e\left(S\right) \geq Pr\left\{ \hspace{-0.05cm} \left\|\mathbf{z}'^n\right\|^2 \geq e^{-\frac{\delta}{t}}\left(\frac{\det(\mathbf{H}^{n\dagger}\mathbf{H}^n)}{V_{2nt}}\right)^{\frac{1}{nt}} \hspace{-0.05cm} \right\}.
\end{equation}
\end{theorem}
\begin{proof}
Assume a transmission of an $l = nt$ complex dimensional IC over the $t \times r$ MIMO channel using $n$ channel uses.
For IC where all the Voronoi cells have equal volume $V_{\text{tr}}$, such as lattices, in the receiver given the CSI, we get an IC with Voronoi cell volume that equals $V_{\text{rc}} = V_{\text{tr}}\cdot\det(\mathbf{H}^{n\dagger}\mathbf{H}^n)$.
By the \emph{equivalent sphere} argument \cite{Poltyrev}\cite{Tarokh}, the probability that the noise leaves the Voronoi cell in the receiver is lower bounded by the probability to leave a sphere of the same volume:
\begin{equation}
\label{eq_SPB_mimo}
P_e\left(S\right) \geq Pr\left\{\left\|\mathbf{z}'^n\right\|^2 \geq r_{\text{eff}}^{2}(\mathbf{H}^n) \right\},
\end{equation}
where $V_{2nt} \cdot r_{\text{eff}}^{2nt}(\mathbf{H}^n) \triangleq V_{\text{rc}}$ and $V_l = \frac{\pi^{l/2}}{\frac{l}{2}\Gamma\left(\frac{l}{2}\right)}$.
Combining \eqref{eq_SPB_mimo} with the definition of $\delta = -\frac{\ln(V_{\text{tr}})}{n}$ leads to \eqref{eq_general_SPB_mimo}.

To complete the proof of the converse part, we need to prove that \eqref{eq_general_SPB_mimo} holds for any $l = nt$ complex dimensional IC.
This includes regular IC's with bounded Voronoi's cells and also non-regular IC's, such as IC's with unbounded Voronoi's cells and IC's with density which oscillates with the cube size $a$ (i.e. only the limsup exists in the definition of $\gamma$).
The proof in the case of regular IC's, can be done by applying the \emph{equivalent sphere} argument for any codeword's Voronoi's cell volume given the CSI, and using the Jensen's inequality and the convexity of the obtained lower bound, exactly as done in Lemma \ref{lem_sphere_packing_lower_bound_for_regular_ICs}.
The extension to the case of non-regular IC's, can be done by a very similar regularization process as done in Lemma \ref{lem_regularixation} (for proving Theorem \ref{thm_sphere_packing_lower_bound}), for the received IC's over the MIMO channel.
\end{proof}
\subsection{Proof of Converse Part}
\label{sec_mimo_proof_of_converse_part}
Assume a transmission of IC $S$ with NLD $\delta$, over the MIMO channel under the \emph{FDT} constraint.
Let us define,
\begin{equation}
\zeta_n \triangleq \sqrt{t} \cdot Y_n - \sqrt{Var(\ln(\det(\mathbf{H}^{\dagger}\mathbf{H})))} \cdot S_n,
\end{equation}
where, $Y_n \triangleq \sqrt{nt}\cdot\big( \ln\big(\|\mathbf{z}'^n\|^2\big) - \ln\left(nt\sigma^2\right) \big)$, $S_n \triangleq \frac{\sum_{i=1}^{n}X_i}{\sqrt{n}}$, and $X_i \triangleq \frac{\ln(\det(\mathbf{H}_i^{\dagger}\mathbf{H}_i))-E\{\ln(\det(\mathbf{H}^{\dagger}\mathbf{H}))\}}{\sqrt{Var(\ln(\det(\mathbf{H}^{\dagger}\mathbf{H})))}}$.
Then, by taking the logarithm and rearranging the inequality in the argument of \eqref{eq_general_SPB_mimo}, we obtain:
\begin{equation}
P_e \geq Pr\left\{ \zeta_n \geq \zeta \right\},
\end{equation}
where, $\zeta \triangleq \sqrt{n}\big( \delta^* - \delta + t\ln\left(\frac{\pi e}{nt}\right) - \frac{\ln(V_{2nt})}{n} \big)$.

In a similar way as we done in the case of scalar fading channels, the combination of Lemma \ref{lem_ln_chi2_pdf}, Lemma \ref{lem_Berry_Esseen} and Lemma \ref{lem_sum_of_almost_normal_RVs}, proves that the distribution of $\zeta_n$ is asymptotically normal distribution, with zero mean and variance $V$. Hence,
\begin{equation}
\label{eq_SPB_step4_mimo}
P_e \geq Q\left(\frac{\zeta}{\sqrt{V}}\right) - O\left(\frac{1}{\sqrt{n}}\right).
\end{equation}
By Stirling approximation for the Gamma function, $V_{2nt}$ can be approximated as
\begin{equation}
\frac{\ln(V_{2nt})}{nt} = \ln\left(\frac{\pi e}{nt}\right) - \frac{1}{2nt}\ln(n) + O\left(\frac{1}{n}\right)
\end{equation}
and hence we get:
\begin{equation}
\label{eq_zeta_mimo}
\zeta = \sqrt{n}\left( \delta^* - \delta + \frac{1}{2n}\ln(n) + O\left(\frac{1}{n}\right) \right).
\end{equation}
The assignment of \eqref{eq_zeta_mimo} in \eqref{eq_SPB_step4_mimo} gives us:
\begin{equation}
\label{eq_SPB_step5_mimo}
\hspace{-1.5cm}\epsilon \geq P_e \geq
Q\left(\frac{ \delta^* - \delta + \frac{1}{2n}\ln(n) + O\left(\frac{1}{n}\right) }{ \sqrt{ \frac{V}{n} } }\right) \hspace{-0.1cm} - \hspace{-0.05cm} O\left(\frac{1}{\sqrt{n}}\right) \hspace{-0.1cm}. \hspace{-0.8cm}
\end{equation}
Taking $Q^{-1}(\cdot)$ from both sides of \eqref{eq_SPB_step5_mimo} and using the following Taylor approximation, $  Q^{-1}\left(\epsilon + O\left(\frac{1}{\sqrt{n}}\right)\right) = Q^{-1}(\epsilon) + O\left(\frac{1}{\sqrt{n}}\right)$, gives us the desired result:
\begin{align}
\begin{aligned}
\delta &\leq \delta^* - \sqrt{ \frac{V}{n} }Q^{-1}\left(\epsilon\right) + \frac{1}{2n}\ln(n) + O\left(\frac{1}{n}\right),
\end{aligned}
\end{align}
which completes the proof of the converse part.
\section{Direct Part}
\label{sec_mimo_direct_part}
In this section we prove the direct part of Theorem \ref{thm_mimo_main_result}. The direct part is based on normal approximation of the \emph{Dependence Testing upper bound} on the average error probability.
The dependence testing upper bound over MIMO fading channels is presented in Section \ref{sec_mimo_dependence_testing_bound}, and in Section \ref{sec_mimo_proof_of_direct_part} we complete the proof by a derivation of an appropriate normal approximation technique.
\subsection{Dependence Testing Bound}
\label{sec_mimo_dependence_testing_bound}
In this section we present an extension of Polyanskiy's \emph{Dependence Testing Bound} to the case of MIMO fast fading channels with available CSI at the receiver. In \cite{Polyanskiy} the DT bound was used to prove the dispersion analysis for DMCs, or more precisely, for memoryless channels without a power constraint (or any other constraint on the channel input). Here, the channel input does not have any restriction, and hence we can use the DT bound to prove the direct part of our main result.
\begin{theorem}
\label{thm_DT_bound_mimo}
(DT bound) For any input distribution $f_{\mathbf{x}}(\cdot)$ on $\mathbb{C}^t$, there exists a code with $M$ codewords and an average error probability over the MIMO fast fading channel, with available CSI at the receiver, not exceeding
\begin{align}
\begin{aligned}
P_e &\leq E\left\{e^{-\left[i(\mathbf{x}^n;\mathbf{y}'^n,\mathbf{H}^n) - \ln\left(\frac{M-1}{2}\right)\right]^+}\right\}
\\  &= Pr\left\{ i\left(\mathbf{x}^n;\mathbf{y}'^n,\mathbf{H}^n\right) \leq \ln\left(\frac{M-1}{2}\right) \right\}
\\  &+ \frac{M-1}{2} E\left\{ e^{-i\left(\mathbf{x}^n;\mathbf{y}'^n,\mathbf{H}^n\right)}1_{\left\{i\left(\mathbf{x}^n;\mathbf{y}'^n,\mathbf{H}^n\right)>\ln\left(\frac{M-1}{2}\right)\right\}} \right\},
\end{aligned}
\end{align}
where $f_{\mathbf{x}^n\mathbf{y}'^n\mathbf{H}^n}({x},{y'},{h}) = f_{\mathbf{x}^n}({x})f_{\mathbf{y}'^n|\mathbf{x}^n,\mathbf{H}^n}({y}'|{x},{h})f_{\mathbf{H}^n}({h})$ is the joint PDF of all the random vectors and matrices arising above, $f_{\mathbf{x}^n}({x})=\Pi_{i=1}^nf_{\mathbf{x}}(x_i)$ and $i({x};{y},{h}) \triangleq \ln\left( \frac{f_{\mathbf{x}^n\mathbf{y}'^n\mathbf{H}^n}({x}, {y}', {h})}{f_{\mathbf{x}^n}({x})f_{\mathbf{y}'^n\mathbf{H}^n}({y}, {h})} \right)$.
\end{theorem}
\begin{proof}
A trivial extension of Theorem \ref{thm_DT_bound}.
\end{proof}
\subsection{Proof of Direct Part}
\label{sec_mimo_proof_of_direct_part}
For the proof of the direct part, we will first construct an ensemble of finite constellations with $M$ codewords, which are uniformly distributed in an $l = n \cdot t$ complex dimensional cube $\text{Cb}(a,l)$, for some fixed $a$, $n$ and $t \leq r$. Then, using the \emph{Dependence Testing bound} of Theorem \ref{thm_DT_bound_mimo} with $f_{\mathbf{x}}(x) = \frac{1_{\{x\in \text{Cb}(a,t)\}}}{a^{2t}}$, we will find a lower bound on the optimal achievable number of codewords, for a FC in such an ensemble, whose error probability is upper bounded by some fixed $\epsilon > 0$. We will denote this lower bound by $M(n,\epsilon,a/\sigma)$. Theorem \ref{thm_DT_bound_mimo} also ensures the existence of such a FC that achieves this lower bound. Finally, we will construct an IC by tiling this FC to the whole space $\mathbb{C}^{l}$, in a way that will preserve the density of codewords and the error probability, asymptotically in number of the channel uses $n$, as in the FC.

To use the DT bound of Theorem \ref{thm_DT_bound_mimo}, we need to prove that for some $\gamma$ the following inequality holds:
\begin{align}
\label{align_DT_bound_2}
\begin{aligned}
P_e &\leq Pr\left\{ i\left(\mathbf{x}^n;\mathbf{y}'^n,\mathbf{H}^n\right) \leq \ln(\gamma) \right\}
\\  &+ \gamma E\left\{ e^{-i\left(\mathbf{x}^n;\mathbf{y}'^n,\mathbf{H}^n\right)}1_{\left\{i\left(\mathbf{x}^n;\mathbf{y}'^n,\mathbf{H}^n\right)>\ln(\gamma)\right\}} \right\} \leq \epsilon.
\end{aligned}
\end{align}
Denote for arbitrary $\tau$
\begin{equation}
\ln(\gamma) = nI(\mathbf{x};\mathbf{y}',\mathbf{H}) - \tau\sqrt{nVar(i(\mathbf{x};\mathbf{y}',\mathbf{H}))}.
\end{equation}
The information density is a sum of $n$ i.i.d. RVs:
\begin{equation}
i\left(\mathbf{x}^n;\mathbf{y}'^n,\mathbf{H}^n\right) = \sum_{j=1}^{n}{ i(\mathbf{x}_j;\mathbf{y}'_j,\mathbf{H}_j) },
\end{equation}
\noindent
where $i(\mathbf{x};\mathbf{y}',\mathbf{H}) \triangleq \ln\left( \frac{f(\mathbf{y}'|\mathbf{H},\mathbf{x})}{f(\mathbf{y}'|\mathbf{H})} \right)$ and its moments are given by the following lemma.
\begin{lem}
\label{lem_information_density_moments_mimo}
(Information density's moments) If $\mathbf{x}$ is distributed uniformly in $\emph{Cb}(a,t)$, then for large enough $a/\sigma$ the moments of the information density $i(\mathbf{x};\mathbf{y}',\mathbf{H})$ are given by:
\begin{enumerate}
  \item $I(\mathbf{x};\mathbf{y}',\mathbf{H}) \triangleq E\{i(\mathbf{x};\mathbf{y}',\mathbf{H})\} =
E\left\{\ln\left({\det\left(\frac{a^2\mathbf{H}^{\dagger}\mathbf{H}}{\pi e\sigma^2}\right)}\right) \right\} + O\left(\left(\frac{\sigma}{a}\right)^{2t}\right)$
  \item $Var(i(\mathbf{x};\mathbf{y}',\mathbf{H})) = t + Var\left( \ln\left(\det(\mathbf{H}^{\dagger}\mathbf{H})\right) \right) + O\left(\left(\frac{\sigma}{a}\right)^{2t}\right)$
  \item $\rho_3 \triangleq E\left\{|i(\mathbf{x};\mathbf{y}',\mathbf{H}) - I(\mathbf{x};\mathbf{y}',\mathbf{H})|^3\right\} < \infty$.
\end{enumerate}
\end{lem}
\begin{proof}
It is easy to show that the PDF of $\mathbf{y}'$ given $\mathbf{H}$ is given by
\begin{align}
\begin{aligned}
f(\mathbf{y}'|\mathbf{H}) &= \int_{\mathbf{x} \in E}{f(\mathbf{y}'|\mathbf{x},\mathbf{H})}d\mathbf{x}
\\                        &= \frac{1}{a^{2t}\det(\mathbf{H}^{\dagger}\mathbf{H})}\int_{\mathbf{x} \in E}{\frac{1}{(\pi\sigma^2)^t}e^{-\frac{\|\mathbf{y}'-\mathbf{x}\|^2}{\sigma^2}}}d\mathbf{x}
\\                        &= \frac{\sigma^{2t}}{a^{2t}\det(\mathbf{H}^{\dagger}\mathbf{H})}\int_{\mathbf{x} \in E/\sigma}{\frac{1}{\pi^t}e^{-\|\frac{\mathbf{y}'}{\sigma}-\mathbf{x}\|^2}}d\mathbf{x}
\end{aligned}
\end{align}
where $E \triangleq \mathbf{D}'\mathbf{V}^{\dagger}\text{Cb}(a,t)$ and $\mathbf{D}' = \text{diag}(\lambda_1^{\frac{1}{2}},\dots,\lambda_t^{\frac{1}{2}})$.
Since $\text{Ball}(\lambda^{\frac{1}{2}}_{\min}a/2) \subseteq E \subseteq \text{Ball}(\lambda^{\frac{1}{2}}_{\max}a/2)$ where $\lambda_{\min}\triangleq\min(\lambda_1,\dots,\lambda_t)$ and $\lambda_{\max}\triangleq\max(\lambda_1,\dots,\lambda_t)$, then
\begin{align}
\begin{aligned}
f(\mathbf{y}'|\mathbf{H}) &\leq C_U'\cdot\frac{\sigma^{2t}e^{-\frac{\|\mathbf{y}'\|^2}{\sigma^2}}}{a^{2t}\det(\mathbf{H}^{\dagger}\mathbf{H})} \int_{0}^{\lambda^{\frac{1}{2}}_{\max}\frac{a}{2\sigma}}{r^{2t-1}e^{-\left(r^2+\frac{2r\|\mathbf{y}'\|}{\sigma}\right)}}dr
\\&\leq C_U'\cdot\frac{\sigma^{2t}e^{-{\frac{\|\mathbf{y}'\|^2}{\sigma^2}}  }}{a^{2t}\det(\mathbf{H}^{\dagger}\mathbf{H})}\int_{0}^{\infty}{r^{2t-1}e^{-r^2}}dr
\\&= C_U\cdot\frac{\sigma^{2t}e^{-{\frac{\|\mathbf{y}'\|^2}{\sigma^2}}  }}{a^{2t}\det(\mathbf{H}^{\dagger}\mathbf{H})}
\triangleq f_{U}(\mathbf{y}'|\mathbf{H})
\end{aligned}
\end{align}
and for $a/\sigma\geq2$
\begin{align}
\begin{aligned}
f(\mathbf{y}'|\mathbf{H}) &\geq
C_L'\cdot\frac{\sigma^{2t} e^{-\frac{\|\mathbf{y}'\|^2}{\sigma^2}}}{a^{2t}\det(\mathbf{H}^{\dagger}\mathbf{H})}\int_{0}^{\lambda^{\frac{1}{2}}_{\min}\frac{a}{2\sigma}}{r^{2t-1}
e^{-\left(r^2+\frac{2r\|\mathbf{y}'\|}{\sigma}\right)}}dr
\\&\geq
C_L'\cdot\frac{\sigma^{2t} e^{-\frac{\|\mathbf{y}'\|^2}{\sigma^2}}}{a^{2t}\det(\mathbf{H}^{\dagger}\mathbf{H})}\int_{0}^{\lambda^{\frac{1}{2}}_{\min}}{r^{2t-1}
e^{-\big(\lambda_{\min}+\frac{2\lambda^{\frac{1}{2}}_{\min}\|\mathbf{y}'\|}{\sigma}\big)}}dr
\\&= C_L\cdot\frac{\sigma^{2t}\lambda^{t}_{\min}e^{-\big({\frac{\|\mathbf{y}'\|}{\sigma}} + \sqrt{\lambda_{\min}}\big)^2 }}{a^{2t}\det(\mathbf{H}^{\dagger}\mathbf{H})}
\triangleq f_{L}(\mathbf{y}'|\mathbf{H}),
\end{aligned}
\end{align}
for some positive constants $C_U$ and $C_L$.
By straight forward algebraic manipulations over the definition of $i(\mathbf{x};\mathbf{y}',\mathbf{H})$, we obtain
\begin{equation}
i(\mathbf{x};\mathbf{y}',\mathbf{H}) = \ln\left( \det\left( \frac{a^2\mathbf{H}^{\dagger}\mathbf{H}}{\pi e \sigma^2} \right) \right) - \frac{\|\mathbf{z}'\|-t\sigma^2}{\sigma^2} + e_{a/\sigma}\left(\mathbf{y}',\mathbf{H}\right)
\end{equation}
where the error random variable is given by,
\begin{equation}
e_{a/\sigma}\left(\mathbf{y}',\mathbf{H}\right) \triangleq -\ln\left(\frac{a^{2t}\det(\mathbf{H}^{\dagger}\mathbf{H})}{\sigma^{2t}}f(\mathbf{y}'|\mathbf{H})\right)\geq0.
\end{equation}
The conditional expectation of the error random variable, given $\mathbf{H}$, is given by
\begin{align}
\begin{aligned}
&e_{a/\sigma}\left(\mathbf{H}\right) \triangleq E\{e_{a/\sigma}\left(\mathbf{y}',\mathbf{H}\right)|\mathbf{H}\} = \int_{\mathbf{y}'\in \mathbb{C}^{t}}{f(\mathbf{y}'|\mathbf{H})e_{a/\sigma}\left(\mathbf{y}',\mathbf{H}\right)}d\mathbf{y}'
\\ &\leq -\int_{\mathbf{y}'\in \mathbb{C}^{t}}{f_U(\mathbf{y}'|\mathbf{H})\ln\left(\frac{a^{2t}\det(\mathbf{H}^{\dagger}\mathbf{H})}{\sigma^{2t}}f_L(\mathbf{y}'|\mathbf{H})\right)}d\mathbf{y}'
\\ &= \frac{\sigma^{2t}}{a^{2t}\det(\mathbf{H}^{\dagger}\mathbf{H})}\cdot\left(c_0 + c_1\ln(\lambda_{\min}) + c_2\lambda^{\frac{1}{2}}_{\min} + c_3\lambda_{\min}\right).
\end{aligned}
\end{align}
Finally, the error's expectation, is given by
\begin{align}
\begin{aligned}
e_{a/\sigma} \triangleq E\{e_a\left(\mathbf{H}\right)\} = O\left(\left(\frac{\sigma}{a}\right)^{2t}\right).
\end{aligned}
\end{align}
Hence,
\begin{equation}
I(\mathbf{x};\mathbf{y}',\mathbf{H}) = E\left\{\ln\left({\det\left(\frac{a^2\mathbf{H}^{\dagger}\mathbf{H}}{\pi e\sigma^2}\right)}\right) \right\} + O\left(\left(\frac{\sigma}{a}\right)^{2t}\right).
\end{equation}
In a similar way we can calculate the variance and also to bound the third absolute moment.
\end{proof}
According to the Berry-Essen lemma (see Lemma \ref{lem_Berry_Esseen}) for i.i.d. RVs,
\begin{equation}
\label{eq_using_Berry_Esseen_UB_mimo}
|Pr\{ i\left(\mathbf{x}^n;\mathbf{y}'^n,\mathbf{H}^n\right) \leq \ln\gamma \} - Q(\tau)| \leq \frac{B(a/\sigma)}{\sqrt{n}}
\end{equation}
where $B(a/\sigma) = \frac{6\rho_3}{Var^{\frac{3}{2}}(i(\mathbf{x};\mathbf{y}',\mathbf{H}))}$.

\noindent
For sufficiently large $n$, let
\begin{equation}
\tau = Q^{-1}\left(\epsilon - \left(\frac{2\ln(2)}{\sqrt{2\pi Var(i(\mathbf{x};\mathbf{y}',\mathbf{H}))}} + 5B(a/\sigma)\right)\frac{1}{\sqrt{n}} \right).
\end{equation}
Then, from \eqref{eq_using_Berry_Esseen_UB_mimo} we obtain
\begin{align}
\label{align_part1_UB_for_DT_bound}
\begin{aligned}
&Pr\left\{ i\left(\mathbf{x}^n;\mathbf{y}'^n,\mathbf{H}^n\right) \leq \ln(\gamma) \right\} \leq
\\ & \epsilon - 2\left(\frac{\ln(2)}{\sqrt{2\pi Var(i(\mathbf{x};\mathbf{y}',\mathbf{H}))}} + 2B(a/\sigma)\right)\frac{1}{\sqrt{n}}.
\end{aligned}
\end{align}
Using Lemma \ref{lem_47_in_Polyanskiy} (see in Appendix \ref{app_lemma_47_in_Polyanskiy}), we get
\begin{align}
\label{align_part2_UB_for_DT_bound}
\begin{aligned}
&\gamma E\left\{ e^{-i\left(\mathbf{x}^n;\mathbf{y}'^n,\mathbf{H}^n\right)}1_{\left\{i\left(\mathbf{x}^n;\mathbf{y}'^n,\mathbf{H}^n\right)>\ln(\gamma)\right\}} \right\} \leq \\ &2\left(\frac{\ln(2)}{\sqrt{2\pi Var(i(\mathbf{x};\mathbf{y}',\mathbf{H}))}} + 2B(a/\sigma)\right)\frac{1}{\sqrt{n}}.
\end{aligned}
\end{align}
Summing \eqref{align_part1_UB_for_DT_bound} and \eqref{align_part2_UB_for_DT_bound} we prove the inequality \eqref{align_DT_bound_2}. Hence, by Theorem \ref{thm_DT_bound_mimo}, there exists a FC with $M(n,\epsilon,a/\sigma)$ codewords, denoted by $S(n, \epsilon, a/\sigma)$, such that
\begin{align}
\label{align_Direct_step_1_mimo}
\begin{aligned}
&\ln \left(M(n,\epsilon,a/\sigma)\right) = \ln(\gamma) + O(1)
\\ &= nI(\mathbf{x};\mathbf{y}',\mathbf{H}) - \tau\sqrt{nVar(i(\mathbf{x};\mathbf{y}',\mathbf{H}))}  + O(1)
\\ &= nI(\mathbf{x};\mathbf{y}',\mathbf{H}) - \sqrt{nVar(i(\mathbf{x};\mathbf{y}',\mathbf{H}))}Q^{-1}(\epsilon) + O(1),
\end{aligned} \hspace{-1cm}
\end{align}
where the last equality is derived by a first order Taylor's approximation for $Q^{-1}\left(\epsilon + O\left(\frac{1}{\sqrt{n}}\right)\right)$ around $\epsilon$.
Let us define the NLD of the FC in $\text{Cb}(a,l)$ by
\begin{equation}
\label{eq_NLD_in_FC}
\delta(n, \epsilon, a/\sigma) \triangleq \frac{1}{n}\ln\left(\frac{M(n,\epsilon,a/\sigma)}{a^{2nt}}\right).
\end{equation}
From \eqref{align_Direct_step_1_mimo} we obtain
\begin{align}
\begin{aligned}
\delta(n, \epsilon, a/\sigma) &= I(\mathbf{x};\mathbf{y}',\mathbf{H}) - \ln(a^{2t})
\\ &- \sqrt{\frac{Var(i(\mathbf{x};\mathbf{y}',\mathbf{H}))}{n}}Q^{-1}(\epsilon) + O\left(\frac{1}{n}\right).\nonumber
\end{aligned}
\end{align}
Note that the results of Lemma \ref{lem_information_density_moments_mimo} hold in general for large enough $a$. Specifically, we can choose $a$ to be a monotonic increasing function of $n$ s.t. $\lim_{n \to \infty}a=\infty$, and then the results of Lemma \ref{lem_information_density_moments_mimo} will hold for any large enough $n$. Assigning the results of Lemma \ref{lem_information_density_moments_mimo} with appropriate choice of $a = a(n)$, we get
\begin{align}
\begin{aligned}
\delta(n, \epsilon, a/\sigma) &= \delta^* - \sqrt{ \frac{V + O\left(\left(\frac{\sigma}{a}\right)^{2t}\right)}{n} }Q^{-1}(\epsilon) + O\left(\frac{1}{n} + \left(\frac{\sigma}{a}\right)^{2t}\right)
\\&= \delta^* - \sqrt{ \frac{V}{n} }Q^{-1}(\epsilon) + O\left(\frac{1}{n} + \left(\frac{\sigma}{a}\right)^{2t}\right)
,\nonumber
\end{aligned}
\end{align}
where the last equality is derived by Taylor approximation for large enough $n$.
By tiling the FC, denoted by $S(n, \epsilon, a/\sigma)$, to the whole space $\mathbb{C}^{l}$ (in a similar way as done in Appendix \ref{app_tiling}, for scalar fading channels), we can construct an IC with average error probability which is upper bounded by $\epsilon$, and NLD $\delta(n,\epsilon)$ that satisfies
\begin{equation}
\label{eq_Direct_step_5}
\delta(n, \epsilon) = \delta^* - \sqrt{ \frac{V}{n} }Q^{-1}(\epsilon) + O\left(\frac{1}{n}\right).
\end{equation}
\noindent
Hence, the optimal NLD necessarily satisfies $\delta^*(n,\epsilon) \geq \delta(n, \epsilon)$. This completes the proof of the direct part.
\section{Derivation of simple expressions for $V$ and $\delta^*$}
\label{sec_derivation_of_simple_expressions_for_V_and_delta}
From Theorem \ref{thm_mimo_main_result}, the Poltyrev's capacity and the channel dispersion in MIMO fading channels under the \emph{FDT} constraint are given by the following:
\begin{align}
\delta^* &= E\left\{ \ln\left(\det\left(\frac{\mathbf{H}^{\dagger}\mathbf{H}}{\pi e\sigma^2}\right)\right) \right\}
\\V     &= t + Var(\ln(\det(\mathbf{H}^{\dagger}\mathbf{H}))).
\end{align}
Although at first glance, the evaluation of $\delta^*$ and $V$ seems an extremely difficult problem, in this section we derive a very simple expressions for them. These simplified expressions involve only summation operations that depend on the basic model parameters $t, r$ and $\sigma^2$.
This derivation is based on the distribution of the random variable $W \triangleq \ln\left(\det\left(\mathbf{H}^{\dagger}\mathbf{H}\right)\right)$, which appears in $\delta^*$ and $V$. This distribution can be derived immediately, by the following lemma.
\begin{lem}[The determinant's logarithm distribution]
\label{lem_det_log_dist}
The random variable
\begin{equation}
\hat{W} \triangleq \ln\left(\det\left(\mathbf{H}^{\dagger}\mathbf{H}\right)\right) + \ln(2^t),\nonumber
\end{equation}
where $\mathbf{H} \in \mathbb{C}^{r \times t}$ is a random matrix with entries that are distributed as i.i.d. circular symmetric $CN(0,1)$ random variables and  $t \leq r$, is distributed as is the sum of t independent $\chi^2$ random variables with $2r,2(r-1),\dots,2(r-t+1)$ degrees of freedom respectively.
\end{lem}
\begin{proof}
Follows directly from \cite[Theorem 1.1]{Goodman}.
\end{proof}
In the next lemma we will derive a simple analytic expressions for the expectation and the variance of $W$.
\begin{lem}[The determinant's logarithm moments]
\label{lem_det_log_moments}
The expectation and the variance of the random variable $W \triangleq \ln\left(\det\left(\mathbf{H}^{\dagger}\mathbf{H}\right)\right)$, where $\mathbf{H} \in \mathbb{C}^{r \times t}$ is a random matrix with entries that are distributed as i.i.d. circular symmetric $CN(0,1)$ random variables and $t \leq r$, are given by:
\begin{align}
E\{W\} &= -\gamma t + 1 - t + t\sum_{p=1}^{r-t}\frac{1}{p} + r\sum_{p=r-t+1}^{r-1}\frac{1}{p}
\\
Var(W) &= \frac{\pi^2t}{6} - t\sum_{p=1}^{r-t}\frac{1}{p^2} - \sum_{p=r-t+1}^{r-1}\frac{r-p}{p^2},
\end{align}
where $\gamma=0.577\dots$ is the Euler's constant.
\end{lem}
\begin{proof}
Using the result of Lemma \ref{lem_det_log_dist} we get immediately that the expectation and the variance of $W$ are given by:
\begin{equation}
E\{W\} = \sum_{i=1}^{t}E\left\{\ln\left(\frac{X_i}{2}\right)\right\}
\end{equation}
and
\begin{equation}
Var(W) = \sum_{i=1}^{t}Var\left(\ln\left(X_i\right)\right),
\end{equation}
where $X_i\sim\chi^2_{2(r-i+1)}$ for $i=1,2,\dots,t$.
Moreover, it is known that $E\left\{\ln\left(\frac{X_i}{2}\right)\right\} = \psi(i)$ and $Var\left(\ln\left(X_i\right)\right) = \psi'(i)$, where $\psi(x) \triangleq \frac{d}{dx}\ln(\Gamma(x))$ is the \emph{digamma} function. From the \emph{digamma} function properties we have that $\psi(x) = -\gamma + \sum_{p=1}^{x-1}\frac{1}{p}$ for integer $x$, and $\psi'(x) = \sum_{p=1}^{\infty}\frac{1}{(p+x-1)^2}$ for any $x$ (see for example \cite{Oyman}\footnote{Note that in \cite{Oyman} $\chi^2_{n}$ was defined as the distribution of the sum of squares of $n$ i.i.d. $N(0,\frac{1}{2})$ RVs, and not of $N(0,1)$ as commonly used.}).
Combining the above we get that
\begin{align}
E\{W\} &= \sum_{i=1}^{t}\psi(r-i+1)
\\     &= -\gamma t + \sum_{i=1}^{t}\sum_{p=1}^{r-i}\frac{1}{p}
\\     &= -\gamma t + \sum_{p=1}^{r-t}\sum_{i=1}^{t}\frac{1}{p} + \sum_{p=r-t+1}^{r-1}\sum_{i=1}^{r-p}\frac{1}{p}
\\     &= -\gamma t + 1 - t + t\sum_{p=1}^{r-t}\frac{1}{p} + r\sum_{p=r-t+1}^{r-1}\frac{1}{p}
\end{align}
and
\begin{align}
Var(W) &= \sum_{i=1}^{t}\psi'(r-i+1)
\\     &= \sum_{i=1}^{t}\sum_{p=1}^{\infty}\frac{1}{(p+r-i)^2}
\\     &= \sum_{i=1}^{t}\left(\sum_{p=1}^{\infty}\frac{1}{p^2} - \sum_{p=1}^{r-i}\frac{1}{p^2}\right)
\\     &= \frac{\pi^2t}{6} - \sum_{i=1}^{t}\sum_{p=1}^{r-i}\frac{1}{p^2}
\\     &= \frac{\pi^2t}{6} - t\sum_{p=1}^{r-t}\frac{1}{p^2} - \sum_{p=r-t+1}^{r-1}\frac{r-p}{p^2}
\end{align}
where, $\sum_{p=1}^{\infty}\frac{1}{p^2} = \frac{\pi^2}{6}$ is the known solution for the \emph{Basel problem} (see for example \cite{Basel}).
\end{proof}
By Lemma \ref{lem_det_log_moments} we can derive simple analytic expressions for $\delta^*$ and $V$ in the MIMO fading channel, which are summarized by the following theorem.
\begin{theorem}
\label{thm_MIMO_channel_delta_and_V}
The Poltyrev's capacity $\delta^*$ and the channel dispersion $V$ of the $t \times r$ MIMO fast fading channel ($t \leq r$) under the \emph{FDT} constraint, are given by:
\begin{align}
\delta^* &= -\gamma t + 1 - t + t\sum_{p=1}^{r-t}\frac{1}{p} + r\sum_{p=r-t+1}^{r-1}\frac{1}{p} - t\ln(\pi e \sigma^2)
\\
\label{align_analytic_V_for_mimo}
V &= t + \frac{\pi^2t}{6} - t\sum_{p=1}^{r-t}\frac{1}{p^2} - \sum_{p=r-t+1}^{r-1}\frac{r-p}{p^2}.
\end{align}
\end{theorem}
\begin{proof}
Follows directly from the fact that
\begin{align}
\delta^* &= E\left\{ \ln\left(\det\left(\frac{\mathbf{H}^{\dagger}\mathbf{H}}{\pi e\sigma^2}\right)\right) \right\}
\\V     &= t + Var(\ln(\det(\mathbf{H}^{\dagger}\mathbf{H})))
\end{align}
and from Lemma \ref{lem_det_log_moments}.
\end{proof}
In Figures \ref{fig_delta_star_vs_r} and \ref{fig_channel_dispersion_vs_r} we demonstrate this result for different number of transmit and receive antennas. It can be observed in Figure \ref{fig_channel_dispersion_vs_r}, that for fixed number of transmit antennas $t$, the channel dispersion decreases as the number of the receiver antennas grows. In addition, this dispersion converges to $t$, when $r\to\infty$ (clearly, from \eqref{align_analytic_V_for_mimo} and the solution of \emph{Basel problem} \cite{Basel}). Note that $t$ is the channel dispersion of $t$ parallel, identical and independent complex AWGN channels. This hints us that increasing the number of receive antennas whitens the MIMO fading channel.
\begin{figure}[htp]
\center{\includegraphics[width=0.7\columnwidth]{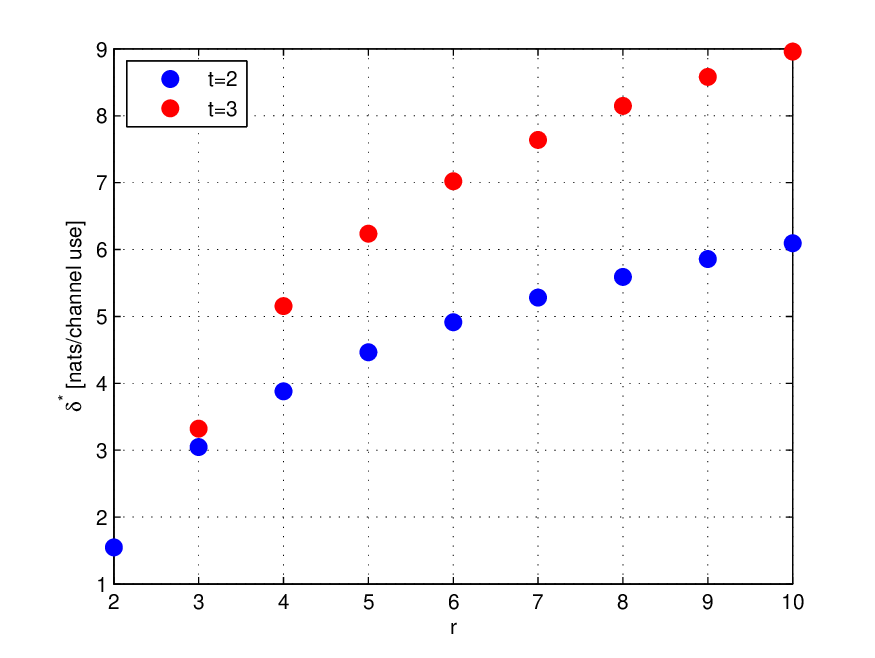}}
\caption{\label{fig_delta_star_vs_r} Poltyrev's capacity under the \emph{FDT} constraint vs. the number of receive antennas $r$, for fixed number of transmit antennas $t$ and noise variance $\sigma^2=0.05$.}
\end{figure}
\begin{figure}[htp]
\center{\includegraphics[width=0.7\columnwidth]{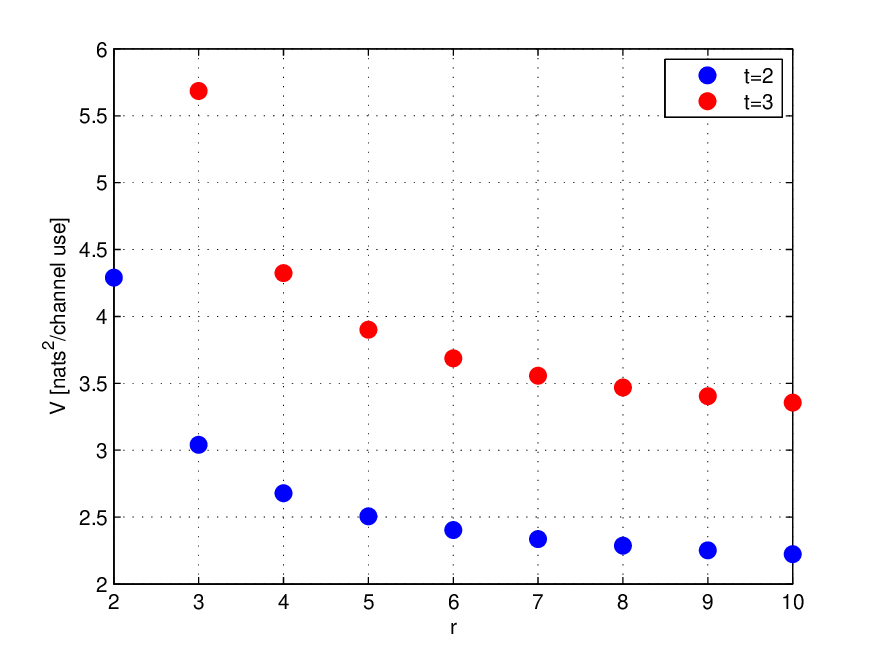}}
\caption{\label{fig_channel_dispersion_vs_r} The channel dispersion under the \emph{FDT} constraint vs. the number of receive antennas $r$, for fixed number of transmit antennas $t$.}
\end{figure}

Note that in \cite{Oyman} (also presented here in Section \ref{sec_moments_of_the_information_density}) a similar analysis provided approximations for the capacity and the variance of the mutual information given the CSI, in the high SNR regime of the power constrained MIMO channel with normal input distribution. Moreover, in \cite[Theorem 2]{Telatar} (also presented here in Section \ref{sec_capacity_and_error_exponent} Theorem \ref{thm_Telatar_capacity}) Telatar derived an easy to evaluate (numerically) one dimensional integral expression for the capacity of the power constrained MIMO channel. Here, the evaluation of $\delta^*$ and $V$ are not only exact in contrast to the results in \cite{Oyman}, but also much easier to evaluate than the capacity in \cite{Telatar}, and only involve summation operations as function of the basic parameters $t, r$ and $\sigma^2$.
\section{Relation to the Power Constrained Model}
\label{sec_relation_to_the_MIMO_power_constrained_model}
As we already mentioned in Section \ref{sec_relation_to_the_power_constrained_model}, the error exponent at rates near the capacity can be approximated by a parabola of the form
\begin{equation}
\label{eq_error_exponent_near_capacity_MIMO}
E\left(R\right) \approx \frac{\left(C-R\right)^2}{2V},
\end{equation}
where $V$ is the channel dispersion. By taking uniform input distribution, within the cube $\text{Cb}(a,t)$, in Gallager's random coding error exponent, over the power constrained MIMO fading channel with available CSI at the receiver, it can be shown (see Appendix \ref{app_error_exponent_uni_prior_sec}) that \eqref{eq_error_exponent_near_capacity_MIMO} holds with $C = E\left\{\ln\left({\det\left(\frac{a^2\mathbf{H}^{\dagger}\mathbf{H}}{\pi e\sigma^2}\right)}\right) \right\}$ and $V = t + Var\left(  \ln\left({\det\left( \mathbf{H}^{\dagger}\mathbf{H} \right)}\right) \right)$, when $a/\sigma$ tends to infinity (the high SNR regime). Since the setting without power constraint and under the \emph{FDT} constraint can be thought of as the limit of the power constrained setting, when the SNR tends to infinity, this result hints that $\delta^* = E\left\{\ln\left({\det\left(\frac{\mathbf{H}^{\dagger}\mathbf{H}}{\pi e\sigma^2}\right)}\right) \right\}$ and $V = t + Var\left(  \ln\left({\det\left( \mathbf{H}^{\dagger}\mathbf{H} \right)}\right) \right)$, in that setting.

According to \cite{Telatar} the capacity of the average power constrained MIMO channel is given by:
\begin{equation}
\label{eq_avg_power_constrained_capacity}
C = E\left\{ \ln\left( {\det\left( I_{t} + \mathbf{H}^{\dagger}\mathbf{H}\cdot SNR \right)} \right) \right\}. \nonumber
\end{equation}

\noindent
In the high SNR regime, this capacity can be approximated by:
\begin{equation}
\label{eq_approx_avg_power_constrained_capacity}
C = E\left\{ \ln\left( {\det\left(\mathbf{H}^{\dagger}\mathbf{H}\cdot SNR \right)} \right) \right\}.
\end{equation}

\noindent
It is a well known fact that the capacity of the amplitude constrained channel, or the capacity with the constraint that all the codewords are contained in a cube $\text{Cb}(a,t)$, loses the ``Shaping Gain'' which equals $\frac{2 \pi e}{12}$ (\cite[Section IV.A]{Forney}), relative to the capacity of the average power constrained channel model. Hence, by the assignment of $SNR = \frac{a^2}{\pi e\sigma^2}$ in \eqref{eq_approx_avg_power_constrained_capacity}, we obtain the following capacity in $\text{Cb}(a,t)$,
\begin{equation}
C_a = E\left\{\ln\left({\det\left(\frac{a^2\mathbf{H}^{\dagger}\mathbf{H}}{\pi e\sigma^2}\right)}\right) \right\}. \nonumber
\end{equation}
Finally, we can normalize $C_a$ by the logarithm of the cube volume, which hints that the optimal NLD under the \emph{FDT} constraint is indeed equal:
\begin{equation}
\delta^* = C_a - \ln\left(a^{2t}\right) = E\left\{\ln\left({\det\left(\frac{\mathbf{H}^{\dagger}\mathbf{H}}{\pi e\sigma^2}\right)}\right) \right\}. \nonumber
\end{equation}
\section{Comparison to the Parallel Channels Model}
\label{sec_comparison_to_the_parallel_channels_model}
Let us define the \emph{independent parallel channels model} by the following $L$ independent and identical scalar complex fast fading channels:
\begin{equation}
Y^{(l)} = H^{(l)} \cdot X^{(l)} + Z^{(l)}
\end{equation}
for $l = 1,2,\dots,L$.
Equivalently, in vector notation, the channel model is given by:
\begin{equation}
\mathbf{y} = \mathbf{H} \cdot \mathbf{x} + \mathbf{z}
\end{equation}
where, $\mathbf{x},\mathbf{y},\mathbf{z} \in \mathbb{C}^L$ and $\mathbf{H} = \text{diag}(H^{(1)},\dots,H^{(L)}) \in \mathbb{C}^{L \times L}$. Let us focus on the case where $H^{(l)}$ is distributed as circular symmetric $CN(0,1)$ RV, and the noise vector $\mathbf{z}$ is distributed as circular symmetric $CN(0, \sigma^2 \cdot I_L)$ random vector.

In Section \ref{sec_dispersion_of_parallel_channels_model} we analyze the dispersion of this model and derive simple expressions for its Poltyrev's capacity and channel dispersion. Then, in Sections \ref{subsec_comparison_in_capacity} and \ref{subsec_comparison_in_dispersion} we compare between this model and the MIMO fading model under the \emph{FDT} constraint in terms of Poltyrev's capacity and channel dispersion, respectively.
\subsection{Dispersion of Parallel Channels Model}
\label{sec_dispersion_of_parallel_channels_model}
\begin{theorem}
Let $\epsilon>0$ be a given, fixed, error probability. Denote by $\delta^*(n,\epsilon)$ the optimal NLD for which there exists an $n \cdot L$ complex-dimensional infinite constellation with average error probability at most $\epsilon$, over the $L$ independent parallel channels model. Then, as $n$ grows,
\begin{equation}
\delta^*(n,\epsilon) = \delta^* - \sqrt{\frac{V}{n}}Q^{-1}(\epsilon) + O\left(\frac{\ln(n)}{n}\right),
\end{equation}
where,
\begin{equation}
\delta^* = E\left\{ \ln\left(\det\left(\frac{\mathbf{H}^{\dagger}\mathbf{H}}{\pi e\sigma^2}\right)\right) \right\}
\end{equation}
and
\begin{equation}
V = L + Var(\ln(\det(\mathbf{H}^{\dagger}\mathbf{H}))).
\end{equation}
\end{theorem}
\begin{proof}
From the scalar complex fast fading channel dispersion, we have
\begin{equation}
\delta_0^*(n_0,\epsilon) = \delta_0^* - \sqrt{\frac{V_0}{n_0}}Q^{-1}(\epsilon) + O\left( \frac{\ln(n_0)}{n_0} \right)
\end{equation}
where, $\delta_0^* = E\left\{\ln\left(\frac{|H|^2}{\pi e \sigma^2}\right)\right\}$ and $V_0 = 1 + Var(\ln(|H|^2))$.
Since that in the parallel channels model any channel use is equivalent to $L$ channel uses of the the scalar channel, then by defining $n = \frac{n_0}{L}$ to be the number of channel uses of the parallel channels model, we get trivially that:
\begin{align}
\delta^*(n,\epsilon) &= L\cdot\delta_0^*(n_0,\epsilon) \nonumber
\\&= L\cdot\left(\delta_0^* - \sqrt{\frac{V_0}{n_0}}Q^{-1}(\epsilon) + O\left( \frac{\ln(n_0)}{n_0} \right)\right) \nonumber
\\&= L\cdot\delta_0^* - \sqrt{\frac{L^2 \cdot V_0}{n_0}}Q^{-1}(\epsilon) + O\left( \frac{\ln(n)}{n} \right) \nonumber
\\&= \delta^* - \sqrt{\frac{V}{n}}Q^{-1}(\epsilon) + O\left( \frac{\ln(n)}{n} \right) \nonumber
\end{align}
where, $\delta^* = L\cdot\delta_0^*$ and $V = L \cdot V_0$.
Clearly, since $\mathbf{H}$ is a diagonal matrix with $L$ i.i.d. RV's in its diagonal,
\begin{equation}
\delta^* = E\left\{L\cdot\ln\left(\frac{|H|^2}{\pi e \sigma^2}\right)\right\}
= E\left\{ \ln\left(\det\left(\frac{\mathbf{H}^{\dagger}\mathbf{H}}{\pi e\sigma^2}\right)\right) \right\} \nonumber
\end{equation}
and
\begin{equation}
V = L \cdot (1 + Var(\ln(|H|^2))) = L + Var(\ln(\det(\mathbf{H}^{\dagger}\mathbf{H}))). \nonumber
\end{equation}
\end{proof}
In a similar way as we done for MIMO fading channels under the \emph{FDT} constraint in Theorem \ref{thm_MIMO_channel_delta_and_V}, we can derive simple analytic expressions for $\delta^*$ and $V$, also for the parallel channels model. These expressions are given by the following theorem.
\begin{theorem}
\label{thm_parallel_channel_delta_and_V}
The Poltyrev's capacity $\delta^*$ and the channel dispersion $V$ of the $L$ independent parallel channels model, are given by:
\begin{align}
\delta^* &= -\gamma L  - L\ln(\pi e \sigma^2)
\\
V &= L + \frac{\pi^2L}{6}.
\end{align}
\end{theorem}
\begin{proof}
Follows directly from the fact that
\begin{align}
\delta^* &= E\left\{L\cdot\ln\left(\frac{|H|^2}{\pi e \sigma^2}\right)\right\}
\\V      &= L \cdot (1 + Var(\ln(|H|^2)))
\end{align}
and from Lemma \ref{lem_det_log_moments} for $r=t=1$, where $H \sim CN(0,1)$.
\end{proof}
\subsection{Comparison in terms of Poltyrev's Capacity}
\label{subsec_comparison_in_capacity}
For ``apples to apples'' comparison we will compare between the $t$ independent parallel channels model and the $t \times t$ MIMO channel under the \emph{FDT} constraint. According to Theorems \ref{thm_MIMO_channel_delta_and_V} and \ref{thm_parallel_channel_delta_and_V} we obtain:
\begin{align}
\Delta\delta^* &\triangleq \delta^*_{\text{MIMO}} - \delta^*_{\text{Parallel}} \nonumber
\\             &= -\gamma t + 1 - t + t\sum_{p=1}^{t-1}\frac{1}{p} - t\ln(\pi e \sigma^2) - \left( -\gamma t - t\ln(\pi e \sigma^2) \right)
\\             &= 1 - t + t\sum_{p=1}^{t-1}\frac{1}{p}
\\             &=
                    \left\{ \begin{array}{ll}
                    0 & \textrm{$t=1$}\\
                    1 + t\sum_{p=2}^{t-1}\frac{1}{p} & \textrm{$t>1$}\\
                    \end{array} \right.
\\             &\geq 0 ~\forall t,
\end{align}
which means that the MIMO channel has a greater Poltyrev's capacity than the parallel channels model (with the same noise variance $\sigma^2$) for any $t > 1$.
This result proves that the channel capacity is increased due to the dependency between the channels.

Another way to compare between the capacities of the channels is in terms of the ratio between the highest noise variance that is tolerable in each channel model. It is easy to show that this ratio is given by $\Delta\mu^* \triangleq e^{\frac{\Delta\delta^*}{t}}$ in linear scale, or by $10\log_{10}(e)\cdot\frac{\Delta\delta^*}{t}\cong4.3429\cdot\frac{\Delta\delta^*}{t}$ in dB.
In Figure \ref{fig_mu_vs_t} we can see this ratio for different values of $t$.
\begin{figure}[htp]
\center{\includegraphics[width=0.7\columnwidth]{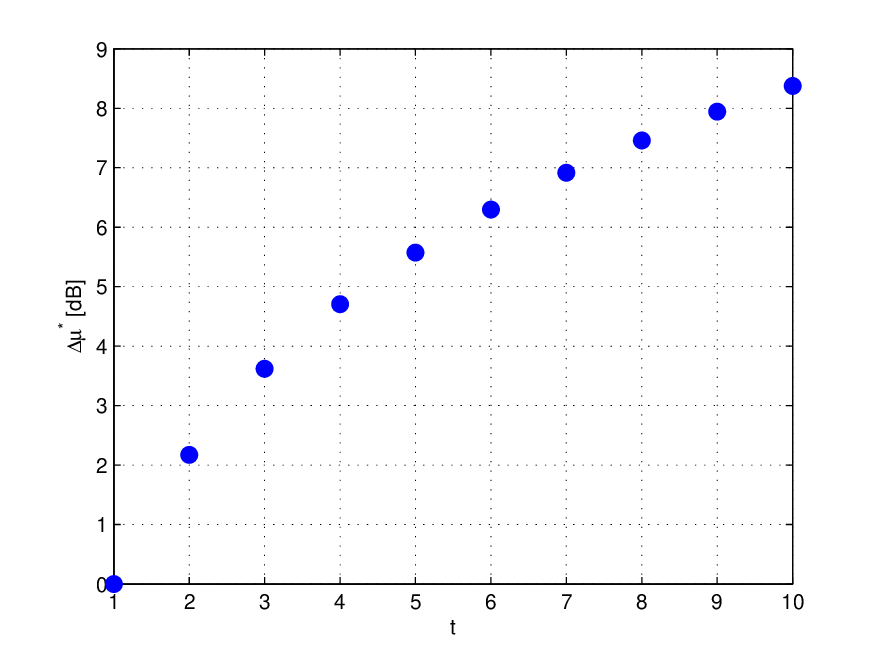}}
\caption{\label{fig_mu_vs_t} $\Delta\mu^*$ vs. the number of antennas $t$.}
\end{figure}
\subsection{Comparison in terms of Channel Dispersion}
\label{subsec_comparison_in_dispersion}
For a fair comparison between the channels in terms of channel dispersion, we need to compare between the $t$ independent parallel channels model and the $t \times t$ MIMO channel model under the \emph{FDT} constraint, in a way that will ensure equal VNR (the analogous SNR for IC's), for any IC that is transmitted over each one of them.
Since the VNR for IC with NLD $\delta$, is proportional to $(\delta^*-\delta)$ in dB (see Section \ref{sec_VNR}), then for getting equal VNR in both of the channels, we need to normalize the fading matrices in a way that will cause theirs Poltyrev's capacity to be equal. It can be verified that the multiplication of the parallel fading matrix by the constant $\rho \triangleq \sqrt{\Delta\mu^*} \geq 1$, where $\Delta\mu^*$ is defined in Section \ref{subsec_comparison_in_capacity}, ensures it.
But since, after normalization, we obtain:
\begin{align}
V &= t + Var\left(\ln(\det(\rho^2\mathbf{H}^{\dagger}\mathbf{H}))\right)
\\&= t + Var\left(\ln(\det(\mathbf{H}^{\dagger}\mathbf{H})) + \ln(\rho^{2t})\right)
\\&= t + Var\left(\ln(\det(\mathbf{H}^{\dagger}\mathbf{H}))\right),
\end{align}
we can observe that the channel dispersion is not affected by normalization. This result does not need to surprise us, since the channel dispersion is not a function of the noise variance. So, by using Theorems \ref{thm_MIMO_channel_delta_and_V} and \ref{thm_parallel_channel_delta_and_V} the channel dispersion difference between the models is given by:
\begin{align}
\Delta V &\triangleq V_{\text{Parallel}} - V_{\text{MIMO}}
\\       &= t + \frac{\pi^2t}{6} - \left( t + \frac{\pi^2t}{6} - \sum_{p=1}^{t-1}\frac{t-p}{p^2} \right)
\\       &= \sum_{p=1}^{t-1}\frac{t-p}{p^2} \geq 0 ~ \forall t.
\end{align}
This result proves that the channel dispersion is decreased due to the dependency between the MIMO channels.
An intuitive explanation for it, is that this dependency has an effect of ``coding'' on the transmitted data. Hence, effectively in the MIMO receiver, we get larger codeword relative to the independent parallel channels model.
Figure \ref{fig_delta_V_vs_t} demonstrates this fact for different values of $t$.
\begin{figure}[htp]
\center{\includegraphics[width=0.7\columnwidth]{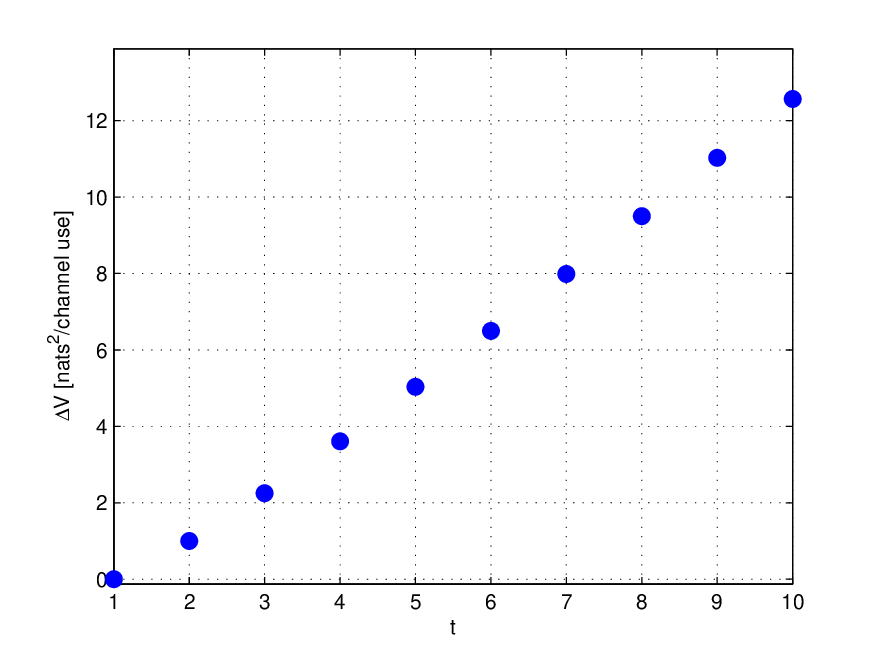}}
\caption{\label{fig_delta_V_vs_t} $\Delta V$ vs. the number of antennas $t$.}
\end{figure}
\section{Generalization}
\label{sec_generalization}
In this section we analyze the dispersion of MIMO fading channels without the constraint of \emph{Full Dimensional Transmission}.
Here, we enable the transmitter to discard part of its dimensions during the transmission.
This section generalizes the previous MIMO dispersion result under the \emph{FDT} constraint and hints about the MIMO dispersion and Poltyrev's capacity without any constraint.
Surprisingly, in IC's over MIMO channels, this reduction of dimensions can increase the Poltyrev's capacity.

Assume a transmission of an $l = n \cdot \bar{t}$ complex dimensional IC $S$ with NLD $\delta$ over the $t \times r$ MIMO channel, where $\bar{t} \leq t \leq r$ using $n$ channel uses. Let us denote by $p_i\geq0,~i=1,\dots,t$ the fraction of channel uses where only $i$ transmit antennas are in use and the rest are zeroed. Clearly, $\sum_{i=1}^{t}p_i=1$ and $\bar{t} = \sum_{i=1}^{t}p_i \cdot i$. In addition, let us denote by $\mathbf{H}_{i \times r}$ the effective MIMO fading matrix in channel uses where only $i$ transmit antennas are in use. Without loss of generality, we can assume that the overall fading matrix is given by the following concatenation of $n$ block diagonal matrices:
\begin{equation}
\mathbf{H}^n = \text{diag}\left(\mathbf{H}_{1 \times r}^{p_1 \cdot n},\dots,\mathbf{H}_{t \times r}^{p_t \cdot n}\right),
\end{equation}
and the effective $n \cdot \bar{t}$ complex dimensional circular symmetric Gaussian noise vector is given by the concatenation of the following $n$ consecutive noise vectors:
\begin{displaymath}
\mathbf{z}'^{n} =
\left( \begin{array}{c}
\mathbf{z}'^{p_1 \cdot n} \\
\vdots\\
\mathbf{z}'^{p_t \cdot n} \\
\end{array} \right).
\end{displaymath}
Hence, by using the same arguments as in Section \ref{sec_mimo_sphere_packing_bound} we can get the following \emph{sphere packing bound} for any such IC $S$ with NLD $\delta$:
\begin{equation}
P_e\left(S\right) \geq Pr\left\{ \hspace{-0.05cm} \left\|\mathbf{z}'^n\right\|^2 \geq e^{-\frac{\delta}{\bar{t}}}\left(\frac{\det(\mathbf{H}^{n\dagger}\mathbf{H}^n)}{V_{2n\bar{t}}}\right)^{\frac{1}{n\bar{t}}} \hspace{-0.05cm} \right\},
\end{equation}
and using similar arguments as in Section \ref{sec_mimo_proof_of_converse_part} for any fixed average error probability $\epsilon$:
\begin{align}
\delta \leq \sum_{i=1}^{t} p_i\cdot\delta^*(i,r) - \sqrt{ \frac{\sum_{i=1}^{t} p_i \cdot V(i,r)}{n} }Q^{-1}\left(\epsilon\right) + \frac{1}{2n}\ln(n) + O\left(\frac{1}{n}\right),
\end{align}
where $\delta^*(i,r)$ and $V(i,r)$ are the Poltyrev's capacity and channel dispersion over the $i \times r$ MIMO channel under the constraint of \emph{FDT}, which are given in Theorem \ref{thm_mimo_main_result}.
For simplicity, let us denote $\bar{\delta}^* = \sum_{i=1}^{t} p_i\cdot\delta^*(i,r)$ and $\bar{V} = \sum_{i=1}^{t} p_i \cdot V(i,r)$ to get the following:
\begin{align}
\label{eq_converse_under_BDUT}
\delta \leq \bar{\delta}^* - \sqrt{ \frac{\bar{V}}{n} }Q^{-1}\left(\epsilon\right) + \frac{1}{2n}\ln(n) + O\left(\frac{1}{n}\right).
\end{align}
Notice that instead of zeroing the rest of the transmit antennas in any channel use where we are using only $i$ antennas we can use all of them by the transmission of the following vector
\begin{displaymath}
U\cdot\mathbf{x} = U \cdot
\left( \begin{array}{c}
x_{1} \\
\vdots\\
x_{i} \\
0     \\
\vdots\\
0     \\
\end{array} \right),
\end{displaymath}
where $U\in\mathbb{C}^{t \times t}$ is an arbitrary unitary matrix. Note that in any channel use we have the freedom to choose a different unitary matrix.
By this transmission we do not change the IC density and we get in any channel use an effective channel matrix of $\mathbf{H} \cdot U$, which has the same statistics as $\mathbf{H}$ of $t \times r$ i.i.d. circular symmetric Gaussian RV's with unit variance. Hence, those operations do not change the optimal transmission scheme and the converse of \eqref{eq_converse_under_BDUT} is still valid. Let us call such a transmission \emph{Block Diagonal Unitary Transmission (BDUT)}. Note that \emph{BDUT} spreads the transmitted signals among all the transmit antennas.

By taking the limit $n \to \infty$ and due to the averaging property, we get that the Poltyrev's capacity under this constraint of \emph{BDUT} holds the following:
\begin{equation}
\delta^* \leq \bar{\delta}^* \leq \max_{i\in\{1,\dots,t\}}\delta^*(i,r).
\end{equation}
Since $\delta^*(i,r)$ is also achievable according to Theorem \ref{thm_mimo_main_result}, then
\begin{equation}
\delta^* = \delta^*(t_{\text{opt}},r)
\end{equation}
where,
\begin{equation}
t_{\text{opt}} \triangleq \argmax_{i\in\{1,\dots,t\}}\delta^*(i,r).
\end{equation}
Hence, by taking $\mathbf{p}_{\text{opt}} \triangleq (p_1,\dots,p_{t_{\text{opt}}},\dots,p_t) = (0,\dots,1,\dots,0)$ for any fixed average error probability $\epsilon$ and for large enough $n$,  the highest achievable NLD under the constraint of \emph{BDUT} is given by the following
\begin{equation}
\delta^*(n,\epsilon) = \delta^* - \sqrt{\frac{V}{n}}Q^{-1}(\epsilon) + O\left(\frac{\ln(n)}{n}\right),
\end{equation}
where,
\begin{equation}
\label{eq_delta_star_BDUT}
\delta^* \triangleq \delta^*(t_{\text{opt}},r) = E\left\{\ln\left({\det\left(\frac{\mathbf{H}_{r \times t_{\text{opt}}}^{\dagger}\mathbf{H}_{r \times t_{\text{opt}}}}{\pi e\sigma^2}\right)}\right) \right\}
\end{equation}
and
\begin{align}
V \triangleq V(t_{\text{opt}},r) = t_{\text{opt}} + Var\left( \ln\left(\det(\mathbf{H}_{r \times t_{\text{opt}}}^{\dagger}\mathbf{H}_{r \times t_{\text{opt}}})\right) \right).
\end{align}
Although the Poltyrev's capacity under the constraint of \emph{BDUT} does not necessarily gives the optimal NLD without any constraint, we conjecture that $\delta^*$ in \eqref{eq_delta_star_BDUT} is actually the Poltyrev's capacity without any constraint.

Notice that this generalized dispersion result reveals a very surprising phenomena of infinite constellations in MIMO fading channels. In contrast to the capacity of finite constellations in MIMO fading channels \cite{Telatar}, the Potyrev's capacity can be increased by discarding part of its transmission dimensions.
Let's demonstrate this result by an example: assume a transmission over the $3 \times 3$ MIMO fading channel with noise variance of $\sigma^2$.
It can be seen in Figure \ref{fig_delta_star_vs_inv_sigma_sqr} that the inverse noise variance, or the SNR-like region, can be separated into 3 (or $t$ in the general case) regions of \emph{High}, \emph{Moderate} and \emph{Low} SNR regions. In the \emph{Low} SNR region the optimal number of transmit antennas equals $t_{\text{opt}} = 1$, in the \emph{Moderate} region $t_{\text{opt}} = 2$ and in the \emph{High} region $t_{\text{opt}} = t = 3$.
In other words, not for any inverse noise variance $\nicefrac{1}{\sigma^2}$, or SNR, the optimal number of transmit antennas equals to the full transmit dimension of $t$.
\begin{figure}[htp]
\center{\includegraphics[width=0.7\columnwidth]{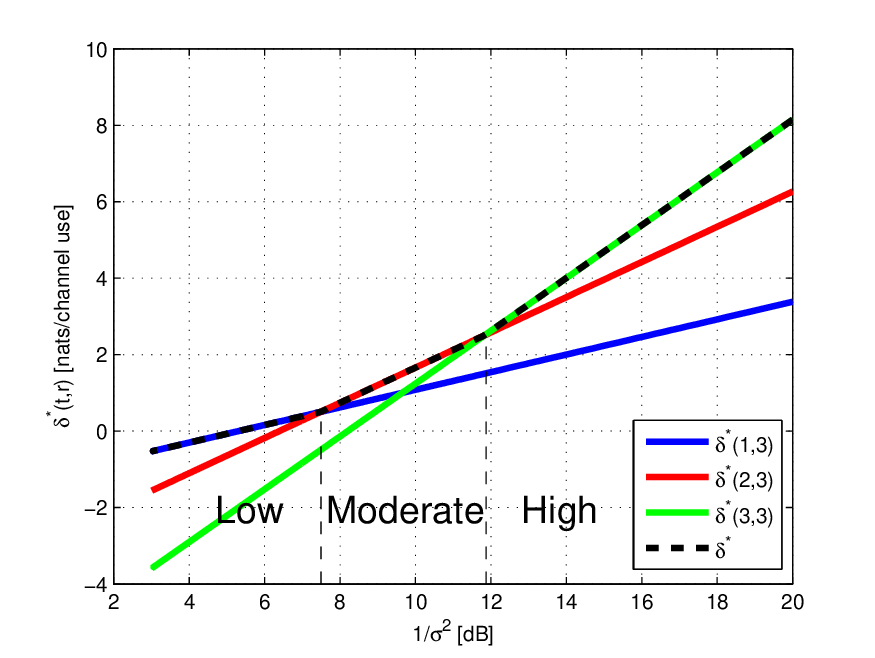}}
\caption{\label{fig_delta_star_vs_inv_sigma_sqr} The Poltyrev's capacities under the \emph{BDUT} and under the \emph{FDT} constraints vs. the SNR-like $\nicefrac{1}{\sigma^2}$ over the $3 \times 3$ MIMO fading channel.}
\end{figure}
In Figure \ref{fig_dispersion_vs_inv_sigma_sqr} we can see also the channel dispersion under the constraint of \emph{BDUT} as function of the inverse noise variance $\nicefrac{1}{\sigma^2}$.
\begin{figure}[htp]
\center{\includegraphics[width=0.7\columnwidth]{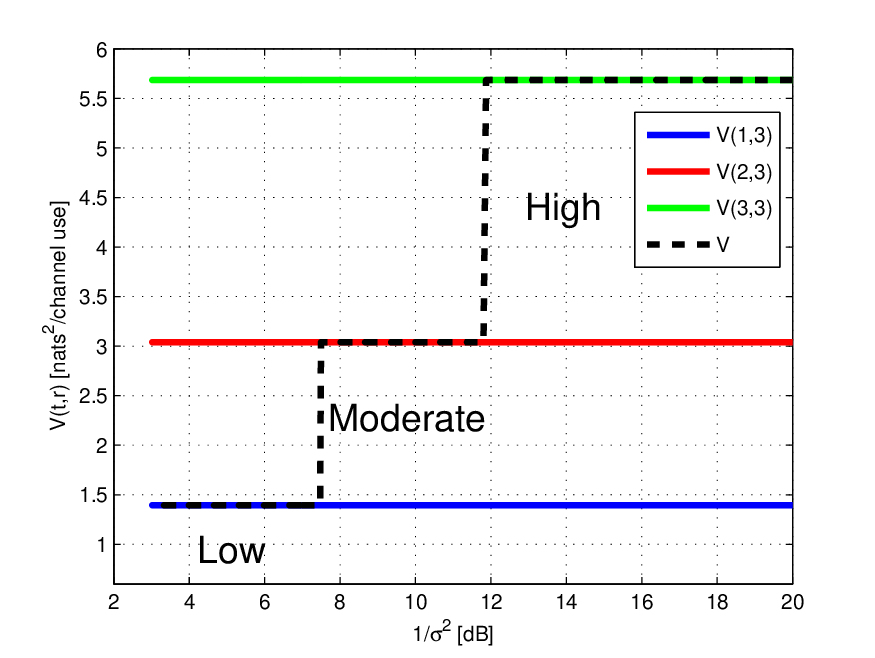}}
\caption{\label{fig_dispersion_vs_inv_sigma_sqr} The channel dispersions under the \emph{BDUT} and under the \emph{FDT} constraints vs. the SNR-like $\nicefrac{1}{\sigma^2}$ over the $3 \times 3$ MIMO fading channel.}
\end{figure}

Another interesting relation of this surprising result to finite constellations in MIMO fading channels is in the sense of \emph{Shaped Lattices}. Let us restrict the discussion to the case of lattices with a shaping of a complex hypercube, which without loss of generality, can be assumed to have unit volume. Inspired by the the results of Loeliger in \cite{Loeliger}, for any $i\in\{1,\dots,t\}$, there exists an $n \cdot i$ complex dimensional (translated) shaped lattice (in an $n \cdot i$ complex dimensional unit volume cube) with achievable rate of $R_i = \delta^*(i,r)$ with arbitrarily error probability under the suboptimal \emph{Lattice Decoder} (which is not aware to the shaping) over the $t \times r$ MIMO channel, when $n$ tends to infinity. Note that it was also conjectured in \cite{Loeliger} that this is the highest achievable rate of cube shaped lattices under \emph{Lattice Decoding}. Hence, in moderate and low SNR it seems beneficial to reduce dimensions in MIMO fading channels where we are using shaped lattices and \emph{Lattice Decoder}. This is in contrast to optimal finite constellations that can achieve the MIMO capacity by using all of the transmit dimensions.

In this section we generalized the previous MIMO dispersion result under the \emph{FDT} constraint and gave hints about the MIMO dispersion and Poltyrev's capacity without any constraint. Nevertheless, the general case of MIMO Poltyrev's capacity and channel dispersion of IC's without any constraint are still subjects for further research. In addition, the dispersion analysis of FC's over the power constrained MIMO fading channel is also a subject for further research.

\chapter{Summary and conclusions}
\label{chapter:Conclusions}
In this thesis we considered infinite constellations
over the fading channels, with perfect CSI available at the receiver.
We applied the ``dispersion analysis'', which provides the optimal asymptotic
relation between the achievable NLD and the block length for a given
error probability. This relation essentially quantifies the gap between the
optimal NLD (Poltyrev's capacity) and the highest attainable NLD at finite block
length and a fixed error probability.

We analyzed first the case of scalar fast fading channels, where the
fading process is a series of i.i.d. RV's. Using the dependence testing bound, the
sphere packing bound and some normal approximation techniques, we proved
that the dispersion analysis holds in that setting, and we also found
the relevant terms - Poltyrev's capacity and the channel dispersion.
Using similar, but more elaborate tools, we extended the analysis
to the general case of stationary fading processes. In that setting, we showed
that unlike the capacity, the channel dispersion is affected by the fading dynamics.
Moreover, in typical fading processes, this dispersion is increased relative to
the fast fading channel, with the same marginal fading distribution. This fact can
motivate the usage of random interleaver in practical systems with finite
block length.

In the setting of MIMO Rayleigh fast fading channels under the constraint
of \emph{Full Dimensional Transmission}, our analysis showed similar results,
which promise lower channel dispersion and greater Poltyrev's capacity, relative to the
independent parallel channels, due to the dependency between the received signals.
Partial analysis of IC's in the general MIMO case revealed a
very surprising phenomena of Poltyrev's capacities in MIMO fading channels:
In contrast to the capacity of FC's over MIMO fading channels, reducing the IC's transmission
dimension can increase the Poltyrev's capacity of the channel.

Finally, relations to the amplitude and to the power constrained fading channels were also
discussed, especially in terms of capacity, channel dispersion and error exponents.
These relations hint that in most cases, including SISO and \emph{FDT} MIMO the unconstrained
model can be interpreted as the limit of the constrained model, when the SNR tends to infinity.

There are still some open problems for further research. First, in our proof
of the direct part, we used the dependence testing bound, which is based on
a suboptimal decoder. Hence, a proof which is based on an optimal ML decoder,
can achieve a more refined result. We conjecture that the dispersion analysis
accuracy will be $O\left(\frac{1}{n}\right)$, and the highest achievable NLD,
in the setting of fixed error probability and finite block length, will increase
by $\frac{1}{2n}\ln(n)$.
In MIMO channels, the analysis presented here, is the first
that was done in that setting.
Hence, major problems to be analyzed are the case where the number of transmit antennas
is greater than the number of receive antennas and a completion of the analysis for the Poltyrev's
capacity and the channel dispersion without any constraint. In addition, the dispersion of
MIMO channels with spatial correlation, power constraint, memory and different fading distributions
are very interesting problems for further research.

\appendix
\chapter{Proof of the Regularization Lemma}
\label{app_proof_regularization_lemma}
\begin{proof}[Proof of Lemma \ref{lem_regularixation}]
Fix $\xi > 0$ and consider the receiver's IC $S_{\mathbf{H}}$, where $\mathbf{H}$ is not a \emph{$\xi$ - strong fading realization}. First, we will find large enough $a_*$ s.t. the density of the codewords in $S_{\mathbf{H}} \bigcap \mathbf{H}\cdot\text{Cb}(a_*)$, and the average error probability in transmitting codewords from it, over the AWGN channel, are close enough to $\gamma_{\text{rc}}(\mathbf{H})$ and $\epsilon(\mathbf{H})$. Then we will construct a regular IC by tiling this FC over the whole space $\mathbb{R}^n$. For this IC the desired bounds of the lemma will hold.

By definition we have
\begin{equation}
P_e\left(S_{\mathbf{H}}\right) = P_e(S|\mathbf{H}) = \epsilon(\mathbf{H}) = \limsup_{a \to \infty} \frac{1}{M(S_{\mathbf{H}},a)} \sum_{s_{\text{rc}} \in S_{\mathbf{H}} \bigcap \mathbf{H}\cdot\text{Cb}(a)} P_e(s_{\text{rc}} | \mathbf{H})
\end{equation}
\begin{equation}
\gamma_{\text{rc}}(\mathbf{H}) = \gamma_{\text{rc}} = \limsup_{a \to \infty} \frac{M(S_{\mathbf{H}},a)}{\text{Vol}(\mathbf{H}\cdot\text{Cb}(a))}
                                                    = \limsup_{a \to \infty} \frac{M(S_{\mathbf{H}},a)}{\det(\mathbf{H})a^n}.
\end{equation}
From the existence of the limits above there exists $a_0$ s.t. for every $a > a_0$ the following holds:
\begin{equation}
\label{eq_sup_error_probability}
\sup_{b > a} \frac{1}{M(S_{\mathbf{H}},b)} \sum_{s_{\text{rc}} \in S_{\mathbf{H}} \bigcap \mathbf{H}\cdot\text{Cb}(b)} P_e(s_{\text{rc}} | \mathbf{H}) < \epsilon(\mathbf{H})(1 + \xi/2)
\end{equation}
and
\begin{equation}
\label{eq_sup_density}
\sup_{b > a} \frac{M(S_{\mathbf{H}},b)}{\det(\mathbf{H})b^n} > \frac{\gamma_{\text{rc}}}{\sqrt{1 + \xi}}.
\end{equation}
Define $\Delta$ s.t.
\begin{equation}
2nQ\left( \frac{h^*_{\min}\Delta}{\sigma} \right) = \frac{\xi}{2}\cdot\epsilon(\mathbf{H}),
\end{equation}
and define $a_{\Delta}$ as the solution of
\begin{equation}
\frac{\text{Vol}(\mathbf{H}\cdot\text{Cb}(a_{\Delta} + 2\Delta))}{\text{Vol}(\mathbf{H}\cdot\text{Cb}(a_{\Delta}))}
= \left( \frac{a_{\Delta} + 2\Delta}{a_{\Delta}} \right)^n = \sqrt{1 + \xi}.
\end{equation}
Define $a_{\max} = \max(a_0,a_{\Delta})$. According to \eqref{eq_sup_error_probability} and \eqref{eq_sup_density} there exists $a_* > a_{\max}$ s.t.
\begin{align}
\begin{aligned}
\frac{1}{M(S_{\mathbf{H}},a_*)} \sum_{s_{\text{rc}} \in S_{\mathbf{H}} \bigcap \mathbf{H}\cdot\text{Cb}(a_*)} \hspace{-0.5cm} P_e(s_{\text{rc}} | \mathbf{H}) \leq \sup_{b > a_{\max}} \frac{1}{M(S_{\mathbf{H}},b)} \sum_{s_{\text{rc}} \in S_{\mathbf{H}} \bigcap \mathbf{H}\cdot\text{Cb}(b)} \hspace{-0.5cm} P_e(s_{\text{rc}} | \mathbf{H}) < \epsilon(\mathbf{H})(1 + \xi/2)
\end{aligned}
\end{align}
and
\begin{equation}
\label{eq_regularization_density_lower_bound}
\frac{M(S_{\mathbf{H}},a_*)}{\det(\mathbf{H})a_*^n} > \frac{\gamma_{\text{rc}}}{\sqrt{1 + \xi}}.
\end{equation}

Define the FC $G_{\mathbf{H}} = S_{\mathbf{H}} \bigcap \mathbf{H}\cdot\text{Cb}(a_*)$, and denote by $P_e^{G_{\mathbf{H}}}(s_{\text{rc}})$ the decoding error probability of any codeword $s_{\text{rc}} \in G_{\mathbf{H}}$ in transmission over the AWGN channel. Since $G_{\mathbf{H}} \subset S_{\mathbf{H}}$ then $P_e^{G_{\mathbf{H}}}(s_{\text{rc}}) \leq P_e(s_{\text{rc}}|\mathbf{H})$, and the average error probability of the FC is given by
\begin{equation}
\label{eq_FC_for_regularization_ub}
P_e(G_{\mathbf{H}}) = \frac{1}{|G_{\mathbf{H}}|} \sum_{s_{\text{rc}} \in G_{\mathbf{H}}} {P_e^{G_{\mathbf{H}}}(s_{\text{rc}})}
                 \leq \frac{1}{|G_{\mathbf{H}}|} \sum_{s_{\text{rc}} \in G_{\mathbf{H}}} {P_e(s_{\text{rc}}|\mathbf{H})}
                    < \epsilon(\mathbf{H})(1 + \xi/2).
\end{equation}
Now, we will create a regular IC, denoted by $S^{'}_{\mathbf{H}}$, by tiling the FC $G_{\mathbf{H}}$ to the whole space $\mathbb{R}^n$ in the
following way:
\begin{equation}
S^{'}_{\mathbf{H}} = \left\{ s_{\text{rc}} + \mathbf{H} \cdot I \cdot (a_* + 2\Delta) : s_{\text{rc}} \in G_{\mathbf{H}}, I \in \mathbb{Z}_n \right\},
\end{equation}
where $\mathbb{Z}_n$ is the $n$ dimensional integers lattice.

The error probability of any $s_{\text{rc}} \in S^{'}_{\mathbf{H}}$ equals the probability of decoding by a mistake to another codeword from the same copy of the FC $G_{\mathbf{H}}$ or to a codeword in another copy. Hence, the average error probability of $S^{'}_{\mathbf{H}}$, with equiprobable codewords transmission over the AWGN channel, can be upper bounded by the union bound as follows:
\begin{equation}
\label{eq_regularized_IC_upper_bound}
P_e\left(S^{'}_{\mathbf{H}}\right) \leq P_e\left(G_{\mathbf{H}}\right) + \sum_{i=1}^{n}{2Q\left(\frac{H_i\Delta}{\sigma}\right)}.
\end{equation}
Since the given fading channel realization is not a \emph{$\xi$ - strong fading realization}, and from the definition of $\Delta$ we obtain:
\begin{equation}
\label{eq_sum_of_Q_upper_bound}
\sum_{i=1}^{n}{2Q\left(\frac{H_i\Delta}{\sigma}\right)} \leq 2nQ\left(\frac{h^*_{\min}\Delta}{\sigma}\right) = \frac{\xi}{2}\cdot\epsilon(\mathbf{H}),
\end{equation}
where $h^*_{\min}(\xi)$ is the solution of $Pr\{H_{\min} \leq h^*_{\min} \} = \xi$. Combining \eqref{eq_FC_for_regularization_ub}, \eqref{eq_regularized_IC_upper_bound} and \eqref{eq_sum_of_Q_upper_bound} we obtain the desired result:
\begin{align}
P_e\left(S^{'}_{\mathbf{H}}\right) \leq \epsilon(\mathbf{H})(1 + \xi).
\end{align}
The density of $S^{'}_{\mathbf{H}}$ is given by
\begin{equation}
\gamma_{\text{rc}}^{'}(\mathbf{H}) = \gamma_{\text{rc}}^{'} = \frac{|G_{\mathbf{H}}|}{\text{Vol}(\mathbf{H}\cdot\text{Cb}(a_* + 2\Delta))}
                                                            = \frac{M(S_{\mathbf{H}}, a_*)}{\det(\mathbf{H})a_*^n}\cdot\left(\frac{a_*}{a_*+2\Delta} \right)^n.
\end{equation}
Combining \eqref{eq_regularization_density_lower_bound} with the definition of $a_{\Delta}$ and the fact that $a_* > a_{\Delta}$ we obtain the desired result:
\begin{align}
\gamma_{\text{rc}}^{'} > \frac{\gamma_{\text{rc}}}{1 + \xi}.
\end{align}

Let us denote by $\mathbf{H} = \text{diag}(h_1, \dots, h_n)$ the given channel realization. By its construction, for any $s_{\text{rc}} \in S^{'}_{\mathbf{H}}$, the set of points $\{s_{\text{rc}} \pm h_i\cdot (a_* + 2\Delta) \cdot \underline e_i, i=1, \dots, n \}$ is also in $S^{'}_{\mathbf{H}}$, where $\{\underline e_i\}_{i=1}^n$ is the standard basis of $\mathbb{R}^n$. Hence, any Voronoi cell of $S^{'}_{\mathbf{H}}$ is contained within a sphere of radius $r_0 \triangleq \sqrt{n}(a_* + 2\Delta)h_{\max}$ centered around its codeword, where $h_{\max} \triangleq \max(h_1, \dots, h_n)$. This proves that $S^{'}_{\mathbf{H}}$ is indeed a regular IC.
\end{proof}
\chapter{Proof of the Log of Chi Square Distribution Lemma}
\label{app_proof_lemma_ln_chi2_density_moments}
\begin{proof}[Proof of Lemma \ref{lem_ln_chi2_pdf}]
By simple variables substitution, we get the following relation between the CDFs of $Y_n$ and $X$:
\begin{equation}
\label{eq_Yn_Chi2_CDF}
F_{Y_n}(y) = F_{\chi^2_n}(ne^{\sqrt{\frac{2}{n}}y}).
\end{equation}
Then, if we differentiate \eqref{eq_Yn_Chi2_CDF} w.r.t. $y$ we will get the following relation between the RVs' PDFs:
\begin{equation}
\label{eq_Yn_Chi2_PDF}
f_{Y_n}(y) = \sqrt{2n}e^{\sqrt{\frac{2}{n}}y}f_{\chi^2_n}(ne^{\sqrt{\frac{2}{n}}y}).
\end{equation}
Assignment of the $\chi^2_n$'s PDF, $f_{\chi^2_n}(x) = \frac{x^{\frac{n}{2}-1}e^{-\frac{x}{2}}}{2^\frac{n}{2}\Gamma(\frac{n}{2})}, x>0$ will give us
\begin{equation}
f_{Y_n}(y) = \frac{ (\frac{n}{2})^{\frac{n-1}{2}}}{\Gamma(\frac{n}{2})}e^{\sqrt{\frac{n}{2}}y - \frac{n}{2}e^{\sqrt{\frac{2}{n}}y}} \nonumber,
\end{equation}
which completes the proof of \eqref{eq_Yn_PDF}.
From the Stirling approximation for the Gamma function for $z\in\mathbb{R}$ we get
\begin{equation}
\label{eq_Gamma_Stirling_approximation}
\Gamma(z+1)=z\Gamma(z)=\sqrt{2\pi e}\left(\frac{z}{e}\right)^z\left(1+O\left(\frac{1}{z}\right)\right).
\end{equation}
Using \eqref{eq_Gamma_Stirling_approximation} for $z=\frac{n}{2}$ we get
\begin{equation}
\label{eq_Gamma__n_div_2_approximation}
\Gamma\left(\frac{n}{2}\right)=\frac{\Gamma(\frac{n}{2}+1)}{\frac{n}{2}}=\sqrt{\frac{4\pi}{n}}\left(\frac{n}{2e}\right)^{\frac{n}{2}}\left(1+O\left(\frac{1}{n}\right)\right).
\end{equation}
The assignment of \eqref{eq_Gamma__n_div_2_approximation} in \eqref{eq_Yn_PDF} gives us
\begin{align}
\begin{aligned}
\label{align_Yn_PDF_approximation_1}
f_{Y_n}(y) &= \frac{1}{\sqrt{2\pi}}e^{\frac{n}{2} + \sqrt{\frac{n}{2}}y - \frac{n}{2}e^{\sqrt{\frac{2}{n}}y}}\left(\frac{1}{1+O\left(\frac{1}{n}\right)}\right)
\\         &= \frac{1}{\sqrt{2\pi}}e^{\frac{n}{2} + \sqrt{\frac{n}{2}}y - \frac{n}{2}e^{\sqrt{\frac{2}{n}}y}}\left(1+O\left(\frac{1}{n}\right)\right),
\end{aligned}
\end{align}
for any $n > N_0$, for some finite $N_0$.
\newline\indent
By Taylor's theorem for $g(x) = e^x$ around $x_0 = 0$, the following holds:
\begin{equation}
g(x) = \sum_{k=0}^{K}\frac{x^k}{k!} + \frac{e^{\zeta}x^{K+1}}{(K+1)!},
\end{equation}
for some real number $\zeta \in [0,x]$.
Using it with $K=2$ and $x \equiv \sqrt{\frac{2}{n}}y$ we obtain:
\begin{equation}
e^{\sqrt{\frac{2}{n}}y} = 1 + \sqrt{\frac{2}{n}}y + {\frac{1}{n}}y^2 + \frac{\sqrt{2}e^{\zeta(y)}}{3n^{\frac{3}{2}}}y^3,
\end{equation}
where for $y\in[-n^{\frac{1}{6}},n^{\frac{1}{6}}]$, then $\zeta(y)\in[-\frac{\sqrt{2}}{n^{\frac{1}{3}}},\frac{\sqrt{2}}{n^{\frac{1}{3}}}]$.
\newline
Assigning it in \eqref{align_Yn_PDF_approximation_1} , for any $n > N_0$ and for $y\in[-n^{\frac{1}{6}},n^{\frac{1}{6}}]$, gives us:
\begin{equation}
\label{eq_Yn_PDF_approximation_2}
f_{Y_n}(y) = \frac{1}{\sqrt{2\pi}}e^{-\frac{y^2}{2}} \cdot e^{-\frac{e^{\zeta(y)}}{3\sqrt{2n}}y^3}\left({1+O\left(\frac{1}{n}\right)}\right).
\end{equation}
Using Taylor's theorem again with $K=0$ and $x \equiv -\frac{e^{\zeta(y)}}{3\sqrt{2n}}y^3$ we obtain:
\begin{equation}
e^{-\frac{e^{\zeta(y)}}{3\sqrt{2n}}y^3} = 1 - \frac{e^{-\eta(y)} \cdot e^{\zeta(y)}}{3\sqrt{2n}}y^3,
\end{equation}
where for $y\in[-n^{\frac{1}{6}-\delta},n^{\frac{1}{6}-\delta}]$ for some $0 \leq \delta < \frac{1}{6}$, then $\eta(y)\in(-\frac{1}{n^{3\delta}},\frac{1}{n^{3\delta}})$.
\newline
Combining all the above, we get that for any $n > N_0$, and for $y\in[-n^{\frac{1}{6}-\delta},n^{\frac{1}{6}-\delta}]$ for some $0 \leq \delta < \frac{1}{6}$:
\begin{equation}
\label{eq_Yn_PDF_approximation_2}
f_{Y_n}(y) = \frac{1}{\sqrt{2\pi}}e^{-\frac{y^2}{2}} - \frac{e^{\nu(y)}}{6\sqrt{\pi}}\cdot\frac{y^3e^{-\frac{y^2}{2}}}{\sqrt{n}} + O\left(\frac{e^{-\frac{y^2}{2}}}{n}\right),
\end{equation}
where $\nu(y)\triangleq\zeta(y)-\eta(y)$ and $|\nu(y)| < \frac{1}{n^{3\delta}} + \frac{\sqrt{2}}{n^{\frac{1}{3}}}$.

By definition $e_n(y) \triangleq f_{Y_n}(y) - N(0,1)$, then:
\begin{align}
\begin{aligned}
e_n &\triangleq \int_{-\infty}^{\infty}|e_n(y)|dy
\\  &\leq \int_{|y|\leq n^{\frac{1}{6}}}|e_n(y)|dy + \int_{|y|>n^{\frac{1}{6}}}f_{Y_n}(y)dy + \int_{|y|>n^{\frac{1}{6}}}N(0,1)dy
\\  &= \int_{|y|\leq n^{\frac{1}{6}}}|e_n(y)|dy + 1 - \int_{|y|\leq n^{\frac{1}{6}}}f_{Y_n}(y)dy + \int_{|y|>n^{\frac{1}{6}}}N(0,1)dy
\\  &= \int_{|y|\leq n^{\frac{1}{6}}}|e_n(y)|dy + 1 - \int_{|y|\leq n^{\frac{1}{6}}}N(0,1)dy - \int_{|y|\leq n^{\frac{1}{6}}}e_n(y)dy + \int_{|y|>n^{\frac{1}{6}}}N(0,1)dy
\\  &= \int_{|y|\leq n^{\frac{1}{6}}}|e_n(y)|dy - \int_{|y|\leq n^{\frac{1}{6}}}e_n(y)dy + 2\int_{|y|>n^{\frac{1}{6}}}N(0,1)dy
\\  &\leq 2\int_{|y|\leq n^{\frac{1}{6}}}|e_n(y)|dy + 4Q(n^{\frac{1}{6}})
\\  &= O\left(\int_{-\infty}^{\infty}\frac{|y|^3e^{-\frac{y^2}{2}}}{\sqrt{n}}dy\right) + O\left(\int_{-\infty}^{\infty}\frac{e^{-\frac{y^2}{2}}}{n}dy\right) + O\left(e^{-\frac{n^{\frac{1}{3}}}{2}}\right) = O\left(\frac{1}{\sqrt{n}}\right),
\end{aligned}
\end{align}
which completes the proof of \eqref{eq_Yn_asymptotic_PDF}.
\newline\indent
Now let get some insight about the result.
By taking some $0 < \delta < \frac{1}{6}$, we can see that $\nu(y) \approx 0$, for $y\in[-n^{\frac{1}{6}-\delta},n^{\frac{1}{6}-\delta}]$. Hence, in that range, we can get the following approximation:
\begin{equation}
\label{eq_e_n_y_approx}
e_n(y) \approx -\frac{y^3e^{-\frac{y^2}{2}}}{6\sqrt{\pi n}}.
\end{equation}
By taking the derivative of \eqref{eq_e_n_y_approx} w.r.t. $y$ and the comparison to zero, we can observe that $y_0 \approx \pm \sqrt{3}$ are the points of the maximal (absolute) errors regardless of $n$, which equals $e_n(y_0) \approx \mp \frac{\sqrt{3}}{2e^{\frac{3}{2}}\sqrt{\pi n}} \approx \mp \frac{0.1}{\sqrt{n}}$. This property and also the great accuracy between the numerical calculation of $e_n(y)$ and its theoretical approximation can be seen in Figure \ref{fig_e_n_y} for $n=10^4$.
\begin{figure}[htp]
\center{\includegraphics[width=0.7\columnwidth]{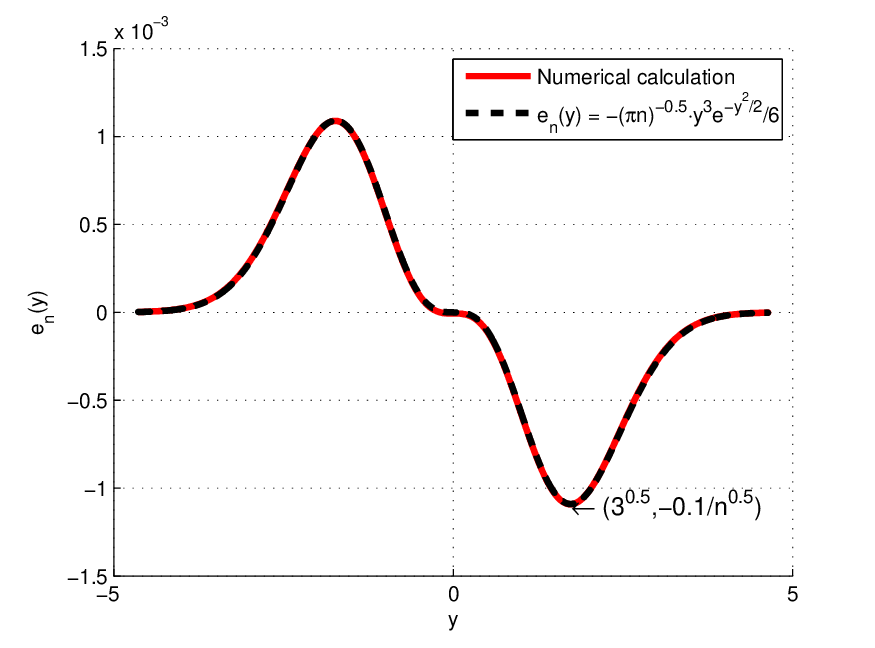}}
\caption{\label{fig_e_n_y} Numerical calculation and theoretical approximation of $e_n(y)$ for $n=10^4$.}
\end{figure}
\newline
Since, this factor contributes the most to the total error $e_n$, we can approximate it for large enough $n$, by the following:
\begin{equation}
\label{eq_e_n_approx}
e_n \approx \int_{-\infty}^{\infty}\frac{|y|^3e^{-\frac{y^2}{2}}}{6\sqrt{\pi n}}dy = \frac{2}{3\sqrt{\pi n}}.
\end{equation}
The great accuracy of \eqref{eq_e_n_approx} as function of $n$ can be seen in Figure \ref{fig_e_n}.
\begin{figure}[htp]
\center{\includegraphics[width=0.7\columnwidth]{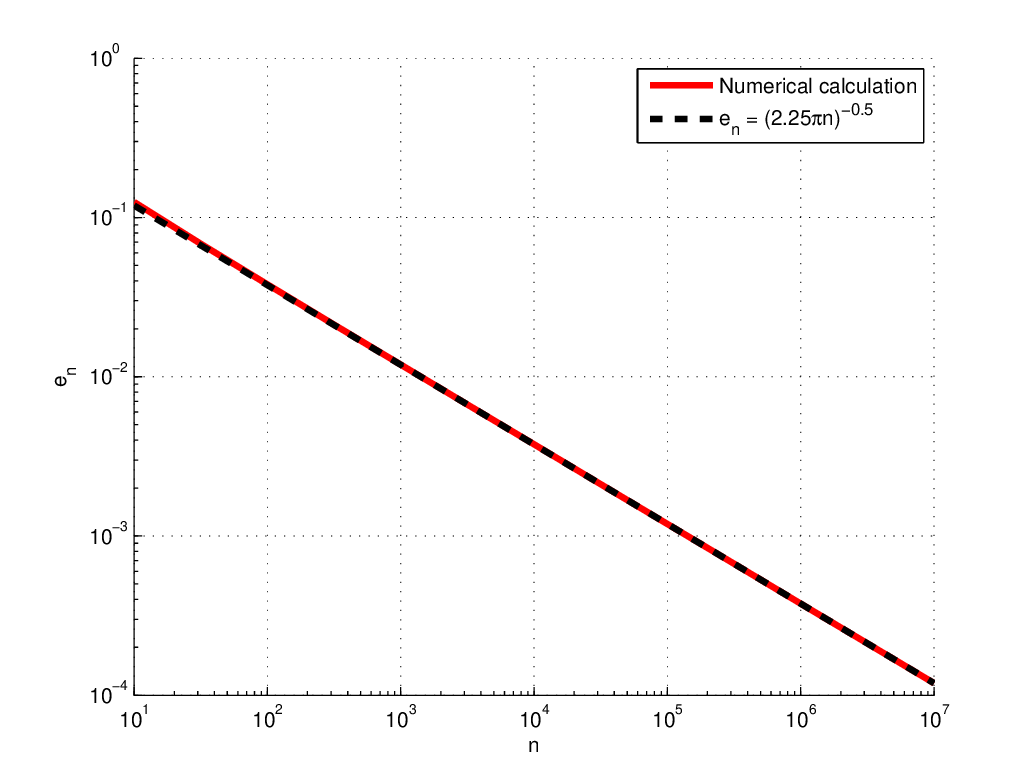}}
\caption{\label{fig_e_n} Numerical calculation and theoretical approximation of $e_n$ as function of $n$.}
\end{figure}
\end{proof}
\chapter{Proof of the Sum of Two Almost Normal RVs Lemma}
\label{app_proof_lemma_sum_of_almost_normal_RVs}
\begin{proof}[Proof of Lemma \ref{lem_sum_of_almost_normal_RVs}]
$X_1$ and $X_2$ are independent. Hence, by definition, the CDF of $Y$ is given by
\begin{align}
\begin{aligned}
F_Y(y) &\triangleq Pr\{ Y \leq y \}
\\     &= Pr\{ X_1 + X_2 \leq y \}
\\     &= \int_{-\infty}^{y}{f_{X_1}(x) \cdot F_{X_2}(y - x)}dx.
\end{aligned}
\end{align}
By the assignment of $f_{X_1}(x)$ and $F_{X_2}(x)$ given by the lemma, we can obtain the following:
\begin{align}
\begin{aligned}
F_Y(y) &= \int_{-\infty}^{y}{N(0,\sigma_1^2) \cdot F_{N(0,\sigma_2^2)}(y - x)}dx
\\     &+ O\left( \int_{-\infty}^{y}{ e_n(x) \cdot F_{N(0,\sigma_2^2)}(y - x)}dx \right)
\\     &+ O\left( \int_{-\infty}^{y}{ \frac{N(0,\sigma_1^2)}{\sqrt{n}}}dx \right)
\\     &+ O\left( \int_{-\infty}^{y}{ \frac{e_n(x)}{\sqrt{n}} }dx \right).
\end{aligned}
\end{align}
Since, $F_{N(0,\sigma_y^2)}(y) = \int_{-\infty}^{y}{N(0,\sigma_1^2) \cdot F_{N(0,\sigma_2^2)}(y - x)}dx$, we can get
\begin{align}
\begin{aligned}
|F_Y(y) - F_{N(0,\sigma_y^2)}(y)| &\leq O\left( \int_{-\infty}^{\infty}{ |e_n(x)| }dx \right)
+ O\left( \int_{-\infty}^{\infty}{ \frac{N(0,\sigma_1^2)}{\sqrt{n}}}dx \right)
+ O\left( \int_{-\infty}^{\infty}{ \frac{|e_n(x)|}{\sqrt{n}} }dx \right)
\\ &= O\left(\frac{1}{\sqrt{n}}\right) + O\left(\frac{1}{\sqrt{n}}\right) + O\left(\frac{1}{n}\right) = O\left(\frac{1}{\sqrt{n}}\right),
\end{aligned}
\end{align}
which completes the proof of \eqref{eq_sum_of_almost_normal_CDF}.
\end{proof}
\chapter{Lemma D.1}
\label{app_lemma_47_in_Polyanskiy}
\begin{lem}
\label{lem_47_in_Polyanskiy}
Let $Z_1,Z_2,\dots,Z_n$ be independent random variables, $\sigma^2=\sum_{i=1}^{n}{Var(Z_i)}$ be non-zero and $T=\sum_{i=1}^{n}{E\{|Z_i-E\{Z_i\}|^3\}}<\infty$; then for any $A$
\begin{equation}
E\left\{e^{-\sum_{i=1}^{n}{Z_i}}1_{\left\{\sum_{i=1}^{n}{Z_i}>A\right\}}\right\}  \leq 2\left(\frac{\ln(2)}{\sqrt{2\pi}} + \frac{12T}{\sigma^2}\right)\frac{1}{\sigma}e^{-A}.
\end{equation}
\end{lem}
\begin{proof}
See \cite[Lemma 47]{Polyanskiy}.
\end{proof}
\chapter{The Channel Output Given CSI Distribution Lemma}
\label{app_lemma_Y_given_H_pdf}
\begin{lem}
\label{lem_Y_given_H_PDF}
Suppose that $Y = H\cdot X + Z$, where $X\sim U(-\frac{a}{2},\frac{a}{2})$ and $Z\sim N(0, \sigma^2)$ are independent RVs. If $H$ is also a random variable independent of $X$ and $Z$, then
\begin{equation}
\label{eq_Y_given_H_PDF}
f(y|h) = \frac{1}{ah}\left(Q\left(\frac{y}{\sigma} - \frac{ah}{2\sigma}\right) - Q\left(\frac{y}{\sigma} + \frac{ah}{2\sigma}\right)\right).
\end{equation}
\end{lem}
\begin{proof}
\begin{align}
\label{align_Y_given_H_PDF_prrof}
\begin{aligned}
f(y|h) &= \int_{-\infty}^{\infty}{f(y,x|h)}dx
\\     &= \int_{-\infty}^{\infty}{f(x)f(y|x,h)}dx
\\     &= \int_{-\frac{a}{2}}^{\frac{a}{2}}{\frac{1}{a}\cdot\frac{1}{\sqrt{2\pi\sigma^2}}e^{-\frac{(y-hx)^2}{2\sigma^2}}}dx
\\     &= \int_{-\frac{ah}{2}}^{\frac{ah}{2}}{\frac{1}{ah}\cdot\frac{1}{\sqrt{2\pi\sigma^2}}e^{-\frac{(x-y)^2}{2\sigma^2}}}dx
\\     &= \frac{1}{ah}\left( \int_{-\frac{ah}{2}}^{\infty}{\frac{1}{\sqrt{2\pi\sigma^2}}e^{-\frac{(x-y)^2}{2\sigma^2}}}dx - \int_{\frac{ah}{2}}^{\infty}{\frac{1}{\sqrt{2\pi\sigma^2}}e^{-\frac{(x-y)^2}{2\sigma^2}}}dx \right)
\\     &= \frac{1}{ah}\left( Q\left(-\frac{ah}{2\sigma} - \frac{y}{\sigma}\right) - Q\left(\frac{ah}{2\sigma} - \frac{y}{\sigma}\right) \right)
\\     &= \frac{1}{ah}\left( Q\left(\frac{y}{\sigma} - \frac{ah}{2\sigma}\right) - Q\left(\frac{y}{\sigma} + \frac{ah}{2\sigma}\right)\right).
\end{aligned}
\end{align}
\end{proof}
\chapter{Proof of the Information Density's Moments Lemma}
\label{app_proof_lemma_information_density_moments}
\begin{proof}[Proof of Lemma \ref{lem_information_density_moments}]
The information density is given by
\begin{align}
\begin{aligned}
i(x;y,h) &\triangleq \ln\left(\frac{f(x,y,h)}{f(x)f(y,h)}\right)
\\       &= \ln\left(\frac{f(x)f(h)f(y|h,x)}{f(x)f(h)f(y|h)}\right)
\\       &= \ln\left(\frac{f(y|h,x)}{f(y|h)}\right)
\\       &= \ln\left(\frac{f(z=y-hx)}{f(y|h)}\right)
\\       &= \ln\left( \frac{1}{\sqrt{2\pi\sigma^2}}e^{-\frac{z^2}{2\sigma^2}} \right) - \ln\left(\frac{1}{ah}\left(Q\left(\frac{y}{\sigma} - \frac{ah}{2\sigma}\right) - Q\left(\frac{y}{\sigma} + \frac{ah}{2\sigma}\right)\right)\right)
\\       &= \frac{1}{2}\ln\left(\frac{a^2h^2}{2\pi e\sigma^2}\right) - \frac{z^2-\sigma^2}{2\sigma^2} - \ln\left(Q\left(\frac{y}{\sigma} - \frac{ah}{2\sigma}\right) - Q\left(\frac{y}{\sigma} + \frac{ah}{2\sigma}\right)\right)
\\       &= \frac{1}{2}\ln\left(\frac{a^2h^2}{2\pi e\sigma^2}\right) - \frac{z^2-\sigma^2}{2\sigma^2} + e_{a/\sigma}(y,h)
\end{aligned}
\end{align}
where $f(y|h)$ is given by Lemma \ref{lem_Y_given_H_PDF} (see Appendix \ref{app_lemma_Y_given_H_pdf}) and the following definition of
\begin{equation}
e_{a/\sigma}(y,h) \triangleq - \ln\left(Q\left(\frac{y}{\sigma} - \frac{ah}{2\sigma}\right) - Q\left(\frac{y}{\sigma} + \frac{ah}{2\sigma}\right)\right)\geq0.
\end{equation}
Define the three error's moments for $i=1,2,3$ by
\begin{align}
e_{a/\sigma,i} &\triangleq E\{e_{a/\sigma}^{i}(Y,H)\}
\\             &= E\{E\{e_{a/\sigma}^{i}(Y,H)|H\}\}
\\             &= E\{e_{a/\sigma,i}(H)\}
\end{align}
where,
\begin{align}
\begin{aligned}
&e_{a/\sigma,i}(h) \triangleq E\{e_{a/\sigma}^{i}(Y,H)|H=h\}
\\              &= (-1)^i\int_{-\infty}^{\infty}{\frac{1}{ah}\left(Q\left(\frac{y}{\sigma} - \frac{ah}{2\sigma}\right) - Q\left(\frac{y}{\sigma} + \frac{ah}{2\sigma}\right)\right)\ln^{i}\left(Q\left(\frac{y}{\sigma} - \frac{ah}{2\sigma}\right) - Q\left(\frac{y}{\sigma} + \frac{ah}{2\sigma}\right)\right)}dy
\\              &= (-1)^i\frac{\sigma}{ah}\int_{-\infty}^{\infty}{\left(Q\left(y - \frac{ah}{2\sigma}\right) - Q\left(y + \frac{ah}{2\sigma}\right)\right)\ln^i\left(Q\left(y - \frac{ah}{2\sigma}\right) - Q\left(y + \frac{ah}{2\sigma}\right)\right)}dy
\\              &= \frac{\sigma}{ah}\eta_i\left(\frac{ah}{\sigma}\right)
\end{aligned}
\end{align}
and
\begin{equation}
\eta_i\left(\frac{ah}{\sigma}\right) \triangleq (-1)^{i}\int_{-\infty}^{\infty}{\left(Q\left(y - \frac{ah}{2\sigma}\right) - Q\left(y + \frac{ah}{2\sigma}\right)\right)\ln^{i}\left(Q\left(y - \frac{ah}{2\sigma}\right) - Q\left(y + \frac{ah}{2\sigma}\right)\right)}dy
\end{equation}
for $i = 1,2,3$.
As can be seen in Figure \ref{fig_eta_Vs_h}, the function $\eta_i\left(\frac{ah}{\sigma}\right)$ is nonnegative, bounded and asymptotically converges to a constant for any $i = 1,2,3$. The function is also monotonically nondecreasing for $i=1$.
\noindent
For small values of $ah/\sigma$, we can approximate $\eta_i\left(\frac{ah}{\sigma}\right)$ by
\begin{align}
\begin{aligned}
\eta_i\left(\frac{ah}{\sigma}\right) &= -(-1)^{i}\frac{ah}{\sigma}\int_{-\infty}^{\infty}{\frac{Q(y + \frac{ah}{2\sigma}) - Q(y - \frac{ah}{2\sigma})}{\frac{ah}{\sigma}}\ln^{i}\left(-\frac{ah}{\sigma}\frac{Q(y + \frac{ah}{2\sigma}) - Q(y - \frac{ah}{2\sigma})}{\frac{ah}{\sigma}}\right)}dy
\\                        &\approx -(-1)^{i}\frac{ah}{\sigma}\int_{-\infty}^{\infty}{Q'(y)\ln^{i}\left(-\frac{ah}{\sigma}Q'(y)\right)}dy
\\                        &= (-1)^{i}\frac{ah}{\sigma}\int_{-\infty}^{\infty}{\frac{1}{\sqrt{2\pi}}e^{-\frac{y^2}{2}}\ln^{i}\left(\frac{ah}{\sigma}\frac{1}{\sqrt{2\pi}}e^{-\frac{y^2}{2}}\right)}dy
\\                        &= (-1)^{i}\frac{ah}{\sigma}\int_{-\infty}^{\infty}{\frac{1}{\sqrt{2\pi}}e^{-\frac{y^2}{2}}\left(C(ah/\sigma) + \frac{1-y^2}{2}\right)^{i}}dy
\\                        &= (-1)^{i}\frac{ah}{\sigma}E_{N(0,1)}\left\{ \left(C(ah/\sigma) + \frac{1-y^2}{2}\right)^{i} \right\}
\end{aligned}
\end{align}
where,
\begin{equation}
C(ah/\sigma) \triangleq \frac{1}{2}\ln\left(\frac{a^2h^2}{2\pi e\sigma^2 }\right).
\end{equation}
By simple calculation of the moments of a standard normal random variable, we get that for small values of $ah/\sigma$
\begin{align}
\begin{aligned}
\eta_1\left(\frac{ah}{\sigma}\right) &\approx -\frac{ah}{\sigma}C(ah/\sigma),
\\
\eta_2\left(\frac{ah}{\sigma}\right) &\approx \frac{ah}{\sigma}\left(C(ah/\sigma)^2 + \frac{1}{2}\right),
\\
\eta_3\left(\frac{ah}{\sigma}\right) &\approx -\frac{ah}{\sigma}\left(C(ah/\sigma)^3 + \frac{3}{2}C(ah/\sigma) - 1\right).
\end{aligned}
\end{align}
It can be seen in Fig. \ref{fig_eta_Vs_h} that for $ah/\sigma < 1$ the approximations above are very accurate.

First, let us calculate the first order error's moment
\begin{align}
\label{align_ea_definition}
e_{a/\sigma,1} &\triangleq E\left\{e_{a/\sigma}(H)\right\}
\\           &= E\left\{\frac{\sigma}{aH}\eta_1\left(\frac{aH}{\sigma}\right)\right\}
\\           &= \int_{0}^{\infty}{f(h)\frac{\sigma}{ah}\eta_1\left(\frac{ah}{\sigma}\right)}dh
\\           &= \int_{0}^{\frac{\sigma}{a}}{f(h)\frac{\sigma}{ah}\eta_1\left(\frac{ah}{\sigma}\right)}dh + \int_{\frac{\sigma}{a}}^{\infty}{f(h)\frac{\sigma}{ah}\eta_1\left(\frac{ah}{\sigma}\right)}dh.
\label{align_ea_as_sum_of_integrals}
\end{align}
For any regular fading distribution there exists a positive constant $\beta > 0$ s.t. near the origin $f(h)\sim\frac{1}{h^{1-\beta}}$. Moreover, for any PDF there exists a positive constant $\beta^{'} > 0$ s.t. $f(h)\sim\frac{1}{h^{1+\beta^{'}}}$ for large enough $h$. Hence, for large enough $a/\sigma$, we can get the following bounds
\begin{enumerate}
  \item
    \begin{align}
    \int_{0}^{\frac{\sigma}{a}}{f(h)\frac{\sigma}{ah}\eta_1\left(\frac{ah}{\sigma}\right)}dh &= O\left(-\int_{0}^{\frac{\sigma}{a}}{f(h)C(ah/\sigma)}dh\right)
    \\ &= O\left(\int_{0}^{\frac{\sigma}{a}}{\frac{1}{h^{1-\beta}}\ln\left(\frac{ah}{\sigma}\right)}dh\right)
    \\ &= O\left(\int_{0}^{\frac{\sigma}{a}}{\frac{\ln(h)}{h^{1-\beta}}}dh\right) + O\left(\ln\left(\frac{a}{\sigma}\right)\int_{0}^{\frac{\sigma}{a}}{\frac{dh}{h^{1-\beta}}}\right)
    \\ \label{align_bounding_ea_part1}
       &= O\left(\ln\left(\frac{a}{\sigma}\right)\left(\frac{\sigma}{a}\right)^{\beta}\right).
    \end{align}
  \item
    \begin{align}
    \label{align_bounding_eta1}
    \int_{\frac{\sigma}{a}}^{\infty}{f(h)\frac{\sigma}{ah}\eta_1\left(\frac{ah}{\sigma}\right)}dh &\leq  \int_{\frac{\sigma}{a}}^{\infty}{f(h)\frac{\sigma}{ah}M}dh
    \\ &= \int_{\frac{\sigma}{a}}^{h_0}{f(h)\frac{\sigma}{ah}M}dh
        + \int_{h_0}^{h_1}{f(h)\frac{\sigma}{ah}M}dh
        + \int_{h_1}^{\infty}{f(h)\frac{\sigma}{ah}M}dh
    \\ &= O\left(\frac{\sigma}{a}\int_{\frac{\sigma}{a}}^{\infty}{\frac{dh}{h^{2-\beta}}}\right)
        + O\left(\frac{\sigma}{a}\int_{h_0}^{h_1}{\frac{f(h)}{h}dh}\right)
        + O\left(\frac{\sigma}{a}\int_{h_1}^{\infty}{\frac{dh}{h^{2+\beta^{'}}}}\right)
    \\ \label{align_bounding_ea_part2}
       &= O\left(\left(\frac{\sigma}{a}\right)^{\beta}\right)
        + O\left(\frac{\sigma}{a}\right)
        = O\left(\left(\frac{\sigma}{a}\right)^{\min(\beta,1)}\right),
    \end{align}
\end{enumerate}
where \eqref{align_bounding_eta1} is due to the fact that $\eta_1(ah/\sigma)\leq M$ for some positive and finite constant $M$.
From \eqref{align_ea_as_sum_of_integrals}, \eqref{align_bounding_ea_part1}, \eqref{align_bounding_ea_part2} and the fact that $\forall\epsilon>0~\lim_{x\to\infty}\frac{\ln(x)}{x^{\epsilon}}=0$, we get that there exists a constant $0 < \alpha \leq 1$ s.t. the following holds:
\begin{equation}
e_{a/\sigma,1} = O\left(\left(\frac{\sigma}{a}\right)^{\alpha}\right).
\end{equation}
Finally, because of the common properties of $\eta_1(ah/\sigma), \eta_2(ah/\sigma)$ and $\eta_3(ah/\sigma)$, with equivalent calculations, we can get that there exists also a constant $0 < \alpha \leq 1$, s.t the error's moments hold the following:
\begin{equation}
e_{a/\sigma,2} \triangleq E\left\{e_{a/\sigma}^{2}(Y,H)\right\}= E\left\{\frac{\sigma}{aH}\eta_2\left(\frac{aH}{\sigma}\right)\right\} = O\left(\left(\frac{\sigma}{a}\right)^{\alpha}\right),
\end{equation}
\begin{equation}
e_{a/\sigma,3} \triangleq E\left\{|e_{a/\sigma}(Y,H)|^{3}\right\} = E\left\{\frac{\sigma}{aH}\eta_3\left(\frac{aH}{\sigma}\right)\right\} = O\left(\left(\frac{\sigma}{a}\right)^{\alpha}\right).
\end{equation}
\begin{figure}[htp]
\center{\includegraphics[width=0.7\columnwidth]{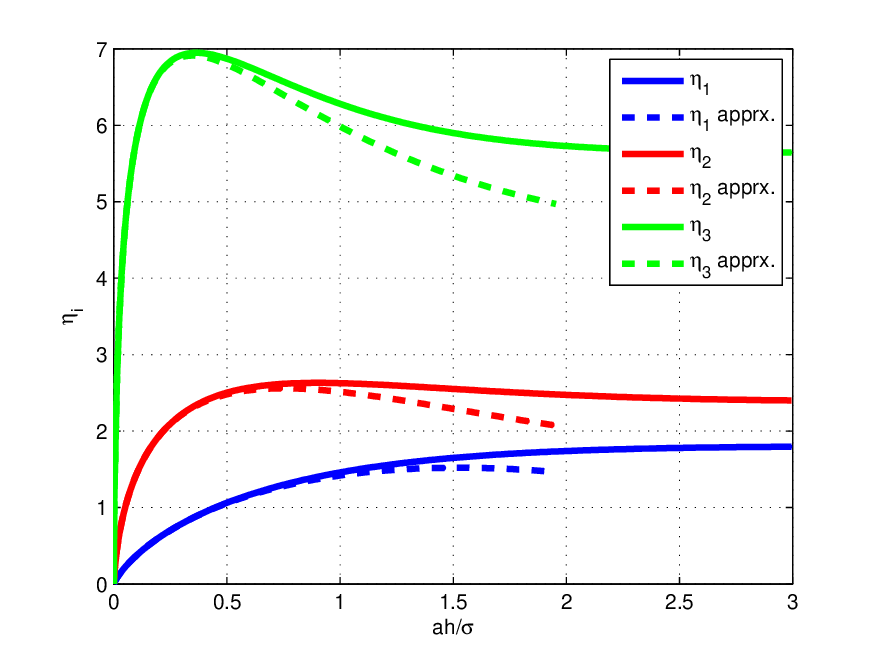}}
\caption{\label{fig_eta_Vs_h} $\eta_{i}(ah/\sigma)$ and its approximation for small values of $ah/\sigma$.}
\end{figure}

Now, let us turn to calculate the information density's moments.
\section{Calculating the Mutual Information}
The mean of the information density is given by
\begin{align}
I(X;Y,H) &\triangleq E\{i(X;Y,H)\}
\\       &= E\left\{\frac{1}{2}\ln\left(\frac{a^2H^2}{2\pi e\sigma^2}\right)\right\} - E\left\{\frac{Z^2-\sigma^2}{2\sigma^2}\right\} + E\left\{e_{a/\sigma}(Y,H)\right\}
\\       &= E\left\{\frac{1}{2}\ln\left(\frac{a^2H^2}{2\pi e\sigma^2}\right)\right\} + e_{a/\sigma,1}
\\       &= E\left\{\frac{1}{2}\ln\left(\frac{a^2H^2}{2\pi e\sigma^2}\right)\right\} +
O\left(\left(\frac{\sigma}{a}\right)^{\alpha}\right).
\end{align}
\section{Calculating the Information Density Variance}
The variance of the information density is given by
\begin{align}
Var(i(X;Y,H)) &= Var\left(\frac{1}{2}\ln\left(\frac{a^2H^2}{2\pi e\sigma^2}\right) - \frac{Z^2-\sigma^2}{2\sigma^2} + e_{a/\sigma}(Y,H)\right)
\\            &= Var\left(\frac{1}{2}\ln\left(H^2\right) - \frac{Z^2}{2\sigma^2} + e_{a/\sigma}(Y,H)\right)
\\            &= Var\left(\frac{1}{2}\ln\left(H^2\right)\right) + Var\left(\frac{Z^2}{2\sigma^2}\right) + Var\left(e_{a/\sigma}(Y,H)\right)
\\            &+ 2Cov\left(\frac{1}{2}\ln\left(H^2\right),e_{a/\sigma}(Y,H)\right) -2Cov\left(\frac{Z^2}{2\sigma^2},e_{a/\sigma}(Y,H)\right)
\\            \label{align_var_calc}
              &= \frac{1}{2} + Var\left(\delta(H)\right) + \Delta(a/\sigma)
\end{align}
where,
\begin{align}
\begin{aligned}
\label{align_line_in_var_bound}
\Delta(a/\sigma) &\triangleq Var\left(e_{a/\sigma}(Y,H)\right) + 2Cov\left(\frac{1}{2}\ln\left(H^2\right),e_{a/\sigma}(Y,H)\right) -2Cov\left(\frac{Z^2}{2\sigma^2},e_{a/\sigma}(Y,H)\right)
\\ &= e_{a/\sigma,2} - e_{a/\sigma,1}^{2} + E\left\{\ln\left(H^2\right)e_{a/\sigma}(Y,H)\right\} - E\left\{\ln\left(H^2\right)\right\}e_{a/\sigma,1}
\\ &- E\left\{\frac{Z^2}{\sigma^2}e_{a/\sigma}(Y,H)\right\} + E\left\{\frac{Z^2}{\sigma^2}\right\}e_{a/\sigma,1}
\\ &= O(e_{a/\sigma,2}) + O(e_{a/\sigma,1}) + O\left(E\left\{\ln\left(H^2\right)e_{a/\sigma}(Y,H)\right\}\right) + O\left(E\left\{\frac{Z^2}{\sigma^2}e_{a/\sigma}(y,h)\right\}\right).
\end{aligned}
\end{align}
By the Cauchy Schwarz inequality,
\begin{equation}
\label{eq_cauchy_schwarz_1}
\left|E\left\{\ln\left(H^2\right)e_{a/\sigma}(Y,H)\right\}\right| \leq \sqrt{E\left\{\ln^{2}\left(H^2\right)\right\}e_{a/\sigma,2}} = O\left(\sqrt{e_{a/\sigma,2}}\right)
\end{equation}
and
\begin{equation}
\label{eq_cauchy_schwarz_2}
\left|E\left\{\frac{Z^2}{\sigma^2}e_{a/\sigma}(Y,H)\right\}\right| \leq \sqrt{E\left\{\frac{Z^4}{\sigma^4}\right\}e_{a/\sigma,2}} = O\left(\sqrt{e_{a/\sigma,2}}\right).
\end{equation}
Combining \eqref{align_line_in_var_bound}, \eqref{eq_cauchy_schwarz_1} and \eqref{eq_cauchy_schwarz_2} we get
\begin{equation}
\label{eq_var_error_bound}
\Delta(a/\sigma) = O\left(\sqrt{e_{a/\sigma,2}}\right) = O\left(\left(\frac{\sigma}{a}\right)^{\frac{\alpha}{2}}\right).
\end{equation}
From \eqref{align_var_calc} and \eqref{eq_var_error_bound} we get the desired result:
\begin{equation}
Var(i(X;Y,H)) = \frac{1}{2} + Var\left(\delta(H)\right) + O\left(\left(\frac{\sigma}{a}\right)^{\frac{\alpha}{2}}\right).
\end{equation}
\section{Bounding the Information Density's Absolute third Order Moment}
The absolute third order moment of the information density is given by
\begin{align}
\rho_3 &\triangleq E\left\{|i(X;Y,H) - I(X;Y,H)|^3\right\}
\\     &= E\left\{\left|\frac{1}{2}\ln\left(\frac{a^2H^2}{2\pi e\sigma^2}\right) - \frac{Z^2-\sigma^2}{2\sigma^2} + e_{a/\sigma}(Y,H) - E\left\{\frac{1}{2}\ln\left(\frac{a^2H^2}{2\pi e\sigma^2}\right)\right\} - e_{a/\sigma,1}\right|^3\right\}
\\     \label{align_bounding_abs_3_moment}
       &\leq \left( \Big\| \frac{1}{2}\ln\left(H^2\right) - E\left\{\frac{1}{2}\ln\left(H^2\right)\right\} \Big\|_3 + \Big\| \frac{Z^2-\sigma^2}{2\sigma^2} \Big\|_3 + \Big\| e_{a/\sigma}(Y,H)\Big\|_3 + e_{a/\sigma,1} \right)^3,
\end{align}
where the last inequality is due to the Minkowski inequality and the definition of $\left\|X\right\|_3 \triangleq E\left\{ |X|^3 \right\}^{\frac{1}{3}}$.
By definition we get
\begin{equation}
\label{eq_norm_3_of_ea_3}
\Big\| e_{a/\sigma}(Y,H)\Big\|_3 = \left(E\left\{e_{a/\sigma}^{3}(Y,H)\right\}\right)^{\frac{1}{3}} = e_{a/\sigma,3}^{\frac{1}{3}} = O\left(\left(\frac{\sigma}{a}\right)^{\frac{\alpha}{3}}\right).
\end{equation}
From \eqref{align_bounding_abs_3_moment} and \eqref{eq_norm_3_of_ea_3} we get the desired result
\begin{equation}
\rho_3 \leq A + O\left(\ln\left(\frac{a}{\sigma}\right)\left(\frac{\sigma}{a}\right)^{\frac{\alpha}{3}}\right)
\end{equation}
for some positive and finite constant $A$, or simply $\rho_3 < \infty$.
\end{proof}
\chapter{Tiling}
\label{app_tiling}
We now turn to construct an IC with average error probability which is upper bounded by $\epsilon$, denoted by $S(n,\epsilon)$, from the FC  $S(n,\epsilon',a/\sigma)$. It is assumed that $S(n,\epsilon',a/\sigma)$ has an average error probability which is upper bounded by $\epsilon'$ (using the suboptimal decoder on which the dependence testing bound is based), and its NLD, $\delta(n,\epsilon',a/\sigma)$ in $\text{Cb}(a)$, holds the following:
\begin{equation}
\delta(n,\epsilon',a/\sigma) = \delta^* - \sqrt{ \frac{V}{n} }Q^{-1}(\epsilon') +
O\left(\frac{1}{n} + \frac{1}{\sqrt{n}}\left(\frac{\sigma}{a}\right)^{\frac{\alpha}{2}}  + \left(\frac{\sigma}{a}\right)^{\alpha}\right).
\end{equation}

Define the IC $S(n,\epsilon)$ as an infinite replication of $S(n,\epsilon',a/\sigma)$ with spacing of $b$ between every two copies as follows:
\begin{equation}
S(n,\epsilon) \triangleq \left\{ s + I\cdot(a+b): s\in S(n,\epsilon',a/\sigma), I\in \mathbb{Z}_n \right\}
\end{equation}
where $\mathbb{Z}_n$ denotes the integer lattice of dimension $n$. This tiling operation is illustrated in Figure \ref{fig_tiling}.
The NLD of the IC is given by
\begin{align}
\begin{aligned}
\label{align_IC_NLD}
\delta(n,\epsilon,a/\sigma,b) &\triangleq \frac{1}{n}\ln\left(\frac{M(n,\epsilon',a/\sigma)}{(a+b)^n}\right)
\\                            &= \frac{1}{n}\ln\left(\frac{M(n,\epsilon',a/\sigma)}{a^n}\right) - \ln\left(1+\frac{b}{a}\right)
\\                            &= \delta(n,\epsilon',a/\sigma) - \ln\left(1+\frac{b}{a}\right),
\end{aligned}
\end{align}
where $M(n,\epsilon',a/\sigma)$ is the number of codewords of the FC.

Define the faded FC in the receiver, given the CSI, as
\begin{equation}
S(n,\epsilon',a/\sigma)_{\mathbf{H}} \triangleq \{ \mathbf{H}\cdot s: s\in S(n,\epsilon',a/\sigma) \}
\end{equation}
where $\mathbf{H}=\text{diag}(H_1,H_2,\dots,H_n)$.
In the receiver, we get the following IC:
\begin{equation}
S(n,\epsilon)_{\mathbf{H}} \triangleq \left\{ s_{\text{rc}} + \mathbf{H} \cdot I\cdot (a+b): s_{\text{rc}}\in S(n,\epsilon',a/\sigma)_{\mathbf{H}}, I\in \mathbb{Z}_n \right\},
\end{equation}
which is a tiled version of the faded FC.

Now consider the ML error probability of a point $s_{\text{rc}}\in S(n,\epsilon)_{\mathbf{H}}$, given the CSI $\mathbf{H}$ at the receiver, denoted by $P_{e,ML}^{IC}(s_{\text{rc}}|\mathbf{H})$. In the same manner, $P_{e,ML}^{FC}(s_{\text{rc}}|\mathbf{H})$ will denote the ML error probability for any $s_{\text{rc}}\in S(n,\epsilon',a/\sigma)_{\mathbf{H}}$.
If $\mathbf{H}$ is a too ``strong'' channel fading realization then we will declare an error. Formally, if $H_{\min} \leq h_{\min}^{*}$ for some arbitrary positive constant $h_{\min}^{*}$, where $H_{\min} \triangleq \min\{H_1,H_2,\dots,H_n\}$, then we will declare an error. Otherwise, this error probability equals the probability of decoding by mistake to another codeword from the same copy of the faded FC $S(n,\epsilon',a/\sigma)_{\mathbf{H}}$ or to a codeword in another copy. Hence, by using the union bound, we obtain the following:
\begin{align}
P_{e,ML}^{IC}(s_{\text{rc}}|\mathbf{H}) &\leq \left( P_{e,ML}^{FC}(s_{\text{rc}}|\mathbf{H}) + \sum_{i=1}^{n}{2Q\left(\frac{H_i\cdot b}{2\sigma}\right)} \right)\cdot 1_{\left\{H_{\min} > h_{\min}^{*}\right\}} + 1_{\left\{H_{\min} \leq h_{\min}^{*}\right\}}
\\ &\leq P_{e,ML}^{FC}(s_{\text{rc}}|\mathbf{H}) + 2nQ\left(\frac{h_{\min}^{*}\cdot b}{2\sigma}\right) + 1_{\left\{H_{\min} \leq h_{\min}^{*}\right\}}.
\end{align}
The average error probability over $S(n,\epsilon)_{\mathbf{H}}$ and $\mathbf{H}$ is then upper bounded by
\begin{equation}
\label{eq_IC_error_pro_ub_1}
P_{e,ML}^{IC} \leq P_{e,ML}^{FC} + 2nQ\left(\frac{h_{\min}^{*}\cdot b}{2\sigma}\right) + Pr\left\{H_{\min} \leq h_{\min}^{*}\right\}.
\end{equation}
Trivially we have
\begin{equation}
\label{eq_FC_error_pro_ub}
P_{e,ML}^{FC} \leq P_{e,DT}^{FC} \leq \epsilon',
\end{equation}
where $P_{e,DT}^{FC}$ is the average error probability of the FC using the suboptimal decoder on which the dependence testing bound is based.
\newline
By the union bound
\begin{equation}
\label{eq_error_pro_according_h_min_ub}
Pr\left\{H_{\min} \leq h_{\min}^{*}\right\} \leq nPr\left\{H \leq h_{\min}^{*}\right\}.
\end{equation}
Combining \eqref{eq_IC_error_pro_ub_1}, \eqref{eq_FC_error_pro_ub} and \eqref{eq_error_pro_according_h_min_ub} we get that
\begin{equation}
\label{eq_IC_error_pro_ub_2}
P_{e,ML}^{IC} \leq \epsilon' + 2nQ\left(\frac{h_{\min}^{*}\cdot b}{2\sigma}\right) + nPr\left\{H \leq h_{\min}^{*}\right\} \triangleq \epsilon.
\end{equation}

From \eqref{align_IC_NLD} and \eqref{eq_IC_error_pro_ub_2} we can see that for any large enough $n$, if we choose small enough $h_{\min}^{*}$, large enough $b$ relative to $h_{\min}^{*}/\sigma$ and large enough $a$ relative to $b$, then we will get an IC with average error probability which is upper bounded by $\epsilon$ and arbitrarily close to $\epsilon'$, and NLD which equals $\delta(n,\epsilon) \triangleq \delta^* - \sqrt{ \frac{V}{n} }Q^{-1}(\epsilon) + O\left(\frac{1}{n}\right)$.

Let us demonstrate this idea by an example. Suppose a regular fading distribution s.t. $f(h)\sim\frac{1}{h^{1-\alpha}}$ for small enough positive $h$ and for some $\alpha > 0$. Hence, $Pr\left\{H \leq h_{\min}^{*}\right\}=O\left(\left(h_{\min}^*\right)^{\alpha}\right)$. If we choose $h_{\min}^*(n) = \frac{1}{n^{\frac{2}{\alpha}}}, b(n) = \sigma\cdot n^{1+\frac{2}{\alpha}}$ and $a(n) = \sigma\cdot n^{2+\frac{2}{\alpha}}$, then we will get:
\begin{align}
\label{eq_IC_error_pro_ub_3}
\begin{aligned}
P_{e,ML}^{IC}
   &\leq \epsilon
\\ &\triangleq \epsilon' + 2nQ\left(\frac{h_{\min}^{*}(n)\cdot b(n)}{2\sigma}\right) + nPr\left\{H_{\min} \leq h_{\min}^{*}(n)\right\}
\\ &\leq \epsilon' + nQ\left(\frac{n}{2}\right) + O\left(\frac{1}{n}\right)
\\ &\leq \epsilon' + ne^{-\frac{n^2}{8}} + O\left(\frac{1}{n}\right)
\\ &= \epsilon' + O\left(\frac{1}{n}\right)
\end{aligned}
\end{align}
and
\begin{align}
\begin{aligned}
\delta(n,\epsilon,a(n)/\sigma,b(n))
   &= \delta(n,\epsilon',a(n)/\sigma) - \ln\left(1+\frac{b(n)}{a(n)}\right)
\\ &= \delta\left(n,\epsilon - O\left(1/n\right),a(n)/\sigma\right) + O\left(\frac{1}{n}\right)
\\ &= \delta^* - \sqrt{ \frac{V}{n} }Q^{-1}\left(\epsilon - O\left(1/n\right)\right)
+ O\left(\frac{1}{n} + \frac{1}{\sqrt{n}}\left(\frac{\sigma}{a(n)}\right)^{\frac{\alpha}{2}} + \left(\frac{\sigma}{a(n)}\right)^{\alpha}\right)
\\ &= \delta^* - \sqrt{ \frac{V}{n} }Q^{-1}(\epsilon) + O\left(\frac{1}{n}\right) \triangleq \delta(n,\epsilon).
\end{aligned}
\end{align}
Note that this operation can be done for any fixed $\epsilon > 0$ (or equivalently for any $\epsilon' > 0$).
\begin{figure}[htp]
\center{\includegraphics[width=0.5\columnwidth]{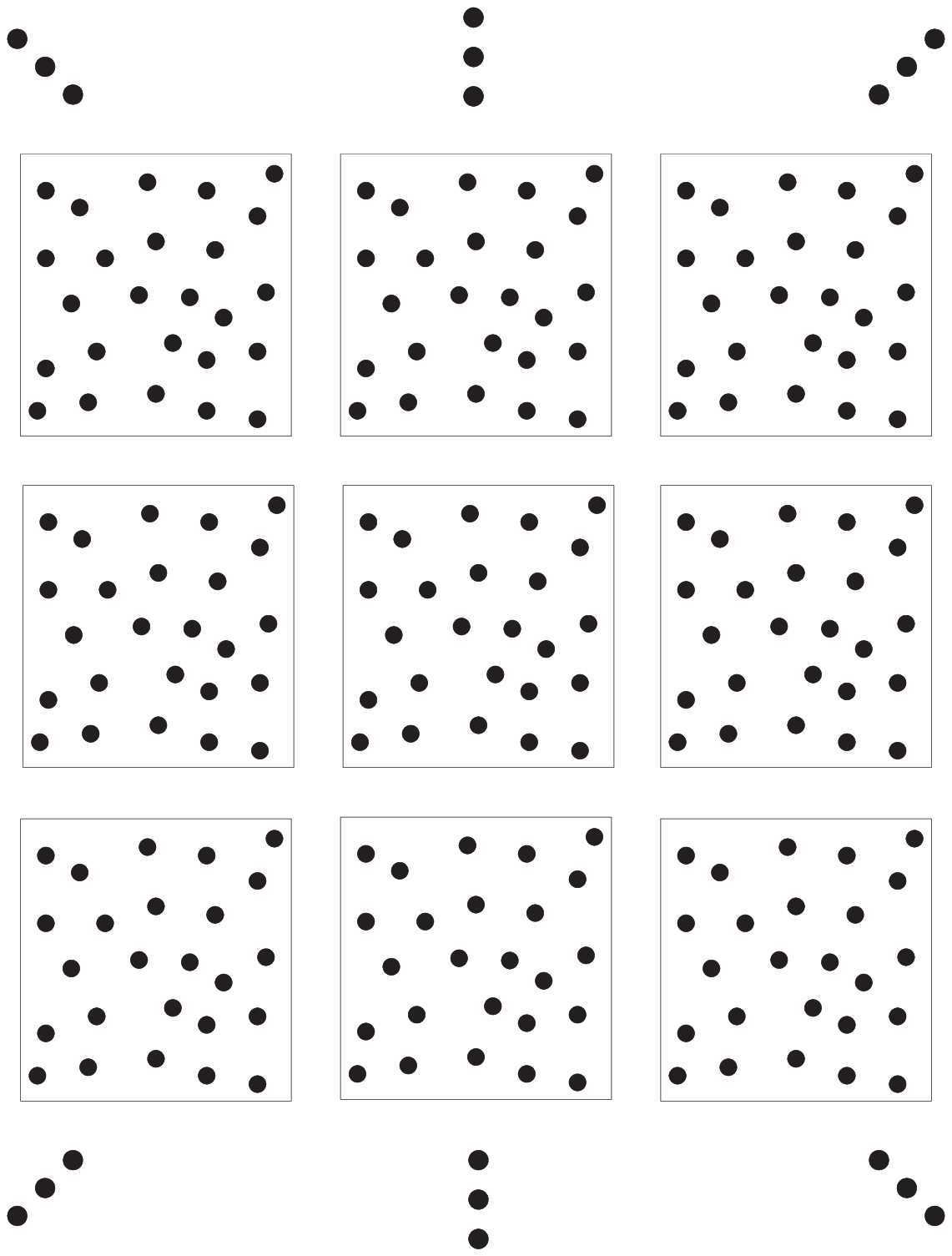}}
\caption{\label{fig_tiling} An illustration of the tiling operation.}
\end{figure}
\chapter{Proof of the Sufficient Typicality Decoder Based Bound Lemma}
\label{app_proof_of_the_sufficient_typicality_decoder_based_bound}
\begin{proof}[Proof of Lemma \ref{lem_sufficient_typicality_decoder_based_bound}]
This Lemma simplifies the typicality decoder based bound of Theorem \ref{thm_typicality_decoder_based_bound}. This simplification is done by upper bounding the third term of the RHS of \eqref{align_typicality_decoder_based_bound} according to the union bound as follows
\begin{align}
\begin{aligned}
Pr\left\{H_{\max} > g_{\max}(n) \cup H_{\min} < g_{\min}(n) \right\} &\leq Pr\left\{ H_{\min} < g_{\min}(n) \right\} + Pr\left\{ H_{\max} > g_{\max}(n) \right\} \\
&\leq nPr\left\{ H < g_{\min}(n) \right\} + nPr\left\{ H > g_{\max}(n) \right\}.
\end{aligned}
\end{align}
and by choosing specific series of $g_{\min}(n)$ and $g_{\max}(n)$.

Let us choose $g_{\min}(n)$ to be a monotonic decreasing series s.t. $\lim_{n\to\infty}g_{\min}(n)=0$.
Since we assume \emph{regular fading distribution}, then for small enough $g_{\min}(n)$ (or large enough $n$) we have
\begin{align}
Pr\left\{ H < g_{\min}(n) \right\}
 \leq \int_{0}^{g_{\min}(n)}h^{\alpha-1}dh
 \leq C'{g_{\min}(n)}^{\alpha},
\end{align}
for some constants $\alpha,C'>0$.

\noindent Using the Markov inequality we have
\begin{align}
Pr\left\{ H > g_{\max}(n) \right\} = Pr\left\{ H^2 > {g_{\max}(n)}^2 \right\}
 &\leq \frac{E\left\{H^2\right\}}{{g_{\max}(n)}^{2}}=\frac{1}{{g_{\max}(n)}^{2}}.
\end{align}
By choosing for example $g_{\min}(n)=n^{-3/\alpha}$ and $g_{\max}(n)=n^{3/2}$, and by combining all the above, the upper bound error probability of Theorem \ref{thm_typicality_decoder_based_bound} can be simplified by the following
\begin{align}
\begin{aligned}
P_e(\Lambda) &\leq Pr\left\{\|\mathbf{z}\|>r\right\} + \gamma V_n r_0^n + \frac{C}{n^2},
\end{aligned}
\end{align}
for some constant $C>0$.
\end{proof}
\chapter{Error Exponents for Scalar Fading Channels}
The general formula of Gallager's random coding error exponent for scalar and real fading channels is given by \cite{Gallager}\cite{Ericson}
\begin{equation}
\label{eq_scalar_error_exponent_definition}
E_r(R) = \max_{\rho \in [0,1]}\left(\max_{f(x)}E_0(f(x),\rho) - \rho R\right)
\end{equation}
where,
\begin{equation}
E_0(f(x),\rho) = -\ln E\left\{ \int{ \left[ \int{f(x)f(y|x,h)^{\frac{1}{1+\rho}}dx} \right]^{1+\rho} dy} \right\}.
\end{equation}
With a slightly abuse of notations we will denote for simplicity $E_0(\rho)$ instead of $E_0(f(x),\rho)$.

Let us denote by $\rho^* = \rho(R)$, the value of $\rho$ that optimizes \eqref{eq_scalar_error_exponent_definition} for a given rate $R$. In addition, let's denote
by $R_{\text{cr}}$, the maximal rate such that $\rho^*$ equals 1.
Then, $\rho^*$ is given by the solution of the following equation
\begin{equation}
\frac{\partial E_0(\rho)}{\partial \rho}\Big|_{\rho^*} = R, \nonumber
\end{equation}
for any rate $R_{\text{cr}} \leq R \leq C$.
Hence, we get by definition:
\begin{equation}
E_r(R) = E_0(\rho^*) - \rho^*R =
\left\{ \begin{array}{ll}
E_0(1)-R,
& 0 \leq R \leq R_{\text{cr}}
\\
E_0(\rho^*) - \rho^*\frac{\partial E_0(\rho)}{\partial \rho}\Big|_{\rho^*},
& R_{\text{cr}}\leq R \leq C
\end{array} \right.. \nonumber
\end{equation}

In Sections \ref{app_scalar_error_exponent_normal_prior_sec} and \ref{app_scalar_error_exponent_uniform_prior_sec} we will analyze the random coding error exponent
behavior for scalar real fading channels at the high SNR regime, with normal and uniform input distributions, respectively.
Note that in \cite{Ericson} this error exponent was analyzed with the optimal uniform input distribution on a ``thin spherical shell''.
While in Section \ref{app_scalar_error_exponent_normal_prior_sec} only the behavior near the capacity will be analyzed, in Section
\ref{app_scalar_error_exponent_uniform_prior_sec} we will derive approximations for the random coding error exponent at any rate.
Finally, in Section \ref{app_scalar_complex_error_exponent_sec} we will mention some notes about the random coding error exponent,
for scalar complex fading channels.
\section{Normal Input Distribution}
\label{app_scalar_error_exponent_normal_prior_sec}
In \cite{Ericson} Ericson analyzed the error exponent of the scalar fading channel with the optimal uniform input distribution on a ``thin spherical shell''. By the assignment of $r=0$, in his expressions, we get the suboptimal error exponent with normal input distribution, which is more easier to analyze.
In that case the capacity equals $C = E\left\{ \frac{1}{2}\ln\left(1 + H^2 \cdot SNR \right) \right\}$, and the error exponent factor $E_0(\cdot)$ is given by the following:
\begin{equation}
\label{eq_E_0_G_scalar}
E_{0,G}(\rho) \triangleq E_0(N(0,P),\rho) = -\ln E\left\{ \left( 1 + H^2 \cdot \frac{SNR}{1+\rho} \right)^{-\frac{\rho}{2}} \right\}.
\end{equation}
In the high SNR regime we can approximate the capacity by $C = E\left\{ \frac{1}{2}\ln\left(H^2 \cdot SNR \right) \right\}$ and \eqref{eq_E_0_G_scalar} by the following:
\begin{align}
E_{0,G}(\rho) &\approx -\ln E\left\{ \left(H^2 \cdot \frac{SNR}{1+\rho} \right)^{-\frac{\rho}{2}} \right\} \nonumber
\\            &= -\frac{\rho}{2} \ln(1+\rho) -\ln E\left\{ \left( H^2 \cdot SNR \right)^{-\frac{\rho}{2}} \right\}. \nonumber
\end{align}
The derivative of $E_{0,G}(\rho)$ w.r.t. $\rho$ gives us the following:
\begin{align}
\label{align_scalar_diff_E_0_G_1}
\frac{\partial E_{0,G}(\rho)}{\partial \rho} &\approx -\frac{1}{2}\ln(1+\rho) - \frac{\rho}{2(1+\rho)} + \frac{1}{2} \frac{ E\left\{ \ln\left( H^2 \cdot SNR \right) \cdot H^{-\rho} \right\} }{ E\left\{ H^{-\rho} \right\} }.
\end{align}
Since near the capacity $\rho^* \to 0$ we can use the following first order Taylor's approximations around zero:
\begin{enumerate}
\item $\ln(1+\rho^*) \approx \rho^*$
\item $\frac{\rho^*}{1+\rho^*} \approx \rho^*$
\item $e^{-\rho^*\ln(H)} \approx 1 - \rho^*\ln(H)$
\end{enumerate}
to get the following approximation of \eqref{align_scalar_diff_E_0_G_1} near the capacity:
\begin{align}
\label{align_scalar_diff_E_0_G_2}
\frac{\partial E_{0,G}(\rho)}{\partial \rho}\Big|_{\rho^*} &\approx -\rho^* + \frac{1}{2}\ln(SNR) - \frac{ \rho^*E\left\{ \ln^2(H) \right\} - E\left\{ \ln(H) \right\} }{ 1 - \rho^*E\left\{ \ln(H) \right\} }.
\end{align}
Using the first order Taylor's approximation of $g(\rho) \triangleq \frac{\rho\cdot a - b}{1-\rho\cdot b}$ around zero we get $g(\rho) \approx -b + \rho\cdot(a - b^2)$. By the assignment of it in \eqref{align_scalar_diff_E_0_G_2} we obtain:
\begin{align}
\label{align_scalar_diff_E_0_G_3}
\frac{\partial E_{0,G}(\rho)}{\partial \rho}\Big|_{\rho^*} &\approx E\left\{ \frac{1}{2}\ln\left(H^2 \cdot SNR \right) \right\} - \rho^*\left(1 + Var\left( \ln\left( H \right) \right) \right) \nonumber
\\ &\approx C - \rho^*\left( 1 + Var\left(\frac{1}{2} \ln\left( H^2 \right) \right) \right).
\end{align}
Hence, near the capacity the optimization factor can be approximated by the following:
\begin{equation}
\rho^* \approx \frac{C-R}{1 + Var\left(\frac{1}{2} \ln\left( H^2 \right) \right)}. \nonumber
\end{equation}
By integrating \eqref{align_scalar_diff_E_0_G_3} w.r.t. $\rho^*$ and the assignment of $\rho^*$ we obtain:
\begin{align}
E_r(R) &= \int_{0}^{\rho^*}{\frac{\partial E_{0,G}(\rho)}{\partial \rho}d\rho} - \rho^*R \nonumber
\\   &\approx \frac{(C-R)^2}{2V_{UB}}, \nonumber
\end{align}
where $V_{UB} \triangleq 1 + Var\left(\frac{1}{2} \ln\left( H^2 \right) \right)$ and $C = E\left\{ \frac{1}{2}\ln\left(H^2 \cdot SNR \right) \right\}$.

Since the uniform distribution on a ``thin spherical shell'' is the optimal input distribution that maximizes the Gallager's error exponent of the scalar real fading channel,
and not the normal distribution, we got only an upper bound of the channel dispersion from the analysis, $V < V_{UB}$.
\section{Uniform Input Distribution}
\label{app_scalar_error_exponent_uniform_prior_sec}
Here we use uniform input distribution, namely $X\sim U(-\nicefrac{a}{2},\nicefrac{a}{2})$. Hence,
\begin{align}
E_{0,U}(\rho) &\triangleq E_0\left(U(-\nicefrac{a}{2},\nicefrac{a}{2}),\rho\right) \nonumber
\\ &= -\ln E\left\{ \int_{-\infty}^{\infty}{ \left[ \int_{-\frac{a}{2}}^{\frac{a}{2}}{\frac{1}{a} \left(\frac{1}{\sqrt{2\pi\sigma^2}}e^{-\frac{(y-Hx)^2}{2\sigma^2}}\right)^{\frac{1}{1+\rho}} dx} \right]^{1+\rho} dy} \right\} \nonumber
\\ &= (1+\rho)C - \frac{1+\rho}{2}E\left\{\ln\left(\frac{H^2}{e}\right)\right\} - \frac{1+\rho}{2}\ln(1+\rho) - I(\rho), \nonumber
\end{align}
where,
\begin{equation}
C \triangleq E\left\{\frac{1}{2} \ln\left(\frac{a^2H^2}{2 \pi e \sigma^2}\right) \right\}, \nonumber
\end{equation}
and
\begin{equation}
I(\rho) \triangleq \ln E\left\{ \frac{1}{H^{1+\rho}}\int_{-\frac{aH}{\sigma\sqrt{1+\rho}}}^{\infty}\sqrt{\frac{1+\rho}{2\pi}} \left(Q(x) - Q\left(x+\frac{2aH}{\sigma\sqrt{1+\rho}}\right)\right)^{1+\rho}dx \right\}. \nonumber
\end{equation}
For large enough $a/\sigma$, we can obtain the following approximation:
\begin{align}
I(\rho) &= \ln E\left\{ \frac{1}{H^{1+\rho}}\int_{-\frac{aH}{\sigma\sqrt{1+\rho}}}^{\infty}\sqrt{\frac{1+\rho}{2\pi}} Q^{1+\rho}(x)dx \right\}+o(1) \nonumber
\\      &= \ln E\left\{ \frac{1}{H^{1+\rho}}\sqrt{\frac{1+\rho}{2\pi}}\cdot\frac{aH}{\sigma\sqrt{1+\rho}} \right\}+o(1) \nonumber
\\      &= C - \frac{1}{2}E\left\{ \ln\left(\frac{H^2}{e}\right) \right\} + \ln E\left\{H^{-\rho}\right\}+o(1) \nonumber,
\end{align}
where $o(1)$ denotes a term that vanishes with $a/\sigma$. As a result we get:
\begin{equation}
\label{eq_scalar_E0_uniform}
E_{0,U}(\rho) = \rho C - \frac{\rho}{2}E\left\{\ln\left(\frac{H^2}{e}\right)\right\} - \frac{1+\rho}{2}\ln(1+\rho) - \ln E\left\{H^{-\rho}\right\}+o(1).
\end{equation}
The derivative of $E_{0,U}(\rho)$ w.r.t. $\rho$ gives us the following:
\begin{equation}
\label{eq_scalar_diff_E0_uniform}
\frac{\partial E_{0,U}(\rho)}{\partial \rho} = C -\frac{1}{2}\ln(1+\rho) - \frac{1}{2}E\left\{\ln(H^2)\right\} - \frac{ E\left\{ \ln(H) \cdot H^{-\rho} \right\} }{ E\left\{ H^{-\rho} \right\} }+o(1).
\end{equation}
Hence, for any rate in the range, $R_{\text{cr}}\leq R \leq C$, then
\begin{equation}
\label{eq_scalar_convex_range}
E_r(R) = E_{0,U}(\rho^*) - \rho^*\frac{\partial E_0(\rho)}{\partial \rho}\Big|_{\rho^*},
\end{equation}
and for $0\leq R \leq R_{\text{cr}}$,
\begin{equation}
\label{eq_scalar_linear_range}
E_r(R) = E_{0,U}(1) - R.
\end{equation}
Combining \eqref{eq_scalar_E0_uniform}, \eqref{eq_scalar_diff_E0_uniform}, \eqref{eq_scalar_convex_range} and \eqref{eq_scalar_linear_range} we get the
Gallager's error exponent in the high SNR regime, by the following:
\begin{equation}
\label{eq_scalar_uniform_error_exponent}
E_r(R) \hspace{-0.11cm}=\hspace{-0.11cm}
\left\{ \hspace{-0.25cm} \begin{array}{ll}
C-R-E\left\{\frac{1}{2}\ln\left(\frac{4H^2}{e}\right)\right\}-\ln E\left\{H^{-1}\right\}+o(1),
\hspace{-0.3cm}& 0 \leq R \leq R_{\text{cr}}
\\
\rho^*(C-R)-\frac{\rho^*}{2}E\left\{\ln\left(\frac{H^2}{e}\right)\right\}-\ln E\left\{H^{-\rho^*}\right\}-\frac{1+\rho^*}{2}\ln(1+\rho^*)+o(1),
\hspace{-0.3cm}& R_{\text{cr}}\leq R \leq C
\end{array} \right.
\end{equation}
where $\rho^*$ is given by,
\begin{equation}
R = \frac{\partial E_{0,U}(\rho)}{\partial \rho}\Big|_{\rho^*} = C -\frac{1}{2}\ln(1+\rho^*) - \frac{1}{2}E\left\{\ln(H^2)\right\} - \frac{ E\left\{ \ln(H) \cdot H^{-\rho^*} \right\} }{ E\left\{ H^{-\rho^*} \right\} }+o(1). \nonumber
\end{equation}
Using the relation $R_{\text{cr}} = \frac{\partial E_{0,U}(\rho)}{\partial \rho}\Big|_{\rho^*=1}$, and \eqref{eq_scalar_diff_E0_uniform} we get:
\begin{equation}\
R_{\text{cr}} = \frac{1}{2}\ln\left(\frac{a^2}{4\pi e\sigma^2}\right) + \frac{ E\left\{ \ln(H) \cdot H^{-1} \right\} }{ E\left\{ H^{-1} \right\}}+o(1). \nonumber
\end{equation}
Now let's turn to the approximation of the Gallager's error exponent for rates near the capacity. Since near the capacity $\rho^* \to 0$ by the same Taylor's approximations as we used in Appendix \ref{app_scalar_error_exponent_normal_prior_sec} we obtain:
\begin{align}
\frac{\partial E_{0,U}(\rho)}{\partial \rho}\Big|_{\rho^*} &\approx C - \rho^*\left( \frac{1}{2} + Var\left(\frac{1}{2}\ln\left(H^2\right) \right) \right), \nonumber
\\ \rho^* &\approx \frac{C-R}{\frac{1}{2} + Var\left(\frac{1}{2}\ln\left(H^2\right) \right)} \nonumber
\end{align}
and
\begin{align}
E_r(R) &= \int_{0}^{\rho^*}{\frac{\partial E_{0,U}(\rho)}{\partial \rho}d\rho} - \rho^*R \nonumber
\\   &\approx \frac{(C-R)^2}{2V}, \nonumber
\end{align}
where $V \triangleq  \frac{1}{2} + Var\left(\frac{1}{2}\ln\left(H^2\right) \right) $ and $C = E\left\{\frac{1}{2} \ln\left(\frac{a^2H^2}{2 \pi e \sigma^2}\right) \right\}$.

Note that by the definition of, $\delta = R - \ln(a)$, and by taking the limit $a/\sigma\to\infty$, we can get from $\eqref{eq_scalar_uniform_error_exponent}$, the following error exponent of IC's over fast fading channels:
\begin{equation}
E_r(\delta) =
\left\{ \begin{array}{ll}
\delta^*-\delta-E\left\{\frac{1}{2}\ln\left(\frac{4H^2}{e}\right)\right\}-\ln E\left\{H^{-1}\right\},
& 0 \leq \delta \leq \delta_{\text{cr}}
\\
\rho^*(\delta^*-\delta)-\frac{\rho^*}{2}E\left\{\ln\left(\frac{H^2}{e}\right)\right\}-\ln E\left\{H^{-\rho^*}\right\}-\frac{1+\rho^*}{2}\ln(1+\rho^*),
& \delta_{\text{cr}}\leq \delta \leq \delta^*
\end{array} \right.\nonumber
\end{equation}
where,
\begin{align}
\delta^* &= E\left\{\frac{1}{2} \ln\left(\frac{H^2}{2 \pi e \sigma^2}\right) \right\}, \nonumber
\\
\delta_{\text{cr}} &= \frac{1}{2}\ln\left(\frac{1}{4\pi e\sigma^2}\right) + \frac{ E\left\{ \ln(H) \cdot H^{-1} \right\} }{ E\left\{ H^{-1} \right\}}, \nonumber
\end{align}
and $\rho^* = \rho^*(\delta)$, is given by the solution of
\begin{equation}
\delta = \delta^* -\frac{1}{2}\ln(1+\rho^*) - \frac{1}{2}E\left\{\ln(H^2)\right\} - \frac{ E\left\{ \ln(H) \cdot H^{-\rho^*} \right\} }{ E\left\{ H^{-\rho^*} \right\} }. \nonumber
\end{equation}
In Figure \ref{fig_fading_error_exp_of_IC}, we can see this error exponent, in the case of Rayleigh fading channel with noise variance $\sigma^2=1$. Moreover, it can be seen, that near the Poltyrev's capacity, the error exponent behaves approximately as the parabola $\frac{(\delta^*-\delta)^2}{2V}$.
\begin{figure}[htp]
\center{\includegraphics[width=0.7\columnwidth]{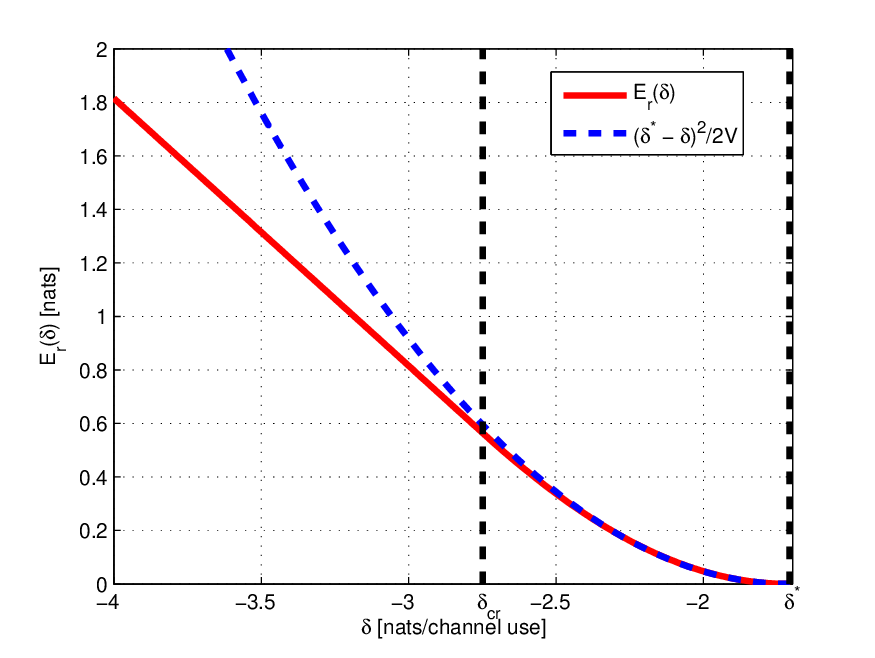}}
\caption{\label{fig_fading_error_exp_of_IC} The error exponent of IC's over the scalar Rayleigh fading channel with noise variance $\sigma^2=1$.}
\end{figure}
\section{Error Exponent for Scalar Complex Fading Channels}
\label{app_scalar_complex_error_exponent_sec}
The scalar complex fading channels are a private case of the MIMO fading channels with one transmit and one receive antennas.
The random coding error exponent for MIMO fading channels, with normal and uniform input distributions, is analyzed in Appendix \ref{sec_MIMO_error_exponents}. Hence, by the assignment of one transmit and one receive antennas in the results of Appendix \ref{sec_MIMO_error_exponents}, namely $r=t=1$, we get the results for this private case.
In addition, the error exponent with the optimal uniform input distribution on a ``thin spherical shell'', for the scalar complex fading channels, can be found in \cite{Ahmed}.
\chapter{Error Exponents for MIMO Fading Channels}
\label{sec_MIMO_error_exponents}
The general formula of Gallager's random coding error exponent for MIMO channels is given by \cite{Gallager}\cite{Telatar}\cite{MIMOErExp}
\begin{equation}
\label{eq_error_exponent_definition}
E_r(R) = \max_{\rho \in [0,1]}\left(\max_{f(\mathbf{x})}E_0(f(\mathbf{x}),\rho) - \rho R\right)
\end{equation}
where,
\begin{equation}
E_0(f(\mathbf{x}),\rho) = -\ln E\left\{ \int{ \left[ \int{f(\mathbf{x})f(\mathbf{y}|\mathbf{x},\mathbf{H})^{\frac{1}{1+\rho}}d\mathbf{x}} \right]^{1+\rho} d\mathbf{y}} \right\}.
\end{equation}
With a slightly abuse of notations we will denote for simplicity $E_0(\rho)$ instead of $E_0(f(\mathbf{x}),\rho)$.

Let us denote by $\rho^* = \rho(R)$, the value of $\rho$ that optimizes \eqref{eq_error_exponent_definition} for a given rate $R$. In addition, let's denote
by $R_{\text{cr}}$, the maximal rate such that $\rho^*$ equals 1.
Then, $\rho^*$ is given by the solution of the following equation
\begin{equation}
\frac{\partial E_0(\rho)}{\partial \rho}\Big|_{\rho^*} = R, \nonumber
\end{equation}
for any rate $R_{\text{cr}} \leq R \leq C$, and by definition:
\begin{equation}
E_r(R) = E_0(\rho^*) - \rho^*R. \nonumber
\end{equation}

In Sections \ref{app_error_exponent_normal_prior_sec} and \ref{app_error_exponent_uni_prior_sec} we will analyze the MIMO random coding error exponent behavior at the high SNR regime and near the capacity, with normal and uniform input distributions, respectively. More results can be found in \cite{Telatar}\cite{MIMOErExp}.
\section{Normal Input Distribution}
\label{app_error_exponent_normal_prior_sec}
In \cite{Telatar} Telatar derived the error exponent of the MIMO channel with the suboptimal capacity-achieving input distribution $\mathbf{x} \sim CN(0,\nicefrac{P}{t} \cdot I_t)$.
In that case the capacity equals $C = E\left\{ \det\left( I_t + \mathbf{H}^{\dagger}\mathbf{H}\cdot SNR \right) \right\}$, and the error exponent factor $E_0(\cdot)$ is given by the following:
\begin{equation}
\label{eq_E_0_G_mimo}
E_{0,G}(\rho) \triangleq E_0(CN(0,\nicefrac{P}{t} \cdot I_t),\rho) = -\ln E\left\{ \det\left( I_t + \mathbf{H}^{\dagger}\mathbf{H}\cdot\frac{SNR}{1+\rho} \right)^{-\rho} \right\},
\end{equation}
where $SNR\triangleq\frac{P/t}{\sigma^2}$. In the high SNR regime we can approximate the capacity by $C \approx E\left\{ \det\left( \mathbf{H}^{\dagger}\mathbf{H}\cdot SNR \right) \right\}$ and \eqref{eq_E_0_G_mimo} by the following:
\begin{align}
E_{0,G}(\rho) &\approx -\ln E\left\{ \det\left( \mathbf{H}^{\dagger}\mathbf{H}\cdot\frac{SNR}{1+\rho} \right)^{-\rho} \right\} \nonumber
\\            &= -\rho  t \ln(1+\rho) -\ln E\left\{ \det\left( \mathbf{H}^{\dagger}\mathbf{H}\cdot SNR \right)^{-\rho} \right\}. \nonumber
\end{align}
The derivative of $E_{0,G}(\rho)$ w.r.t. $\rho$ gives us the following:
\begin{align}
\label{align_diff_E_0_G_1}
\frac{\partial E_{0,G}(\rho)}{\partial \rho} &\approx -t\ln(1+\rho) -\frac{\rho t}{1+\rho} + \frac{ E\left\{ \ln\left( \det\left( \mathbf{H}^{\dagger}\mathbf{H} \cdot SNR \right) \right) \cdot \det\left( \mathbf{H}^{\dagger}\mathbf{H} \right)^{-\rho} \right\} }{ E\left\{ \det\left( \mathbf{H}^{\dagger}\mathbf{H} \right)^{-\rho} \right\} }.
\end{align}
Since near the capacity $\rho^* \to 0$ we can use the following first order Taylor's approximations around zero:
\begin{enumerate}
\item $\ln(1+\rho^*) \approx \rho^*$
\item $\frac{\rho^*}{1+\rho^*} \approx \rho^*$
\item $e^{-\rho^*\ln(\det(\mathbf{H}^{\dagger}\mathbf{H}))} \approx 1 - \rho^*\ln(\det(\mathbf{H}^{\dagger}\mathbf{H}))$
\end{enumerate}
to get the following approximation of \eqref{align_diff_E_0_G_1} near the capacity:
\begin{align}
\label{align_diff_E_0_G_2}
\frac{\partial E_{0,G}(\rho)}{\partial \rho}\Big|_{\rho^*} &\approx -2\rho^* t + t\ln(SNR) - \frac{ \rho^*E\left\{ \ln^2\left( \det\left( \mathbf{H}^{\dagger}\mathbf{H} \right) \right) \right\} - E\left\{ \ln\left( \det\left( \mathbf{H}^{\dagger}\mathbf{H} \right) \right) \right\} }{ 1 - \rho^*E\left\{ \ln\left( \det\left( \mathbf{H}^{\dagger}\mathbf{H} \right) \right) \right\} }.
\end{align}
Using the first order Taylor's approximation of $g(\rho) \triangleq \frac{\rho\cdot a - b}{1-\rho\cdot b}$ around zero we get $g(\rho) \approx -b + \rho\cdot(a - b^2)$. By the assignment of it in \eqref{align_diff_E_0_G_2} we obtain:
\begin{align}
\label{align_diff_E_0_G_3}
\frac{\partial E_{0,G}(\rho)}{\partial \rho}\Big|_{\rho^*} &\approx E\left\{ \ln\left( \det\left( \mathbf{H}^{\dagger}\mathbf{H}\cdot SNR \right) \right) \right\} - \rho^*\left( 2t + Var\left( \ln\left( \det\left( \mathbf{H}^{\dagger}\mathbf{H} \right) \right) \right) \right) \nonumber
\\ &\approx C - \rho^*\left( 2t + Var\left( \ln\left( \det\left( \mathbf{H}^{\dagger}\mathbf{H} \right) \right) \right) \right).
\end{align}
Hence, near the capacity the optimization factor can be approximated by the following:
\begin{equation}
\rho^* \approx \frac{C-R}{2t + Var\left( \ln\left( \det\left( \mathbf{H}^{\dagger}\mathbf{H} \right) \right) \right)}. \nonumber
\end{equation}
By integrating \eqref{align_diff_E_0_G_3} w.r.t. $\rho^*$ and the assignment of $\rho^*$ we obtain:
\begin{align}
E_r(R) &= \int_{0}^{\rho^*}{\frac{\partial E_{0,G}(\rho)}{\partial \rho}d\rho} - \rho^*R \nonumber
\\   &\approx \frac{(C-R)^2}{2V_{UB}}, \nonumber
\end{align}
where $V_{UB} \triangleq 2t + Var\left( \ln\left( \det\left( \mathbf{H}^{\dagger}\mathbf{H} \right) \right) \right)$ and $C = E\left\{ \det\left( \mathbf{H}^{\dagger}\mathbf{H}\cdot SNR \right) \right\}$.

Since the uniform distribution on a ``thin spherical shell'' is the optimal input distribution that maximizes the Gallager's error exponent of the MIMO channel, and not the normal distribution, we got only an upper bound of the channel dispersion from the analysis, $V < V_{UB}$.
\section{Uniform Input Distribution}
\label{app_error_exponent_uni_prior_sec}
Here the input vector is distributed uniformly in $t$ complex dimensional hypercube $\text{Cb}(a,t)$ of size $a$, namely, $f(\mathbf{x})=\frac{1}{a^{2t}} \cdot I_{\{ \mathbf{x} \in \text{Cb}(a,t) \}}$. In the equivalent channel model (using the SVD analysis) $f(\mathbf{y}'|\mathbf{x},\mathbf{H}) = \left(\frac{1}{\pi\sigma^2}\right)^t e^{-\frac{\|\mathbf{y}'-\mathbf{H'}\mathbf{x}\|^2}{\sigma^2}}$, where $\mathbf{H'}\triangleq\mathbf{D'V}^{\dagger}$. With a slightly abuse of notations we will ignore the superscript. Hence,
\begin{align}
E_{0,U}(\rho) &\triangleq E_0(\nicefrac{1}{a^{2t}} \cdot I_{\{ \mathbf{x} \in \text{Cb}(a,t) \}},\rho) \nonumber
\\ &= -\ln E\left\{ \int_{\mathbf{y}\in\mathbb{C}^t}{ \left[ \int_{\mathbf{x}\in\text{Cb}(a,t)}{\frac{1}{a^{2t}} \cdot \left(\frac{1}{\pi\sigma^2}\right)^{\frac{t}{1+\rho}} \cdot e^{-\frac{\|\mathbf{y}-\mathbf{H}\mathbf{x}\|^2}{(1+\rho)\sigma^2}} d\mathbf{x}} \right]^{1+\rho} d\mathbf{y}} \right\} \nonumber
\end{align}
By the variable substitution $\mathbf{x}'=\mathbf{H}\cdot\mathbf{x}$ and some algebraic manipulations we obtain:
\begin{align}
E_{0,U}(\rho) = (1+\rho)C - (1+\rho)E\left\{ \ln\left(\det\left(\frac{\mathbf{H}^{\dagger}\mathbf{H}}{e}\right)\right) \right\} - t(1+\rho)\ln(1+\rho) - I(\rho) \nonumber
\end{align}
where,
\begin{align}
C &\triangleq E\left\{ \ln\left(\det\left(\frac{\mathbf{H}^{\dagger}\mathbf{H}a^2}{\pi e \sigma^2}\right)\right) \right\}, \nonumber
\\
I(\rho) &\triangleq \ln E\left\{ \left(\frac{1}{\det(\mathbf{H}^{\dagger}\mathbf{H})}\right)^{1+\rho} F(\rho,a,\mathbf{H}) \right\} \nonumber
\end{align}
and
\begin{equation}
F(\rho,a,\mathbf{H}) \triangleq \int_{\mathbf{y}\in\mathbb{C}^t}{ \left(\frac{1}{\pi\sigma^2}\right)^t \left[ \int_{\mathbf{x}\in\mathbf{H}\cdot\text{Cb}(a,t)}{ \frac{e^{-\frac{\|\mathbf{y}-\mathbf{x}\|^2}{(1+\rho)\sigma^2}}}{(\pi(1+\rho)\sigma^2)^t} d\mathbf{x}} \right]^{1+\rho} d\mathbf{y}}. \nonumber
\end{equation}
Now we will give a sketch of proof that shows that for large enough $a/\sigma$ (the high SNR regime) $F(\rho,a,\mathbf{H})$ does'nt depend on $\rho$. For doing it let's investigate the derivative of $F(\cdot)$ w.r.t. $\rho$:
\begin{equation}
\frac{\partial{F(\rho,a,\mathbf{H})}}{\partial{\rho}} = \int_{\mathbf{y}\in\mathbb{C}^t}{(1+\rho)\left(\frac{1}{\pi\sigma^2}\right)^tG(\rho,a,\mathbf{y},\mathbf{H})^{\rho}\ln(G(\rho,a,\mathbf{y},\mathbf{H}))
\frac{\partial{G(\rho,a,\mathbf{y},\mathbf{H})}}{\partial{\rho}}}d\mathbf{y} \nonumber
\end{equation}
where,
\begin{align}
G(\rho,a,\mathbf{y},\mathbf{H}) &\triangleq \int_{\mathbf{x}\in\mathbf{H}\cdot\text{Cb}(a,t)}{ \frac{e^{-\frac{\|\mathbf{y}-\mathbf{x}\|^2}{(1+\rho)\sigma^2}}}{(\pi(1+\rho)\sigma^2)^t} d\mathbf{x}} \nonumber
\\
\frac{\partial{G(\rho,a,\mathbf{y},\mathbf{H})}}{\partial{\rho}} &\triangleq \frac{1}{1+\rho} \int_{\mathbf{x}\in\mathbf{H}\cdot\text{Cb}(a,t)}{ \left(\frac{\|\mathbf{y}-\mathbf{x}\|^2}{(1+\rho)\sigma^2} - t\right)\cdot\frac{e^{-\frac{\|\mathbf{y}-\mathbf{x}\|^2}{(1+\rho)\sigma^2}}}{(\pi(1+\rho)\sigma^2)^t} d\mathbf{x}}. \nonumber
\end{align}
By taking the limit when $a\to\infty$ (and for fix $\sigma$), we obtain:
\begin{equation}
G(\rho) \triangleq \lim_{a\to\infty}G(\rho,a,\mathbf{y},\mathbf{H}) = \int_{\mathbf{x}\in\mathbb{C}^t}{ \frac{e^{-\frac{\|\mathbf{y}-\mathbf{x}\|^2}{(1+\rho)\sigma^2}}}{(\pi(1+\rho)\sigma^2)^t} d\mathbf{x}} = 1 \nonumber
\end{equation}
and
\begin{equation}
\frac{\partial{G(\rho)}}{\partial{\rho}} \triangleq \lim_{a\to\infty}\frac{\partial{G(\rho,a,\mathbf{y},\mathbf{H})}}{\partial{\rho}} = \frac{1}{1+\rho} E_{CN(\mathbf{y},(1+\rho)\sigma^2I_t)}\left\{\frac{\|\mathbf{y}-\mathbf{x}\|^2}{(1+\rho)\sigma^2} - t \right\} = 0. \nonumber
\end{equation}
Hence,
\begin{equation}
\frac{\partial{F(\rho,\mathbf{H})}}{\partial{\rho}} \triangleq \lim_{a\to\infty}\frac{\partial{F(\rho,a,\mathbf{H})}}{\partial{\rho}} = 0. \nonumber
\end{equation}
Since $F(\cdot)$ does'nt depend on $\rho$ for large enough $a/\sigma$, we can approximate its value by taking $\rho = 0$ in the high SNR regime:
\begin{align}
F(\rho,a,\mathbf{H}) &\approx \int_{\mathbf{y}\in\mathbb{C}^t}{ \left(\frac{1}{\pi\sigma^2}\right)^t \int_{\mathbf{x}\in\mathbf{H}\cdot\text{Cb}(a,t)}{ \left(\frac{1}{\pi\sigma^2}\right)^t \cdot e^{-\frac{\|\mathbf{y}-\mathbf{x}\|^2}{\sigma^2}} d\mathbf{x}} d\mathbf{y}} \nonumber
\\ &= \int_{\mathbf{x}\in\mathbf{H}\cdot\text{Cb}(a,t)}{ \left(\frac{1}{\pi\sigma^2}\right)^t \int_{\mathbf{y}\in\mathbb{C}^t}{ \left(\frac{1}{\pi\sigma^2}\right)^t \cdot e^{-\frac{\|\mathbf{y}-\mathbf{x}\|^2}{\sigma^2}} d\mathbf{y}} d\mathbf{x}} \nonumber
\\ &= \int_{\mathbf{x}\in\mathbf{H}\cdot\text{Cb}(a,t)}{ \left(\frac{1}{\pi\sigma^2}\right)^t d\mathbf{x}}  = \det(\mathbf{H}^{\dagger}\mathbf{H})\cdot\left(\frac{a^2}{\pi\sigma^2}\right)^t. \nonumber
\end{align}
As a result we get:
\begin{align}
I(\rho) &\approx \ln E\left\{ {\det(\mathbf{H}^{\dagger}\mathbf{H})}^{-\rho} \left(\frac{a^2}{\pi\sigma^2}\right)^t \right\} \nonumber
\\      &= C - E\left\{ \ln\left(\det\left(\frac{\mathbf{H}^{\dagger}\mathbf{H}}{e}\right)\right) \right\} + \ln E\left\{ \left(\frac{1}{\det(\mathbf{H}^{\dagger}\mathbf{H})}\right)^{\rho} \right\} \nonumber
\end{align}
and
\begin{align}
E_{0,U}(\rho) \approx \rho C - \rho E\left\{ \ln\left(\det\left(\frac{\mathbf{H}^{\dagger}\mathbf{H}}{e}\right)\right) \right\} - t(1+\rho)\ln(1+\rho) - \ln E\left\{ {\det(\mathbf{H}^{\dagger}\mathbf{H})}^{-\rho} \right\}. \nonumber
\end{align}
The derivative of $E_{0,U}(\rho)$ w.r.t. $\rho$ gives us the following:
\begin{align}
\label{align_diff_E_0_U_1}
\frac{\partial E_{0,U}(\rho)}{\partial \rho} &\approx C - t\ln(1+\rho) - E\left\{ \ln\left(\det\left({\mathbf{H}^{\dagger}\mathbf{H}}\right)\right) \right\} + \frac{ E\left\{ \ln\left( \det\left( \mathbf{H}^{\dagger}\mathbf{H} \right) \right) \cdot \det\left( \mathbf{H}^{\dagger}\mathbf{H} \right)^{-\rho} \right\} }{ E\left\{ \det\left( \mathbf{H}^{\dagger}\mathbf{H} \right)^{-\rho} \right\} }. \nonumber
\end{align}
Since near the capacity $\rho^* \to 0$ by the same Taylor's approximations as we used in Appendix \ref{app_error_exponent_normal_prior_sec} we obtain:
\begin{align}
\frac{\partial E_{0,U}(\rho)}{\partial \rho}\Big|_{\rho^*} &\approx C - \rho^*\left( t + Var\left( \ln\left( \det\left( \mathbf{H}^{\dagger}\mathbf{H} \right) \right) \right) \right), \nonumber
\\ \rho^* &\approx \frac{C-R}{t + Var\left( \ln\left( \det\left( \mathbf{H}^{\dagger}\mathbf{H} \right) \right) \right)} \nonumber
\end{align}
and
\begin{align}
E_r(R) &= \int_{0}^{\rho^*}{\frac{\partial E_{0,U}(\rho)}{\partial \rho}d\rho} - \rho^*R \nonumber
\\   &\approx \frac{(C-R)^2}{2V}, \nonumber
\end{align}
where $V \triangleq t + Var\left( \ln\left( \det\left( \mathbf{H}^{\dagger}\mathbf{H} \right) \right) \right)$ and $C = E\left\{ \ln\left(\det\left(\frac{a^2\mathbf{H}^{\dagger}\mathbf{H}}{\pi e \sigma^2}\right)\right) \right\}$.


\bibliographystyle{unsrt}

\end{document}